\documentclass[11pt, letterpaper]{article}
\usepackage{fullpage}
\usepackage{amsmath,amsthm,amstext,amssymb,amsfonts,nicefrac,mathtools}
\usepackage{xspace}
\usepackage{color}
\usepackage{url}
\usepackage{hyperref}
\usepackage{bm}
\usepackage{bbm}
\usepackage{times}

\usepackage{enumitem}
\usepackage{ytableau}

\newtheorem{thm}{Theorem}[section]

\newtheorem{lem}[thm]{Lemma}
\newtheorem{cor}[thm]{Corollary}
\newtheorem{claim}[thm]{Claim}
\newtheorem{prop}[thm]{Proposition}
\newtheorem{defn}[thm]{Definition}
\newtheorem{remark}[thm]{Remark}

\newtheorem{fact}[thm]{Fact}

\newcommand\eps{\varepsilon}
\renewcommand\epsilon{\eps}
\newcommand\inner[2]{\langle{#1},{#2}\rangle}
\newcommand\E{\mathop{\mathbb{E}}}
\renewcommand\leq{\leqslant}
\renewcommand\le{\leqslant}
\renewcommand\geq{\geqslant}
\renewcommand\ge{\geqslant}
\newcommand\defeq{\stackrel{def}{=}}
\newcommand\norm[1]{\left\|#1\right\|}
\newcommand\card[1]{\left|#1\right|}
\newcommand\Prob[2]{{\Pr_{#1}\left[ {#2} \right]}}
\newcommand\Expect[2]{{\mathop{\mathbb{E}}_{#1}\left[ {#2} \right]}}
\newcommand\set[1]{{\left\{ #1 \right\}}}
\newcommand\sett[2]{\left\{ \left. #1 \;\right\vert #2 \right\}}
\newcommand*\coef[4]{{#1}_{#2}({#3},{#4})}
\newcommand\power[1]{\set{0,1}^{#1}}

\makeatletter
\newtheorem*{rep@theorem}{\rep@title}
\newcommand{\newreptheorem}[2]{%
\newenvironment{rep#1}[1]{%
\def\rep@title{\bf #2 \ref*{##1} \text{(Restated)} }%
\begin{rep@theorem} }%
{\end{rep@theorem} } }

\newtheorem*{rep@claim}{\rep@title}
\newcommand{\newrepclaim}[2]{%
\newenvironment{rep#1}[1]{%
\def\rep@title{\bf #2 \ref*{##1} \text{(Restated)} }%
\begin{rep@claim} }%
{\end{rep@claim} } }

\newtheorem*{rep@lemma}{\rep@title}
\newcommand{\newreplemma}[2]{%
\newenvironment{rep#1}[1]{%
\def\rep@title{\bf #2 \ref*{##1} \text{(Restated)} }%
\begin{rep@lemma} }%
{\end{rep@lemma} } }

\makeatother
\newreptheorem{theorem}{Theorem}
\newrepclaim{claim}{Claim}
\newreplemma{lemma}{Lemma}

\begin{document}
\title{Hypercontractivity on the symmetric group}
\author{
Yuval Filmus
\thanks{Department of Computer Science, Technion, Israel. This project has received funding from the European Union's Horizon 2020 research and innovation programme under grant agreement No~802020-ERC-HARMONIC.}
\and
Guy Kindler
\thanks{Einstein Institute of Mathematics, Hebrew University of Jerusalem.}
\and
Noam Lifshitz
\thanks{Einstein Institute of Mathematics, Hebrew University of Jerusalem.}
\and
Dor Minzer
\thanks{Department of Mathematics, Massachusetts Institute of Technology.}
}
\date{\vspace{-5ex}}
\maketitle
\begin{abstract}
The hypercontractive inequality is a fundamental result in analysis, with many applications throughout
discrete mathematics, theoretical computer science, combinatorics and more. So far, variants of this inequality
have been proved mainly for product spaces, which raises the question of whether analogous results hold over
non-product domains.

We consider the symmetric group, $S_n$, one of the most basic non-product domains, and establish
hypercontractive inequalities on it. Our inequalities are most effective for the class of \emph{global functions} on $S_n$,
which are functions whose $2$-norm remains small when restricting $O(1)$ coordinates of the input, and assert that low-degree, global functions
have small $q$-norms, for $q>2$.

As applications, we show:
\begin{enumerate}
  \item An analog of the level-$d$ inequality on the hypercube, asserting that the mass of a global function on low-degrees
  is very small. We also show how to use this inequality to bound the size of global, product-free sets in the alternating group $A_n$.
  \item Isoperimetric inequalities on the transposition Cayley graph of $S_n$ for global functions, that are
  analogous to the KKL theorem and to the small-set expansion property in the Boolean hypercube.
  \item Hypercontractive inequalities on the multi-slice, and stability versions of the Kruskal--Katona Theorem
  in some regimes of parameters.
\end{enumerate}
\end{abstract}

\section{Introduction}
The hypercontractive inequality is a fundamental result in analysis that allows one to compare various norms of low-degree functions over a given domain.
A notorious example is the Boolean hypercube $\power{n}$ equipped with the uniform measure, in which case the inequality states that for any function
$f\colon\power{n}\to\mathbb{R}$ of degree at most $d$, one has that $\norm{f}_q\leq \sqrt{q-1}^d\norm{f}_2$ for any $q\geq 2$.
(Here and throughout
the paper, we use expectation norms, $\norm{f}_q = \Expect{x}{\card{f(x)}^q}^{1/q}$, where the input distribution is clear from context, uniform
in this case). While the inequality may appear technical and mysterious at first sight, it has proven itself as remarkably useful, and lies at the heart
of numerous important results, e.g.~\cite{KKL,Friedgut,Bourgainjunta,MOO}.

While the hypercontractive inequality holds for general product spaces, in some important cases it is very weak quantitatively.
Such cases include the $p$-biased cube for $p=o(1)$, the multi-cube $[m]^n$ for $m = \omega(1)$, and the bilinear graph (closely related to the Grassmann graph).
This quantitative deficiency causes various analytical and combinatorial problems on these domains
to be considerably more challenging, and indeed much less is known there (and what is known is considerably more difficult to prove,
see for example~\cite{Friedgutksat}).

\subsection{Global hypercontractivity}
Recently, initially motivated by the study of PCPs (probabilistically checkable proofs) and later by sharp-threshold results,
variants of the hypercontractive inequality have been established in such domains~\cite{KMS2,KLLM,KLLMcodes}. In these variants,
one states an inequality that holds for all functions, but is only meaningful for a special (important) class of functions,
called \emph{global functions}. Informally, a function $f$ on a given product domain $\Omega = \Omega_1\times\dots\times\Omega_n$
is called global, if its $2$-norm, as well as the $2$-norms
of all its restrictions, are all small.\footnote{We remark that this requirement can often be greatly relaxed: (1) it is often enough
to only consider restrictions that fix $O(1)$ of the coordinates of the input, and (2) it is often enough that there are
``very few'' restrictions that have large $2$-norm, for an appropriate notion of ``very few''.} This makes these variants applicable in
cases that were previously out of reach, leading to new results, but at the same time harder to apply, since one has to make sure it is applied to a global function
to get a meaningful bound (see~\cite{KLLM,LM,KLLMcodes} for example applications.). It is worth noting that these variants are in fact
generalizations of the standard hypercontractive inequality, since one can easily show that in domains such as the Boolean hypercube, all
low-degree functions are global.

By now, there are various proofs of the above mentioned results:
(1) a proof by reduction to the Boolean hypercube, (2) a direct proof by expanding $\norm{f}_q^q$ (for even $q$'s),
(3) an inductive proof on $n$.%
\footnote{This inductive proof is actually much trickier than the textbook proof of the hypercontractive inequality over the Boolean cube. The reason is that
the statement of the result itself does not tensorize, thus one has to come up with an alternative, slightly stronger, statement, that does tensorize}
All of these proofs use the product structure of the domain very strongly, and therefore it is unclear how to adapt them beyond the realm of product spaces.

\subsection{Hypercontractivity on non-product spaces}
Significant challenges arise when trying to analyze non-product spaces.
The simplest examples of such spaces are the slice and multi-slice, and the symmetric group.
The classical hypercontractive inequality is equivalent to another inequality, the log-Sobolev inequality. Sharp Log-Sobolev inequalities were proven for the symmetric group and the slice by Lee and Yau~\cite{LeeYau}, and for the multi-slice by Salez~\cite{Salez} (improving on earlier work of Filmus, O'Donnell and Wu~\cite{FOW}).

While such log-Sobolev inequalities are useful for balanced slices and multi-slices, their usefulness for domains such as the symmetric group is limited, due to the similarity between $S_n$ and $[n]^n$. For this reason, Diaconis and Shahshahnai~\cite{DSslice} resorted to representation-theoretic techniques in their analysis of the convergence of Markov chains on $S_n$.
We rectify this issue in a different way, by extending global hypercontractivity to $S_n$.


\subsection{Main results}
The main goal of this paper is to study the symmetric group $S_n$, which is probably the most fundamental non-product domain. Throughout
this paper, we will consider $S_n$ as a probability space equipped with the uniform measure, and use expectation norms, as well as the corresponding expectation inner product, according to the uniform measure.
We will think of $S_n$ as a subset of $[n]^n$, and thereby for $\pi\in S_n$ refer to $\pi(1)$ as ``the first coordinate of the input''.

To state our main results, we begin with defining the notion of globalness on $S_n$.
Given $f\colon S_n\to\mathbb{R}$ and a subset $L\subseteq[n]\times [n]$ of the form $\set{(i_1,j_1),\ldots,(i_t,j_t)}$,
where all of the $i$'s are distinct and all of the $j$'s are distinct, we denote by $S_{n}^T$ the set of permutations
$\pi\in S_n$ respecting $T$ (i.e.\ such that $\pi(i_\ell) = j_\ell$ for all $\ell=1,\ldots,t$), sometimes known as a \emph{double coset} (and corresponding to the notion of link in complexes).
We denote by $f_{\rightarrow T}\colon S_n^T\to\mathbb{R}$
the restriction of $f$ to $S_n^T$, and equip $S_n^T$ with the uniform measure.

\begin{defn}
A function $f\colon S_n\to\mathbb{R}$ is called \emph{$\eps$-global with constant $C$} if for any consistent $T$, it holds that
$\|f_{\rightarrow T}\|_2\leq C^{|T|}\eps$.
\end{defn}
Our basic hypercontractive inequality is concerned with a Markov operator $\mathrm{T}^{(\rho)}$ that may at first not seem very natural.
We defer the precise development and motivation for $\mathrm{T}^{(\rho)}$ to Section~\ref{sec:overview}; for now, we encourage the reader to think of it as averaging
after a long random walk on the transpositions graph, say of length $\Theta(n)$.%
\footnote{Formally, our applications only require that the eigenvalues
corresponding to low-degree functions are bounded away from $0$ (given that $n$ is large enough in comparison the degree of $f$), which
will be the case.}
\begin{thm}\label{thm:hyp_noise_intro}
  For an even $q\in\mathbb{N}$ and $C>0$, there is $\rho>0$ and an operator $\mathrm{T}^{(\rho)}\colon L^2(S_n)\to L^2(S_n)$ satisfying:
  \begin{enumerate}
    \item If $f\colon\power{n}\to\mathbb{R}$ is $\eps$-global with constant $C$, then
    $\norm{\mathrm{T}^{(\rho)} f}_q\leq \eps^{\frac{q-2}{2}}\norm{f}_2^{2/q}$.
    \item There is an absolute constant $K$, such for all $d\in\mathbb{N}$ satisfying $d\leq \sqrt{\log n}/K$, it holds that the
    eigenvalues of $T^{(\rho)}$ corresponding to degree $d$ functions are at least $\rho^{-K\cdot d}$.
  \end{enumerate}
\end{thm}

As is often the case, once one has a hypercontractive inequality involving a noise operator whose eigenvalues are well-understood, one can
state a hypercontractive inequality for low-degree functions. For us, however, it will be important to relax the notion of globalness appropriately,
and we therefore consider the notion of bounded globalness.
\begin{defn}\label{def:bounded_global}
A function $f\colon S_n\to\mathbb{R}$ is called $(d,\eps)$-global if for any consistent $T$ of size at most $d$, it holds that
$\|f_{\rightarrow T}\|_2\leq \eps$.
\end{defn}
A natural example of $(d,\eps)$-global functions is the \emph{low-degree part} of $f$, denoted by $f^{\leq d}$, which is the degree~$d$ function which is closest to $f$ in $L_2$-norm. Here, a function has degree~$d$ if it can be written as a linear combination of indicators of sets $S_n^T$ for $|T| \leq d$.
Naively, one may expect such connection to trivially hold (by Parseval); the issue is that restrictions and degree-truncations
do not commute as well as in product spaces, so such naive arguments fail. Nevertheless, we show that such a connection indeed holds.

With Definition~\ref{def:bounded_global} in hand, we can now state our hypercontractive inequality for low-degree functions.
\begin{thm}\label{thm:Reasonability}
There exists $K>0$ such that the following holds. Let $q\in\mathbb{N}$ be even,
$n\geq q^{K\cdot d^{2}}$.
If $f$ is a $\left(2d,\epsilon\right)$-global function of degree $d$,
then $\|f\|_{q}\le q^{O\left(d^{3}\right)}\epsilon^{\frac{q-2}{q}}\|f\|_{2}^{\frac{2}{q}}$.
\end{thm}

\begin{remark}
  The focus of the current paper is on the case that $n$ is very large in comparison to the degree $d$, and therefore
  the technical conditions imposed on $n$ in Theorems~\ref{thm:hyp_noise_intro} and~\ref{thm:Reasonability} will hold for us.
  It would be interesting to relax or even remove these conditions altogether, and we leave further investigation to future works.
\end{remark}

\subsection{Applications}
We present some applications of Theorem~\ref{thm:hyp_noise_intro} and Theorem~\ref{thm:Reasonability}, as outlined below.
\subsubsection{The level-$d$ inequality}
Our first application is concerned with the weight a global function has on its low degrees, which is an analog of
the classical level-$d$ inequality on the Boolean hypercube (e.g.~\cite[Corollary 9.25]{Od}).
\begin{thm}\label{thm:lvl_d}
There exists an absolute constant $C>0$ such that the following holds.
Let $d,n\in\mathbb{N}$ and $\eps>0$ such that $n\geq 2^{C d^3} \log(1/\eps)^{C d}$.
If $f\colon S_n\to \mathbb{Z}$ is $(2d,\eps)$-global, then
$\norm{f^{\leq d}}_2^2\leq 2^{C \cdot d^4}\eps^4 \log^{C\cdot d}(1/\eps)$.
\end{thm}
   Theorem~\ref{thm:lvl_d} should be compared to the level-$d$ inequality on the hypercube, which asserts that for any function $f\colon\power{n}\to\power{}$ with
  $\E[f] = \delta < 1/2$ we have that $\norm{f^{\leq d}}_2^2\leq \delta^2\left(\frac{10\log(1/\delta)}{d}\right)^{d}$, for all $d\leq \log(1/\delta)$.
  (Quantitatively, the parameter $\delta$ should be compared to $\eps^2$ in Theorem~\ref{thm:lvl_d} due to normalization).

  Note that it may be the case that $\eps$ in Theorem~\ref{thm:lvl_d} is much larger than $\norm{f}_2^{1/2}$, and then
  Theorem~\ref{thm:lvl_d} becomes trivial.\footnote{Parseval's identity implies that the sum of all $\|f^{=d}\|^2$ is $\norm{f}_2^2$,
  so in particular $\norm{f^{\leq d}}_2^2\leq \norm{f}_2^2$.} Fortunately, we can prove a stronger version of Theorem~\ref{thm:lvl_d}
  for functions $f$ whose $2$-norm is not exponentially small, which actually follows relatively easily from Theorem~\ref{thm:lvl_d}.

  \begin{thm}\label{thm:lvl_d_strong}
  There exists an absolute constant $C>0$ such that the following holds.
    Let $d,n\in\mathbb{N}$, $\eps>0$ be parameters and let $f\colon S_n\to\mathbb{Z}$ be
    a $(2d,\eps)$-global function.  If $n\geq 2^{C d^3} \log(1/\norm{f}_2)^{C \cdot d}$, then
    \[
       \norm{f^{\leq d}}_2^2\leq 2^{C \cdot d^4}\norm{f}_2^2\eps^2 \log^{C\cdot d}(1/\norm{f}_2^2).
    \]
\end{thm}

  \paragraph{On the proof of the level-$d$ inequality.}
  In contrast to the case of the hypercube, Theorem~\ref{thm:lvl_d} does not immediately follow from Theorem~\ref{thm:hyp_noise_intro} or Theorem~\ref{thm:Reasonability},
  and requires more work, as we explain below. Recall that one proof of the level-$d$ inequality on the hypercube proceeds, using hypercontractivity, as
  \[
    \norm{f^{\leq d}}_2^2
    = \langle f^{\leq d}, f\rangle
    \leq \norm{f^{\leq d}}_q \norm{f}_{1+1/(q-1)}
    \leq \sqrt{q-1}^d \norm{f^{\leq d}}_2 \norm{f}_{1+1/(q-1)},
  \]
  choosing suitable $q$, and rearranging. Our hypercontractive inequality does not allow us to make the final transition, and instead only tells us
  that $\norm{f^{\leq d}}_q \leq O_{d,q}(\eps^{(q-2)/q}) \norm{f^{\leq d}}_2^{2/q}$. Executing this plan only implies, at best,
  the quantitatively weaker statement that $\norm{f^{\leq d}}_2^2\leq \eps^{3/2} \log^{O_d(1)}(1/\eps)$. Here, the difference between $\eps^{3/2}$ and $\eps^2$ is often
  crucial, because such results are often only useful for very small $\eps$ anyway.

  To explain how we circumvent this issue, note first that the source of the inefficiency is that we used the fact that
  $f^{\leq d}$ is $(2d,\eps)$-global, but the reality could be that it is much more global than that (for example, the statement itself
  asserts a much stronger bound on the $2$-norm of $f^{\leq d}$). To exploit this point, let us consider the restriction that maximizes the $2$-norm of $f^{\leq d}$.
  The most optimistic case would be that the globalness of $f^{\leq d}$
  is achieved already by the function itself, which would say that $f^{\leq d}$ is $(2d,O_d(\norm{f^{\leq d}}_2))$-global.
  In this case, the argument from the hypercube goes through well enough to achieve the desired bound.

  What if the globalness of $f^{\leq d}$ is achieved by a restriction of size $r$ instead? In this case, we show that
  there is a ``derivative''\footnote{We only define the appropriate notion of derivatives we use in Section~\ref{sec:coupling2}, and for now
  encourage the reader to think of it as an analogous operation to the discrete derivative in the Boolean hypercube.}
  $g$ of $f^{\leq d}$ which achieves roughly the same $2$-norm as that restriction of $f^{\leq d}$,
  and taking any further ``derivatives'' only decreases the $2$-norm of $g$.
  We show that this implies that $g$ is $(2d, O_d(\norm{g}_2))$ global, so we have reached the same situation as before!

  The above discussion motivates an inductive approach, and in particular proving the statement for all integer-valued functions (and not only
  Boolean functions), as stated. This
  way, we are able to show that for $g$ above we have that $\norm{g}_2 = \tilde{O}_d(\eps^2)$, which implies that $f^{\leq d}$ is
  $(2d, \tilde{O}_d(\eps^2))$-global. This is a major improvement over our original knowledge regarding $f^{\leq d}$, and in particular
  it allows us to run the argument from the hypercube (described above) successfully.

\subsubsection{Global product-free sets are small}
We say that a family of permutations $F\subseteq S_n$ is product-free if there
are no $\pi_1,\pi_2,\pi_3\in S_n$ such that $\pi_3 = \pi_2\circ \pi_1$.
What is the size of the largest product-free family $F$?

With the formulation above, one can of course take $F$ to be the set of odd permutations, which
has size $\card{S_n}/2$. What happens if we forbid such permutations, i.e.\ only consider families
of even permutations?

Questions of this sort generalize the well-studied problem of finding arithmetic sequences in
dense sets. More relevant to us is the work of Gowers~\cite{GowersQuasi},
which studies this problem for a wide range of groups (referred therein as ``quasi-random groups''), and the
work of Eberhard~\cite{Eberhard} which specialized this question to $A_n$, and improves Gowers' results.
More specifically, Gowers' result shows that a product-free $F\subseteq A_n$ has size at most
$O\left(\frac{1}{n^{1/3}}\card{A_n}\right)$, and Eberhard's work~\cite{Eberhard} improves this bound to
$\card{F} = O\left(\frac{\log^{7/2}n}{\sqrt{n}} \card{A_n}\right)$. We remark that Eberhard's result is
tight up to the polylogarithmic factor, as can be evidenced from the family
\begin{equation}\label{eq:family_prod_free_ex}
 F = \sett{\pi\in A_n}{\pi(1)\in\set{2,\ldots,\sqrt{n}}, \pi(\set{2,\ldots,\sqrt{n}})\subseteq [n]\setminus[\sqrt{n}]}.
\end{equation}

In this section, we consider the problem of determining the maximal size of a \emph{global}, product-free set in $A_n$. In particular, we show:
\begin{thm}\label{thm:global_prod_free}
  There exists $N\in\mathbb{N}$ such that the following holds for all $n\geq N$.
  For every $C>0$ there is $K>0$, such that if $F\subseteq A_n$ is product-free and
  is $(6,C\cdot\sqrt{\delta})$-global, where $\delta = \card{F}/\card{A_n}$, then
  $\delta \leq \frac{\log^K n}{n}$.
\end{thm}
\begin{remark}
  A few remarks are in order.
  \begin{enumerate}
    \item We note that the above result achieves a stronger bound than the family in~\eqref{eq:family_prod_free_ex}.
    There is no contradiction here, of course, since that family is very much not global: restricting to $\pi(1) = 2$
    increases the measure of $F$ significantly.

    \item The junta method, which can be used to study many problems in extremal combinatorics, often considers the question for global families as
    a key component. The rough idea is to show that one can approximate a family $F$ by a union of families $\tilde{F}$ that satisfy an appropriate
    pseudo-randomness condition, such that if $F$ is product-free than so are the families $\tilde{F}$. Furthermore, inside any not-too-small pseudo-random family $\tilde{F}$, one may
    find a global family $\tilde{F}'$ by making any small restriction that increases the size of the family considerably. Thus, in this way one
    may hope to reduce the general question to the question on global families (see~\cite{KLLMcodes} for example).

    While at the moment we do not know how to employ the junta method in the case of product-free sets in $A_n$, one may still hope that it is possible,
    providing some motivation for Theorem~\ref{thm:global_prod_free}.

    \item Our result is in fact more general, and can be used to study the $3$-set version of this problem; see Corollary~\ref{cor:global_prod_free_strong}.

    \item We suspect that much stronger quantitative bounds should hold for global families; we elaborate on this suspicion in Section~\ref{sec:improve_prod_free}.
  \end{enumerate}
\end{remark}

\subsubsection{Isoperimetric inequalities}
Using our hypercontractive inequalities we are able to prove several isoperimetric inequalities for global sets.
Let $S\subseteq S_n$ be a set, and consider the transpositions random walk $\mathrm{T}$ that from a permutation
$\pi\in S_n$ moves to $\pi\circ \tau$, where $\tau$ is a randomly chosen transposition.
We show that if $S$ is ``not too sensitive along any transposition'',\footnote{The formal statement of the result requires an appropriate notion
of discrete derivatives which we only give in Section~\ref{sec:coupling2}.}
then the probability to exit $S$ in a random step according to $\mathrm{T}$ must be significant, similarly to the
classical KKL Theorem on the hypercube~\cite{KKL}.
The formal statement of this result is given in Theorem~\ref{thm:KKL_analog}.

We are also able to analyze longer random walks according to $\mathrm{T}$, of order $\approx n$, and show that
one has small-set expansion for global sets. See Theorem~\ref{thm:SSE} for a formal statement.

\subsubsection{Deducing the results for other non-product domains}
Our results for $S_n$ imply analogous results in the multi-slice. The deduction is done in a black-box fashion,
by a natural correspondence between functions over $S_n$ and over the multi-slice that preserves degrees, globalness, and $L_p$ norms.

This allows us to deduce analogs of our results for $S_n$ essentially for free (see Section~\ref{sec:deduce_other_domains}),
as well as a stability result for the classical Kruskal--Katona Theorem (see Theorem~\ref{thm:KK_stability}).

\subsubsection{Other applications}
Our hypercontractive inequality has also been used in the study of Probabilistically Checkable Proofs~\cite{BKM}. More precisely, to study a new
hardness conjecture, referred to as ``Rich $2$-to-$1$ Games Conjecture'' in~\cite{BKM}, and show that if true, it implies Khot's
Unique-Games Conjecture~\cite{Khot}.

\subsection{Our techniques}\label{sec:overview}
In this section we outline the techniques used in the proofs of Theorem~\ref{thm:hyp_noise_intro} and Theorem~\ref{thm:Reasonability}.

\subsubsection{The coupling approach: proof overview}
\subsubsection*{Obtaining hypercontractive operators via couplings}
Consider two finite probability spaces $X$ and $Y$, and suppose that
$\mathcal{C} = ({\bf x},{\bf y})$ is a coupling between them
(we encourage the reader to think of $X$ as $S_n$, and of $Y$ as a product space in which we
already know hypercontractivity to hold).
Using the coupling $\mathcal{C}$, we may define the averaging operators
$\mathrm{T}_{X\to Y}\colon L^{2}\left(X\right)\to L^{2}\left(Y\right)$ and
$\mathrm{T}_{Y\to X}\colon L^{2}\left(Y\right)\to L^{2}\left(X\right)$ as
\[
\mathrm{T}_{X\to Y}f(y)=\mathbb{E}_{\left(\mathbf{x},\mathbf{y}\right)\sim\mathcal{C}}\left[f\left(\mathbf{x}\right)\mid\mathbf{y}=y\right],
\qquad
\mathrm{T}_{Y\to X}f(x)=\mathbb{E}_{\left(\mathbf{x},\mathbf{y}\right)\sim\mathcal{C}}\left[f\left(\mathbf{y}\right)\mid\mathbf{x}=x\right].
\]

It is easily noted by Jensen's inequality, that each one of the operators $\mathrm{T}_{X\to Y}$ and $\mathrm{T}_{Y\to X}$
is a contraction with respect to the $L^p$-norm, for any $p\geq 1$. The benefit of considering these operators, is that given
an operator $\mathrm{T}_Y$ with desirable properties (say, it is hypercontractive, i.e.\ it satisfies $\|\mathrm{T}_{Y}f\|_{4}\le\|f\|_{2}$),
we may consider the lifted operator on $X$ given by $\mathrm{T}_X \defeq \mathrm{T}_{Y\to X}\mathrm{T}_{Y}\mathrm{T}_{X\to Y}$ and hope
that it too satisfies some desirable properties. Indeed, it is easy to see that if $\mathrm{T}_Y$ is hypercontractive, then
$\mathrm{T}_X$ is also hypercontractive:
\begin{equation}\label{eq1}
\|\mathrm{T}_{Y\to X}\mathrm{T}_{Y}\mathrm{T}_{X\to Y}f\|_{4}\le\|\mathrm{T}_{Y}\mathrm{T}_{X\to Y}f\|_{4}\le\|\mathrm{T}_{X\to Y}f\|_{2}\le\|f\|_{2}.
\end{equation}
We show that the same connection continues to hold for more refined hypercontractive
inequalities such as the one given in~\cite{KLLM,KLLMcodes} (and more concretely, Theorem~\ref{thm:KLLM} below).
We note that the proof in this case is slightly more involved.

While very elegant and appealing, the above approach can only be used to show hypercontractivity for a very special type of operators such as
$\mathrm{T}_X$ defined above, and it is not clear if such results are of any use at all. To remedy this situation, we study
the effect of this operator in the spectral domain. In particular, we show that the action of this operator on
``low-degree functions'' is very similar to the effect of the standard noise operator, and thus we are able deduce a hypercontractive
inequality for low-degree functions, as in Theorem~\ref{thm:Reasonability}.

\subsubsection{Instantiating the coupling approach for the symmetric group}
\subsubsection*{The coupled space}
Let $L = [n]^2$, and let $m$ be large, depending polynomially on $n$ ($m=n^2$ will do).
We will couple $S_n$ and $L^m$, where the idea is to think of each element of $L$ as local information
about the coupled permutation $\pi$. That is, the element $(i,j)\in L$ encodes the fact that $\pi$ maps
$i$ to $j$.

\subsubsection*{Our coupling}
We say
that a set $T=\left\{ \left(i_{1},j_{1}\right),\ldots,\left(i_{t},j_{t}\right)\right\} \subseteq L$
of pairs is \emph{consistent }if there exists a permutation $\pi$
with $\pi\left(i_{k}\right)=j_{k}$ for each $k\in\left[t\right]$,
and any such permutation $\pi$ is said to be consistent with $T$.

Our coupling between $S_n$ and $L^m$ is the following:
\begin{enumerate}
  \item Choose an element $\mathbf{x}\sim L^m$ uniformly at random.
  \item Greedily construct from ${\bf x}$ a set $T$ of consistent pairs. That is,
  starting from $k=1$ to $m$, we consider the $k$-th coordinate of ${\bf x}$, denoted by $(i_k,j_k)$,
  and check whether adding it to $T$ would keep it consistent. If so, we add $(i_k,j_k)$ to $T$, and
  otherwise we do not.
  \item Choose a permutation $\bm{\pi}$ consistent with $T$ uniformly at random.
\end{enumerate}

\subsubsection*{The resulting operator}
Finally, we can specify our hypercontractive operator on $S_{n}$.
Let $X=S_{n}$, $Y=L^{m}$ and $T_{X\to Y},T_{Y\to X}$ be the operators
corresponding to the coupling that we have just constructed.
Let $\mathrm{T}_{Y}=\mathrm{T}_{\rho}$
be the noise operator on the product space $L^{m}$, which can be defined in two equivalent ways:
\begin{enumerate}
\item Every element is retained with probability $\rho$, and resampled otherwise.
\item The $d$'th Fourier level is multiplied by $\rho^d$.
\end{enumerate}
Then $\mathrm{T}^{\left(\rho\right)}=\mathrm{T}_{Y\to X}\mathrm{T}_{Y}\mathrm{T}_{X\to Y}$
is our desired operator on $S_{n}$.

\smallskip

We next explain how to analyze the operator $\mathrm{T}_Y$.

\subsubsection*{Showing that $\mathrm{T}_Y$ satisfies refined hypercontractivity}
Recall the simplistic argument~\eqref{eq1}, showing that hypercontractivity of $\mathrm{T}_X$ implies that the hypercontractivity of $\mathrm{T}_Y$.
We intend to show, in a similar way, that refined hypercontractivity is also carried over by the coupling. Towards this end, we must show that
the notion of globalness is preserved: namely, if $f$ is global, then $g = \mathrm{T}_{S_n\to L^m} f$ is also global. This assertion however
very much depends on the precise notion of globalness we consider.
If we assume that $f$ is $\eps$-global with constant $C$, then it is easy to show
that $g$ is also $\eps$-global with constant $C$ (see Proposition~\ref{prop:n to m}), and the argument goes through smoothly. However, in the case that $f$ is only
guaranteed to be $(d,\eps)$-global, things are more interesting, and in this case we are only able to handle $f$'s that are of low-degree
(this is natural, as we will deal with the low-degree part of $(d,\eps)$-global functions).

A convenient feature of product spaces is that for low-degree functions, the notions of $\eps$-globalness with constant $C$, and $(D,\delta)$-globalness,
are equivalent up to small losses in parameters. This allows one to invoke results such as Theorem~\ref{thm:KLLM} in this case.
While we show that the case of the symmetric group possesses a similar property (at least when $n$ is large enough in comparison to $d$), we are not able to
immediately use it. The issue is that even if $f\colon S_{n}\to\mathbb{R}$
is a function of degree $d$, it may not be the case that $g=\mathrm{T}_{S_{n}\to L_{m}}f$ is also of low degree.

We circumvent this issue as follows. Suppose $f$ is $(2d,\eps)$-global is of degree $d$. Then, as remarked above, we argue that $f$ is $\eps$-global with
some absolute constant $C$, and so it is $(t,C^t\eps)$-global for all $t\in\mathbb{N}$. Thus, $g$ is $(t,C^t\eps)$-global for all $t$. Now, as $g$ is a function
over a product space, it is easily seen that the latter implies that the noisy version of $g$, $h = \mathrm{T}_{\frac{1}{4C}}g$, is $\eps$-global with factor $2$,
and thus we are able to invoke Theorem~\ref{thm:KLLM} on it. Together, this implies that taking
$\mathrm{T}_Y = \mathrm{T}_{1/80^2}\circ\mathrm{T}_{1/(4C)}$ gets us that $\|\mathrm{T}_X f\|_4\leq\sqrt{\eps\|f\|_2}$. (The constant $1/80^2$ arises from Theorem~\ref{thm:KLLM}.)

\subsubsection{The direct approach: proof overview}
Our second approach to establish hypercontractive inequalities goes via a rather different route.
One of the proofs of hypercontractivity in product domains proceeds by finding a convenient, orthonormal basis for the space of real-valued functions over $\Omega$
(which in product cases is easy as the basis tensorizes). This way, proving hypercontractivity amounts to studying moments of this basis functions as well as other forms,
which is often not very hard to do due to the simple nature of the basis.

When dealing with non-product spaces, such as $S_n$, we do not know how to produce such a convenient orthonormal basis. Nevertheless, our direct approach presented in Section~\ref{sec:direct}
relies on a representation of a function $f\colon S_n\to\mathbb{R}$ in a canonical form that is almost as good as in product spaces.
To construct this representation, we start
with obvious spanning sets such as
\[
\sett{\prod\limits_{\ell=1}^d 1_{\pi(i_{\ell}) = j_{\ell}}}{\card{\set{i_1,\ldots,i_d}} = \card{\set{j_1,\ldots,j_d}} = d}.
\]
This set contains many redundancies (and thus is not a basis), and we show how to use these
to enforce a system of linear constraints on the coefficients of the representation that turn out
to be very useful in proving hypercontractive inequalities.

\subsection{Organization of the paper}
In Section~\ref{sec:prelim} we present some basic preliminaries.
Sections~\ref{sec:coupling1},~\ref{sec:coupling2} and~\ref{sec:missing1} are devoted for presenting our approach to hypercontractivity
via coupling and algebraic arguments, and in Section~\ref{sec:direct} we present our direct approach. In Sections~\ref{sec:app} and~\ref{sec:lvl_d}
we present several consequences of our hypercontractive inequalities: the level-$d$ inequality in Section~\ref{sec:lvl_d}, and the other applications in Section~\ref{sec:app}.

\section{Preliminaries}\label{sec:prelim}

We think of the product operation in $S_n$ as function composition, and so $(\tau\sigma)(i) = (\tau\circ\sigma)(i) = \tau(\sigma(i))$.

Throughout the paper, we consider the space of real-valued functions on $S_n$ equipped with the expectation inner product, denoted by $L^{2}\left(S_{n}\right)$.
Namely, for any $f,g\colon S_n\to\mathbb{R}$ we define $\inner{f}{g} = \E_{\sigma\in S_n}[f(\sigma)g(\sigma)]$. A basic property of this space is that it is an
$S_{n}$-bimodule, as can be seen by defining the left operation on a function $f$ and a permutation $\tau$ as $\prescript{\tau}{} f(\sigma) = f(\tau\circ\sigma)$,
and the right operation $f^{\tau}(\sigma) = f(\sigma\circ\tau)$.

\subsection{The level decomposition} \label{sec:level decomposition}

We will define the concept of \emph{degree $d$ function} in several equivalent ways. The most standard definition is the one which we already mentioned in the introduction.

\begin{defn} \label{def:level_space}
Let $T = \{(i_1,j_1),\ldots,(i_t,j_t)\} \subseteq L$ be a set of $t$ consistent pairs, and recall that $S_n^T$ is the set of all permutations such that $\pi(i_k) = j_k$ for all $k \in [t]$.

The space $V_d$	consists of all linear combinations of functions of the form $1_T = 1_{S_n^T}$ for $|T| \leq d$. We say that a real-valued function on $S_n$ has degree (at most) $d$ if it belongs to $V_d$.

\end{defn}
By construction, $V_{d-1} \subseteq V_d$ for all $d \geq 1$. We define the space of functions of pure degree $d$ as
\[
 V_{=d} = V_d \cap V_{d-1}^\perp.
\]

It is easy to see that $V_n = V_{n-1}$, and so we can decompose the space of all real-valued functions on $S_n$ as follows:
\[
 \mathbb{R}[S_n] = V_{=0} \oplus V_{=1} \oplus \cdots \oplus V_{=n-1}.
\]
We comment that the representation theory of $S_n$ refines this decomposition into a finer one, indexed by partitions $\lambda$ of $n$; the space $V_{=d}$ corresponds to partitions in which the largest part is exactly $n-d$.

We may write any function $f\colon S_n\to\mathbb{R}$ in terms of our decomposition uniquely as
$\sum\limits_{i=0}^{n-1}{f^{=i}}$, where $f^{=i}\in V_{=i}$. It will also be convenient for us to have a notation for the projection of
$f$ onto $V_{d}$, which is nothing but $f^{\leq d} = f^{=0} + f^{=1}+\dots+f^{=d}$.

We will need an alternative description of $V_{=d}$ in terms of juntas.

\begin{defn}
Let $A,B \subseteq [n]$. For every $a \in A$ and $b \in B$, let $e_{ab} = 1_{\pi(a)=b}$. We say that a function $f\colon S_n \to \mathbb{R}$ is an $(A,B)$-junta if $f$ can be written as a function of the $e_{ab}$. We denote the space of $(A,B)$-juntas by $V_{A,B}$.

A function is a $d$-junta if it is an $(A,B)$-junta for some $|A|=|B|=d$.
\end{defn}

\begin{lem} \label{lem:V_AB spanning set}
The space $V_{A,B}$ is spanned by the functions $1_T$ for $T \subseteq A \times B$. Consequently, $V_d$ is the span of the $d$-juntas.
\end{lem}
\begin{proof}
If $A = \{i_1,\ldots,i_d\}$ and $B = \{j_1,\ldots,j_d\}$ then an $(A,B)$-junta $f$ can be written as a function of $e_{i_sj_t}$, and in particular as a polynomial in these functions. Since $e_{i_sj_{t_1}} e_{i_sj_{t_2}} = e_{i_{s_1}j_t} e_{i_{s_2}j_t} = 0$ if $t_1 \neq t_2$ and $s_1 \neq s_2$, it follows that $f$ can be written as a linear combination of functions $1_T$ for $T \subseteq A \times B$.

Conversely, if $T = \{(a_1,b_1),\ldots,(a_d,b_d)\}$ then $1_T = e_{a_1b_1} \cdots e_{a_db_d}$.

To see the truth of the second part of the lemma, notice that if $|A|=|B|=d$ and $T \subseteq A \times B$ then $|T| \leq d$, and conversely if $|T| \leq d$ then $T \subseteq A \times B$ for some $A,B$ such that $|A|=|B|=d$.
\end{proof}

We will also need an alternative description of $V_{A,B}$.

\begin{lem} \label{lem:V_AB bimodule}
For each $A,B$, the space $V_{A,B}$ consists of all functions $f\colon S_n \to \mathbb{R}$ such that $f = \prescript{\tau}{}f^\sigma$ for all $\sigma$ fixing $A$ pointwise and $\tau$ fixing $B$ pointwise.
\end{lem}
\begin{proof}
Let $U_{A,B}$ consist of all functions $f$ satisfying the stated condition, i.e., $f(\pi) = f(\tau \pi \sigma)$ whenever $\sigma$ fixes $A$ pointwise and $\tau$ fixes $B$ pointwise.

Let $a \in A$ and $b \in B$. If $\sigma$ fixes $a$ and $\tau$ fixes $b$ then $\pi(a) = b$ iff $\tau \pi \sigma (a) = b$, showing that $e_{ab} \in U_{A,B}$. It follows that $V_{A,B} \subseteq U_{A,B}$.

In the other direction, let $f \in U_{A,B}$. Suppose for definiteness that $A = [a]$ and $B = [b]$. Let $\pi$ be a permutation such that $\pi(1) = 1, \ldots, \pi(t) = t$, and $\pi(i) > b$ for $i=t+1,\ldots,a$. Applying a permutation fixing $B$ pointwise on the left, we turn $\pi$ into a permutation $\pi'$ such that $\pi'(1),\ldots,\pi'(a) = 1,\ldots,t,b+1,\ldots,b+(a-t)$. Applying  a permutation fixing $A$ pointwise on the right, we turn $\pi'$ into the permutation $1,\ldots,t,b+1,\ldots,b+(a-t),\ldots,n,t+1,\ldots,a$. This shows that if $\pi_1,\pi_2$ are two permutations satisfying $e_{ab}(\pi_1) = e_{ab}(\pi_2)$ for all $a \in A,b \in B$ then we can find permutations $\sigma_1,\sigma_2$ fixing $A$ pointwise and permutations $\tau_1,\tau_2$ fixing $B$ pointwise such that $\tau_1 \pi_1 \sigma_1 = \tau_2 \pi_2 \sigma_2$, and so $f(\pi_1) = f(\pi_2)$. This shows that $f \in V_{A,B}$.
\end{proof}



\subsection{Hypercontractivity in product spaces}
We will make use of the following hypercontractive inequality, essentially due to~\cite{KLLMcodes}.
For that, we first remark that we consider the natural analog definitions of globalness for product spaces.
Namely, for a finite product space $(\Omega,\mu) = (\Omega_1\times\dots\times\Omega_m,\mu_1\times\dots\times\mu_m)$,
we say that $f\colon\Omega\to\mathbb{R}$ is $\eps$-global with a constant $C$, if for any $T\subseteq [m]$ and $x\in\prod_{i\in T}\Omega_i$
it holds that $\|f_{T\rightarrow x}\|_{2,\mu_x}^2\leq C^{|T|}\eps$, where $\mu_x$ is the distribution $\mu$ conditioned on coordinates of
$T$ being equal to $x$. Similarly, we say that $f$ is $(d,\eps)$-global if for any $|T|\leq d$
and $x\in\prod_{i\in T}\Omega_i$ it holds that $\|f_{T\rightarrow x}\|_{2,\mu_x}^2\leq \eps$.

\begin{thm} \label{thm:KLLM}
Let $q\in\mathbb{N}$ be even, and suppose $f$ is $\eps$-global with constant $C$, and let $\rho\le\frac{1}{(10qC)^2}$. Then
$\|\mathrm{T}_{\rho}f\|_{q}\le\epsilon^{q-2}\|f\|_{2}^{\frac{2}{q}}$.
\end{thm}
We remark that Theorem~\ref{thm:KLLM} was proved in~\cite{KLLMcodes} for $q=4$, however the proof is essentially the same for all even integers $q$.

\section{Hypercontractivity: the coupling approach}\label{sec:coupling1}

\subsection{Hypercontractivity from full globalness}
In this section we prove the following hypercontractive results for
our operator $\mathrm{T}^{\left(\rho\right)}$ assuming $f$ is global. We begin
by proving two simple propositions.
\begin{prop}
\label{prop:n to m}
Suppose $f\colon S_n\to\mathbb{R}$ is $\eps$-global with constant $C$, and let $g=\mathrm{T}_{S_{n}\to L^{m}}f$.
Then $g$ is $\eps$-global with constant $C$.
\end{prop}
\begin{proof}
Let $S$ be a set of size $t$, and
let $x=\bigl((i_{k},j_{k})\bigr)_{k\in S}\in L^{S}$.
Let $y\sim L^{\left[m\right]\setminus S}$ be chosen uniformly, and
let $\sigma$ be the random permutation that our coupling process
outputs given $\left(x,y\right)$. We have
\[
\|g_{S\to x}\|_{2}^{2}=\mathbb{E}_{y}\left(\mathbb{E}_{\sigma}f\left(\sigma\right)\right)^{2}\le\mathbb{E}_{\sigma}\left[f\left(\sigma\right)^{2}\right]
\]
by Cauchy--Schwarz. Next, we consider the values of $\sigma\left(i_{k}\right)$ for $k\in S$, condition on
them and denote $T=\left\{ \left(i_{k},\sigma(i_k)\right)\right\} $.
The conditional distribution of $\sigma$ given $T$ is uniform by the symmetry of elements in $\left[n\right]\setminus\left\{ i_{k}|\,k\in S\right\}$,
so for any permutation $\pi$ on $\left[n\right]\setminus\left\{ i_{k}|\,k\in S\right\}$ we have that $\sigma\pi$ has the same probability as $\sigma$.
Also, the collection $\{\sigma\pi\}$ consists of all permutations satisfying $T$, so
\[
\E\left[f\left(\sigma\right)^{2}\right]=
\mathbb{E}_{T}\left[\mathrm{\|}f_{T}\|_{2}^{2}\right]\le\max_{T}\|f_{T}\|_{2}^{2}
\leq C^{2|S|}\eps^2.\qedhere
\]
\end{proof}

\begin{fact}
\label{fact:contraction}Suppose that we are given two probability
spaces $\left(X,\mu_{X}\right),\left(Y,\mu_{Y}\right)$. Suppose further
that for each $x\in X$ we have a distribution $N\left(x\right)$
on $Y$, such that if we choose $x\sim\mu_{X}$ and $y\sim N\left(x\right)$,
then the marginal distribution of $y$ is $\mu_{Y}$. Define an operator
$\mathrm{T}_{Y\to X}\colon L^{2}\left(Y\right)\to L^{2}\left(X\right)$
by setting
\[
\mathrm{T}_{Y\to X}f\left(x\right)=\mathbb{E}_{y\sim N\left(x\right)}f\left(y\right).
\]
 Then $\|\mathrm{T}_{Y\to X}f\|_{q}\le\|f\|_{q}$ for each $q\ge1$.
\end{fact}

We can now prove one variant of our hypercontractive inequality for global functions over the symmetric group.
\begin{thm}\label{thm:Hypercontractivity}
Let $q\in\mathbb{N}$ be even, $C,\eps>0$, and $\rho\le\frac{1}{(10qC)^2}$.
If $f\colon S_{n}\to\mathbb{R}$ is $\eps$-global with constant $C$, then
$\norm{\mathrm{T}^{\left(\rho\right)}f}_{q}\le\epsilon^{\frac{q-2}{q}}\|f\|_{2}^{\frac{2}{q}}$.
\end{thm}
\begin{proof}
 Let $f\colon S_n\to\mathbb{R}$ be $\eps$-global with constant $C$.
 By Proposition \ref{prop:n to m}, the function $g=\mathrm{T}_{S_{n}\to L^{m}}f$ is also $\eps$-constant with constant $C$, and
 by Fact~\ref{fact:contraction} we have
 \[
 \norm{\mathrm{T}^{\left(\rho\right)}f}_{q}^{q} = \norm{\mathrm{T}_{L^{m}\to S_{n}}\mathrm{T}_{\rho}g}_{q}^{q}\le
 \norm{\mathrm{T}_{\rho}g}_{q}^{q}.
 \]
 Now, by Theorem \ref{thm:KLLM} we may upper-bound the last norm by $\epsilon^{q-2}\|g\|_{2}^{2}$, and using Fact~\ref{fact:contraction}
 again we may bound $\norm{g}_2^2\leq \norm{f}_2^2$.
\end{proof}
\begin{remark}
Once the statement has been proven for even $q$'s, a qualitatively similar statement can be automatically deduced
for all $q$'s, as follows.
Fix $q$, and take the smallest $q\leq q'\leq q+2$ that is an even integer. Then for
$\rho\le\frac{1}{(10(q+2)C)^2}\leq \frac{1}{(10q'C)^2}$ we may bound
\[
\norm{\mathrm{T}^{\left(\rho\right)}f}_{q}
\leq \norm{\mathrm{T}^{\left(\rho\right)}f}_{q'}
\leq \epsilon^{\frac{q'-2}{q'}}\|f\|_{2}^{\frac{2}{q'}}
\leq \epsilon^{\frac{q}{q+2}}\|f\|_{2}^{\frac{2}{q+2}},
\]
where in the last inequality we used $q'\leq q+2$ and $\norm{f}_2\leq \eps$.
\end{remark}

\subsection{Hypercontractivity for low-degree functions}
Next, we use Theorem \ref{thm:Hypercontractivity} to prove
our hypercontractive inequality
for low-degree functions that assumes considerably weaker globalness properties of $f$, namely Theorem~\ref{thm:Reasonability}.
The proof of the above theorem makes use of the following key lemmas. The first of which asserts that just like in the cube,
bounded globalness of a low-degree function implies (full) globalness.

\begin{lem}\label{lem:Bootstrapping the globalness}
Suppose $n\ge Cd\log d$ for
a sufficiently large constant $C$. Let $f\colon S_n\to\mathbb{R}$ be a $\left(2d,\epsilon\right)$-global
function of degree $d$. Then, $f$ is $\eps$-global with constant $4^8$.
\end{lem}
Thus, to deduce Theorem~\ref{thm:Reasonability} from Theorem~\ref{thm:Hypercontractivity}, it suffices to show that
$f$ may be approximated by linear combinations of $\mathrm{T}^{(\rho^i)} f$ for $i=1,2,\ldots$ in $L^q$, and this is the content
of our second lemma. First, let us introduce some convenient notations.
For a polynomial $P(z)=a_{0}+a_{1}z+\cdots+a_{k}z^k$, we denote the spectral norm of $P$ by
$\|P\|=\sum_{i=0}^{k}\left|a_{i}\right|$. We remark that it is easily seen that $\|P_{1}P_{2}\|\le\|P_{1}\|\|P_{2}\|$
for any two polynomials $P_1,P_2$.

\begin{lem}
\label{lem:Approximation} Let $n\ge C^{d^{3}}q^{-Cd}$ for a sufficiently
large constant $C$, and let $\rho=1/(400 C^3q^2)$. Then there exists a polynomial
$P$ satisfying $P\left(0\right)=0$ and $\|P\|\le q^{O\left(d^{3}\right)}$,
such that
\[
\norm{P\left(\mathrm{T}^{\left(\rho\right)}\right)f-f}_{q}\le\frac{1}{\sqrt{n}}\norm{f}_{2}
\]
 for every function $f$ of degree at most $d$.
\end{lem}
We defer the proofs of Lemmas~\ref{lem:Bootstrapping the globalness}  and~\ref{lem:Approximation} to Sections~\ref{sec:coupling2} and~\ref{sec:missing1}, respectively.
In the remainder of this section we derive Theorem~\ref{thm:Reasonability} from them, restated below.

\begin{reptheorem}{thm:Reasonability}
There exists $C>0$ such that the following holds. Let $q\in\mathbb{N}$ be even,
$n\geq q^{C\cdot d^{2}}$.
If $f$ is a $\left(2d,\epsilon\right)$-global function of degree $d$,
then $\|f\|_{q}\le q^{O\left(d^{3}\right)}\epsilon^{\frac{q-2}{q}}\|f\|_{2}^{\frac{2}{q}}$.
\end{reptheorem}
\begin{proof}
Choose $\rho = 1/(400 C^3 q^2)$, and let $P$ be as in Lemma \ref{lem:Approximation}. Then
\[
\norm{f}_{q}\le\norm{P\left(\mathrm{T}^{\left(\rho\right)}\right)f}_{q}+\frac{1}{\sqrt{n}}\norm{f}_{2}.
\]
 As for the first term, we have
\[
\norm{\sum_{i=1}^{l}a_{i}\left(\mathrm{T}^{\left(\rho\right)}\right)^{i}f}_{q}
\le
\sum_{i=1}^{l}\card{a_{i}}\norm{\left(\mathrm{T}^{\left(\rho\right)}\right)^{i}f}_{q}
\leq
\norm{P}\norm{\mathrm{T}^{\left(\rho\right)}f}_{q}
\leq
q^{O\left(d^{3}\right)}
\norm{\mathrm{T}^{\left(\rho\right)}f}_{q}.
\]
To estimate $\norm{\mathrm{T}^{\left(\rho\right)}f}_{q}$, note first that by Lemma \ref{lem:Bootstrapping the globalness},
$f$ is $\eps$-global for constant $4^8$, thus given that $C$ is large enough we may apply Theorem \ref{thm:Hypercontractivity}
to deduce that $\norm{\mathrm{T}^{\left(\rho\right)}f}_{q}\leq \epsilon^{\frac{q-2}{q}}\|f\|_{2}^{\frac{2}{q}}$.
As $\|f\|_{2}\le\epsilon$ we conclude that
\[
\|f\|_{q}\le q^{O\left(d^{3}\right)}\epsilon^{\frac{q-2}{q}}\|f\|_{2}^{\frac{2}{q}}+\frac{1}{\sqrt{n}}\|f\|_{2}=q^{O\left(d^{3}\right)}\epsilon^{\frac{q-2}{q}}\|f\|_{2}^{\frac{2}{q}}.
\qedhere
\]
\end{proof}

\section{Proof of Lemma~\ref{lem:Bootstrapping the globalness}}\label{sec:coupling2}
We begin by proving Lemma~\ref{lem:Bootstrapping the globalness}. A proof of the corresponding statement in product spaces
proceeds by showing that a function is $(d,\eps)$-global if and only if the $2$-norms of derivatives of $f$ of order $d$
are small. Since then derivatives of order higher than $d$ of $f$ are automatically $0$ (by degree considerations), they are
automatically small. Thus, if $f$ is a $(d,\eps)$-global function of degree $d$, then \emph{all derivatives} of $f$ have
small $2$-norm, and by the reverse relation it follows that $f$ is $\eps$-global for some constant $C$.

Our proof follows a similar high level idea. The main challenge in the proof is to find an appropriate analog of discrete derivatives from product
spaces, that both reduces the degree of the function $f$ and can be related to restrictions of $f$. Towards this end, we make the following key definition.
\begin{defn}\label{def:derivative}
Let $i_{1}\ne i_{2}\in\left[n\right]$ and $j_{1}\ne j_{2}\in\left[n\right]$.
\begin{enumerate}
  \item The Laplacian of $f$ along $(i_1,i_2)$ is defined as
$\mathrm{L}_{\left(i_{1},i_{2}\right)}\left[f\right]=f-f^{\left(i_{1}\,i_{2}\right)}$, where
we denote by $\left(i_{1}\,i_{2}\right)$ the transposition of $i_1$ and $i_2$.
  \item The derivative of $f$ along $(i_1,i_2)\rightarrow(j_1,j_2)$ is
  $(\mathrm{L}_{\left(i_{1},i_{2}\right)}f)_{(i_1,i_2)\rightarrow (j_1,j_2)}$. More explicity, it is a function defined on
  $S_n^{{(i_1,j_1),(i_2,j_2)}}$ (that is isomorphic to $S_{n-2}$) whose value on $\pi$ is
  \[
  f(\pi) - f(\pi\circ (i_1,i_2)).
  \]
  \item For distinct $i_1,\ldots,i_t$ and distinct $j_1,\ldots,j_t$, denote the ordered set
  $S=\left\{ \left(i_{1},j_{1}\right),\ldots,\left(i_{t},j_{t}\right)\right\}$ and
  define the Laplacian of $f$ along $S$ as $L_{S}\left[f\right]=L_{i_{1},j_{1}}\circ\cdots\circ L_{i_{t},j_{t}}\circ f$.

  For $(k_1,\ell_1),\ldots,(k_t,\ell_t)$, the derivative of $f$ along $S\rightarrow \{(k_1,\ell_1),\ldots,(k_t,\ell_t)\}$ is
  \[
   \mathrm{D}_{S\to\left\{ \left(k_{1},l_{1}\right),\ldots,\left(k_{t},l_{t}\right)\right\} } f=
   \left(L_{S}\left[f\right]\right)_{S\rightarrow \left\{ \left(i_{1},k_{1}\right),\left(j_{1},l_{1}\right),\ldots,\left(i_{t},k_{t}\right),\left(j_{t},l_{t}\right)\right\} }
  \]

  We call $\mathrm{D}$ a derivative of order $t$. We also include the case where $t=0$, and call the identity operator a $0$-derivative.
\end{enumerate}
\end{defn}

The following two claims show that the definition of derivatives above is good, in the sense that $2$-norms of derivatives relate to globalness,
and derivatives indeed reduce the degree of $f$.
\begin{claim}
\label{claim:globalness equivalence}
Let $t\in\mathbb{N}$, and $\eps>0$, and $f\colon S_{n}\to\mathbb{R}$.
\begin{enumerate}
  \item If $f$ is $(2t,\eps)$-global, then for each derivative $\mathrm{D}$ of order $t$ we have that
  $\norm{\mathrm{D} f}_2\leq 2^{t} \eps$.
  \item If $t\leq n/2$, and for all $\ell\leq t$ and every derivative $\mathrm{D}$ of order $\ell$ we have that
  $\norm{\mathrm{D} f}_2\leq \eps$, then $f$ is $(t,2^t \eps)$-global.
\end{enumerate}
\end{claim}
\begin{proof}
The first item follows immediately by induction on $t$ using the triangle inequality.
The rest of the proof is devoted to establishing the second item, also by induction on $t$.
\paragraph{Base case $t=0,1$.}
The case $t=0$ is trivial, and we prove the case $t=1$.
Let $i_{1},i_{2}\in\left[n\right]$
be distinct and let $j_{1},j_{2}\in\left[n\right]$ be distinct. Since
$\|\mathrm{D}_{\left(i_{1},i_{2}\right)\to\left(j_{1},j_{2}\right)}f\|_2\le\epsilon$
we get from the triangle inequality that
\begin{equation}
\left|\|f_{i_{1}\to j{}_{1},i_{2}\to j_{2}}\|_{2}-\|f_{i_{2}\to j_{1},i_{1}\to j_{2}}\|_{2}\right|\le\epsilon.\label{eq:tria}
\end{equation}
Multiplying (\ref{eq:tria}) by $\|f_{i_{1}\to j{}_{1},i_{2}\to j_{2}}\|_{2}+\|f_{i_{2}\to j_{1},i_{1}\to j_{2}}\|_{2}$ we get that
\[
\left|\|f_{i_{1}\to j{}_{1},i_{2}\to j_{2}}\|_{2}^{2}-\|f_{i_{2}\to j_{1},i_{1}\to j_{2}}\|_{2}^{2}\right|  \le
\epsilon\left(\|f_{i_{1}\to j{}_{1},i_{2}\to j_{2}}\|_{2}+\|f_{i_{2}\to j_{1},i_{1}\to j_{2}}\|_{2}\right).
\]
Taking average over $j_{2}$ and using the triangle inequality on the left-hand side, we get that
\[
\left|\|f_{i_{1}\to j_{1}}\|_{2}^{2}-\|f_{i_{2}\to j_{1}}\|_{2}^{2}\right|\le\epsilon
\mathbb{E}_{j_{2}}\left[\|f_{i_{1}\to j{}_{1},i_{2}\to j_{2}}\|_{2}+\|f_{i_{2}\to j{}_{1},i_{1}\to j_{2}}\|_{2}\right].
\]
 By Cauchy--Schwarz,
 $\mathbb{E}_{j_{2}}\left[\|f_{i_{1}\to j{}_{1},i_{2}\to j_{2}}\|_{2}\right]\leq \mathbb{E}_{j_{2}}\left[\|f_{i_{1}\to j{}_{1},i_{2}\to j_{2}}\|_{2}^2\right]^{1/2}
 =\|f_{i_{1}\to j{}_{1}}\|_{2}$, and similarly for the other term, so we conclude
\[
\left|\|f_{i_{1}\to j_{1}}\|_{2}^{2}-\|f_{i_{2}\to j_{1}}\|_{2}^{2}\right|\le\epsilon\left(\|f_{i_{1}\to j_{1}}\|_{2}+\|f_{i_{2}\to j_{1}}\|_{2}\right),
\]
 and dividing both sides of the inequality by $\|f_{i_{1}\to j_{1}}\|_{2}+\|f_{i_{2}\to j_{1}}\|_{2}$ we get
\[
\card{\|f_{i_{1}\to j_{1}}\|_{2}-\|f_{i_{2}\to j_{1}}\|_{2}}\le\epsilon.
\]
 Since $\mathbb{E}_{i_{2}\sim\left[n\right]}\|f_{i_{2}\to j_{1}}\|_{2}^{2}=\|f\|_{2}^{2}\leq \eps$, we get that
 there is $i_2$ such that $\|f_{i_{2}\to j_{1}}\|_{2}\le\epsilon$, and the above inequality implies that
 $\|f_{i_{1}\to j_{1}}\|_{2}\le2\epsilon$ for all $i_1$. This completes
the proof for the case $t=1$.

\paragraph{The inductive step.}
Let $t>1$. We prove that $f$ is $\left(t,2^{t}\epsilon\right)$-global, or
equivalently that $f_{T}$ is $\left(1,2^{t}\epsilon\right)$-global
for all consistent sets $T$ of size $t-1$. Indeed, fix a consistent $T$ of size $t-1$.

By the induction hypothesis, $\|f_{T}\|_{2}\le2^{t-1}\epsilon$,
and the claim would follow from the $t=1$ case once we show that $\|\mathrm{D}f_{\rightarrow T}\|_{2}\le2^{t-1}\epsilon$
for all order $1$ derivatives $\mathrm{D}=\mathrm{D}_{\left(i_{1},i_{2}\right)\to\left(j_{1},j_{2}\right)}$,
where $i_{1},i_{2}$ do not appear as the first coordinate of an element in $T$,
and $j_{1},j_{2}$ do not appear as a second coordinate of an element
of $T$ (we're using the fact here that the case $t=1$ applies, as $S_{n}^{T}$ is isomorphic
to $S_{n-\left|T\right|}$ as $S_{n-\left|T\right|}$-bimodules).
Fix such $\mathrm{D}$, and let $g=\mathrm{D}_{\left(i_{1},i_{2}\right)\to\left(j_{1},j_{2}\right)}f$.
By hypothesis, for any order $t-1$ derivative $\tilde{\mathrm{D}}$ we have that $\|\tilde{\mathrm{D}}g\|_{2}\le\epsilon$, hence by the
induction hypothesis $\|g{}_{\rightarrow T}\|_{2}\leq 2^{t-1}\eps$. Since restrictions and derivatives commute, we have
$g{}_{\rightarrow T} = \mathrm{D}_{\left(i_{1},i_{2}\right)\to\left(j_{1},j_{2}\right)}f_{\rightarrow T}$, and we conclude that
$f_{\rightarrow T}$ is $\left(1,2^{t}\epsilon\right)$-global,
as desired.
\end{proof}

\begin{claim}
\label{claim:degree reduction} If $f$ is of degree $d$, and $\mathrm{D}$
is a $t$-derivative, then $\mathrm{D}f$ is of degree $\le d-t$.
\end{claim}
\begin{proof}
It is sufficient to consider the case $t=1$ of the proposition, as
we may apply it repeatedly. By linearity of the derivative $\mathrm{D}$
it is enough to show it in the case where $f=x_{i_{1}\to j_{1}}\cdots x_{i_{t}\to j_{t}}$.
Now note that the Laplacian $L_{(k_{1}k_{2})}$ annihilates $f$ unless
either $k_{1}$ is equal to some $i_{\ell}$, or $k_{2}$ is equal to
some $i_{\ell}$, or both, and we only have to consider these cases.
Each derivative corresponding to the Laplacian $L_{(k_1,k_2)}$
restricts both the image of $k_{1}$ and the image of $k_{2}$,
so after applying this restriction on $L_{(k_1,k_2)} f$ we either get
the $0$ function, a function of degree $d-1$, or a function of degree $d-2$.
\end{proof}

We are now ready to prove Lemma~\ref{lem:Bootstrapping the globalness}. To prove that $f$ is global, we handle restrictions
of size $t\leq n/2$, and restrictions of size $t>n/2$ separately, in the following two claims.
\begin{claim}
\label{claim:globalness up to n/2}
Suppose $f\colon S_n\to\mathbb{R}$ is a $\left(2d,\epsilon\right)$-global
function of degree $d$. Then $f$ is $\left(t,4^{t}\epsilon\right)$-global
for each $t\le\frac{n}{2}$.
\end{claim}
\begin{proof}
By the second item in Claim \ref{claim:globalness equivalence}, it is enough to show that for
each $t$-derivative $\mathrm{D}$ we have $\|\mathrm{D}f\|_{2}\le2^{t}\epsilon$.
For $t\le d$ this follows from the first item in Claim~\ref{claim:globalness equivalence},
and for $t>d$ it follows from Proposition \ref{claim:degree reduction} as we have that
$\mathrm{D} f = 0$ for all derivatives of order $t$.
\end{proof}

For $t\ge\frac{n}{2}$, we use the obvious fact $f$ is always
$\left(t,\|f\|_{\infty}\right)$-global, and upper bound the infinity norm of $f$
using the following claim.
\begin{claim}
\label{claim:upper bound on the infty norm} Let $f$ be a $\left(2d,\epsilon\right)$-global
function of degree $d$. Then $\|f\|_{\infty}\le\sqrt{\left(6d\right)!}4^{3n}\epsilon$.
\end{claim}

\begin{proof}
We prove the claim by induction on $n$. The case $n=1$ is obvious, so let $n>1$.

If $3d\le\frac{n}{2}$, then by Claim~\ref{claim:globalness up to n/2} we have that $f$
is $\left(3d,4^{3d}\epsilon\right)$-global, and hence for each set $S$ of size $d$,
the function $f_{\rightarrow S}$ is $\left(2d,4^{3d}\epsilon\right)$-global. Therefore,
the induction hypothesis implies that
\[
\|f\|_{\infty}=\max_{S:\,\left|S\right|=d}\|f_{S}\|_{\infty}\le\sqrt{\left(6d\right)!}4^{3\left(n-d\right)}\cdot4^{3d}\epsilon=\sqrt{\left(6d\right)!}4^{3n}\epsilon.
\]
 Suppose now that $n\le6d$. Then $\|f\|_{\infty}^{2}\le\left(6d\right)!\|f\|_{2}^{2}$
since the probability of each atom in $S_{6d}$ is $\frac{1}{\left(6d\right)!}$.
Hence, $\|f\|_{\infty}\le\sqrt{\left(6d\right)!}\epsilon$.
\end{proof}
Note that $(6d)!\leq 4^n$ given $C$ is sufficiently large, so for $t>n/2$, Claim~\ref{claim:upper bound on the infty norm} implies that
$f$ is $(t, 4^{4n}\eps) = (t, 4^{8t}\eps)$-global.\qed


\section{Proof of Lemma~\ref{lem:Approximation}}\label{sec:missing1}
\paragraph{Proof overview.}
Our argument first constructs a very strong approximating polynomial in the $L_2$-norm.
The approximation will be in fact strong enough to imply, in a black-box way, that it is also an approximating polynomial in $L_q$.

To construct an $L_2$ approximating polynomial, we use spectral considerations.
Denote by $\lambda_1,\ldots,\lambda_{\ell}$ the eigenvalues of $\mathrm{T}^{(\rho)}$ on the space of degree $d$ functions. Note that
if $P$ is a polynomial such that $P(\lambda_i) = 1$ for all $i$, then $P(\mathrm{T}^{\rho}) f = f$ for all $f$ of degree $d$. However,
as $\ell$ may be very large, there may not be a polynomial $P$ with small $\norm{P}$ satisfying $P(\lambda_i) = 1$ for all $i$, and
to circumvent this issue we must argue that, at least effectively, $\ell$ is small.
Indeed, while we do not show that $\ell$ is small, we do show that there are $d$ distinct values, $\lambda_1(\rho),\ldots,\lambda_d(\rho)$, such that
each $\lambda_i$ is very close to one of the $\lambda_j(\rho)$'s. This, by interpolation, implies that we may find a low-degree polynomial $P$
such that $P(\lambda_i)$ is very close to $1$ for all $i=1,\ldots,\ell$. Finally, to argue that $\norm{P}$ is small, we show that each $\lambda_i(\rho)$ is bounded
away from $0$.

It remains then to establish the claimed properties of the eigenvalues $\lambda_1,\ldots,\lambda_{\ell}$, and
we do so in several steps. We first identify the eigenspaces of $\mathrm{T}^{(\rho)}$ among the space of low-degree functions,
and show that each one of them contains a junta. Intuitively, for juntas it is much easier to understand the action of the $\mathrm{T}^{(\rho)}$,
since when looking on very few coordinates, $S_n$ looks like a product space. Indeed, using this logic we are able to show that all eigenvalues
of $\mathrm{T}^{(\rho)}$ on low-degree functions are bounded away from $0$. To argue that the eigenvalues are concentrated on a few values, we use the
fact that taking symmetry into account, the number of linearly independent juntas is small.

Our proof uses several notations appearing in Section~\ref{sec:level decomposition}, including the actions of $S_n$ on functions from the left $\prescript{\tau}{}f$ and from the right $f^\sigma$, the level decomposition $V_d$, the spaces $V_{A,B}$, and the concept of $d$-junta.


\subsection{Identifying the eigenspaces of $\mathrm{T}^{(\rho)}$}
\subsubsection{$\mathrm{T}^{\left(\rho\right)}$ commutes with the action of $S_{n}$ as a bimodule}
\begin{lem}\label{lem:Trho is symmetric}
The operator $\mathrm{T}^{\left(\rho\right)}$
commutes with the action of $S_{n}$ as a bimodule.
\end{lem}

The proof relies on the following claims.
\begin{claim}
\label{claim:composition}
If $\mathrm{T},\mathrm{S}$ are operators
that commute with the action of $S_{n}$ as a bimodule, then so is
$\mathrm{T\circ S}$.
\end{claim}

\begin{proof}
We have $\prescript{\pi_{1}}{}{\left(\mathrm{T}Sf\right)}^{\pi_{2}}=\mathrm{T}\left(\prescript{\pi_{1}}{}{Sf}^{\pi_{2}}\right)=\mathrm{TS}(\prescript{\pi_{1}}{}f^{\pi_{2}})$.
\end{proof}

Let $X$ and $Y$ be $S_n$-bimodules, and consider $X\times Y$ as an $S_n$-bimodule with the operation
$\prescript{\sigma_{1}}{}{\left(x,y\right)}^{\sigma_{2}}=\left(\prescript{\sigma_{1}}{}x^{\sigma_{2}},\prescript{\sigma_{1}}{}y^{\sigma_{2}}\right)$.
We say that a probability distribution $\mu$ on $X\times Y$ is invariant under the action of $S_n$ on both sides if
$\mu(\prescript{\sigma_{1}}{}{\left(x,y\right)}^{\sigma_{2}}) = \mu(x,y)$ for all $x\in X$, $y\in Y$ and $\sigma_1,\sigma_2\in S_n$.
\begin{claim}
\label{claim:symmetry of operators}
Let $X,Y$ be $S_{n}$-bimodules that are coupled by the probability measure $\mu$, and suppose that $\mu$ is invariant under
the action of $S_{n}$ from both sides. Then the operators $\mathrm{T}_{X\to Y},\mathrm{T}_{Y\to X}$ commute with the action
of $S_{n}$ from both sides.
\end{claim}

\begin{proof}
We prove the claim for $\mathrm{T}_{X\to Y}$ (the argument for $\mathrm{T}_{Y\to X}$ is identical).
Let $\mu_{X},\mu_{Y}$ be the marginal distributions of $\mu$ on $X$ and on $Y$, and for each $x\in X$ denote by
$1_{x}$ the indicator function of $x$. Then the set $\left\{ 1_{x}\right\} _{x\in X}$
is a basis for $L^{2}\left(X\right)$, and so it is enough to show that for all
$x$ and $\sigma_1,\sigma_2\in S_n$ it holds that
$\prescript{\sigma_{1}}{}{\left(\mathrm{T}_{X\to Y}1_{x}\right)} ^{\sigma_2} = \mathrm{T}_{X\to Y}\left(\prescript{\sigma_{1}}{}{1_{x}} ^{\sigma_2}\right)$.
Note that as these are two functions over $Y$, it is enough to show that
\[
\left\langle \prescript{\sigma_{1}}{}{\left(\mathrm{T}_{X\to Y}1_{x}\right)}^{\sigma_{2}},1_{y}\right\rangle =\left\langle \mathrm{T}_{X\to Y}\left(\prescript{\sigma_{1}}{}{1_{x}}^{\sigma_{2}}\right),1_{y}\right\rangle
\]
for all $y$, since $\left\{1_y\right\}_{y\in Y}$ forms a basis for $L^2(Y)$.

Fix $x$ and $y$.
Since $\mu$ is invariant under the action of $S_n$ on both sides, it follows that $\mu_{Y}$ is invariant under the action of
$S_{n}$, so we have
\[
\left\langle \prescript{\sigma_{1}}{}{\left(\mathrm{T}_{X\to Y}1_{x}\right)}^{\sigma_{2}},1_{y}\right\rangle = \left\langle \mathrm{T}_{X\to Y}1_{x},\prescript{\sigma_{1}^{-1}}{}{1_{y}}^{\sigma_{2}^{-1}}\right\rangle
=\left\langle\mathrm{T}_{X\to Y}1_{x},1_{\prescript{\sigma_1}{}y^{\sigma_{2}}}\right\rangle
=\mu\left(x,\sigma_{1} y\sigma_{2}\right),
\]
where in the penultimate transition we used the fact that
$\prescript{\sigma_1^{-1}}{} {1_y} ^{\sigma_2^{-1}} = 1_{\prescript{\sigma_1}{} y ^{\sigma_2}}$.
On the other hand, we also have that the last fact holds for $1_x$, and so
\[
\left\langle \mathrm{T}_{X\to Y}\left(\prescript{\sigma_{1}}{}{1_{x}}^{\sigma_{2}}\right),1_{y}\right\rangle
=
\left\langle \mathrm{T}_{X\to Y}1_{\prescript{\sigma_{1}^{-1}}{}x^{\sigma_{2}^{-1}}},1_{y}\right\rangle
=\mu\left(\sigma_{1}^{-1}x\sigma_{2}^{-1},y\right).
\]
The claim now follows from the fact that $\mu$ is invariant under the action of $S_{n}$ from both sides.
\end{proof}
We are now ready to move on to the proof of Lemma \ref{lem:Trho is symmetric}.
\begin{proof}[Proof of Lemma \ref{lem:Trho is symmetric}]
 We let $S_{n}$ act on $L$ from the right by setting $\left(i,j\right)\pi=\left(\pi\left(i\right),j\right)$
and from the left by setting $\pi\left(i,j\right)=\left(i,\pi\left(j\right)\right)$.
For a function $f$ on $L^{m}$ we write $\prescript{\pi_{1}}{}f^{\pi_{2}}$
for the function
\[
\left(x_{1},\ldots,x_{m}\right)\mapsto f\left(\pi_{1}x_{1}\pi_{2},\ldots,\pi_{1}x_{m}\pi_{2}\right).
\]

By Claim \ref{claim:symmetry of operators} the operators $\mathrm{T}_{\rho},\mathrm{T}_{S_{n}\to L^{m}},\mathrm{T}_{L^{m}\to S_{n}}$
commute with the action of $S_{n}$ as a bimodule, and therefore so
is $\mathrm{T}^{\left(\rho\right)}$ by Claim \ref{claim:composition}.
\end{proof}

\subsubsection{Showing that the spaces $V_{A,B}$ and $V_{d}$ are invariant under
$\mathrm{T}^{\left(\rho\right)}$}

First we show that $V_{A,B}$ is an invariant subspace of $\mathrm{T}^{\left(\rho\right)}$.
\begin{lem}
\label{lem:abjuntas}
Let $\mathrm{T}$ be an endomorphism of $L^{2}\left(S_{n}\right)$
as an $S_{n}$-bimodule. Then $\mathrm{T}V_{A,B}\subseteq V_{A,B}$.
Moreover, $TV_{d}\subseteq V_{d}$.
\end{lem}

\begin{proof}
Let $f\in V_{A,B}$. We need to show that $\mathrm{T}f\in V_{A,B}$.
Let $\sigma_{1}\in S_{\left[n\right]\setminus A},\sigma_{2}\in S_{\left[n\right]\setminus B}$.
Then
\[
\prescript{\sigma_{1}}{}{\left(\mathrm{T}f\right)}^{\sigma_{2}}=\mathrm{T}\left(\prescript{\sigma_{1}}{}f^{\sigma_{2}}\right)=\mathrm{T}f,
\]
 where the first equality used the fact that $\mathrm{T}$ commutes
with the action of $S_{n}$ from both sides, and the second inequality
follows from Lemma~\ref{lem:V_AB bimodule}.
The `moreover' part follows from Lemma~\ref{lem:V_AB spanning set}.
\end{proof}
\begin{lem}
\label{lem:contains a junta} Let $\lambda$ be an eigenvalue of $\mathrm{T}^{\left(\rho\right)}$
as an operator from $V_{d}$ to itself. Let $V_{d,\lambda}$ be
the eigenspace corresponding to $\lambda$. Then $V_{d,\lambda}$
contains a $d$-junta.
\end{lem}

\begin{proof}
Since each space $V_{A,B}$ is $\mathrm{T}^{\left(\rho\right)}$ invariant,
we may decompose each $V_{A,B}$ into eigenspaces $V_{A,B}^{\left(\lambda\right)}$.
Let
\[
V_{d}^{(\lambda)}
=
\sum_{|A|,|B|\leq d}V_{A,B}^{\left(\lambda\right)}.
\]
 Then for each $\lambda$, $V_{d}^{(\lambda)}$ is an eigenspaces of $\mathrm{T}^{\left(\rho\right)}$ with eigenvalue $\lambda$, and
\[
\sum_{\lambda}V_{d}^{(\lambda)}=\sum_{|A|,|B|\leq d}V_{A,B}=V_{d} = \sum_{\lambda}V_{d,\lambda}.
\]
By uniqueness, it follows that $V_{d,\lambda} = V_{d}^{(\lambda)}$ for all $\lambda$. Fix $\lambda$; then we get that
there are $|A|,|B|\leq d$ such that $V_{A,B}^{\lambda}\subseteq V_{d,\lambda}$, and since any function in $V_{A,B}$ is
a $d$-junta by definition, the proof is concluded.
\end{proof}

We comment that the representation theory of $S_n$ supplies us with explicit formulas for $2d$-juntas in $V_{d,\lambda}$ (arising in the construction of Specht modules), which can be turned into $d$-juntas by symmetrization. Since we will not need such explicit formulas here, we skip this description.

\subsection{Finding a basis for $V_{A,B}$}
We now move on to the study of the spaces $V_{A,B}$. These spaces
have small dimension and are therefore easy to analyse. We first construct
a set $\left\{ v_{T}\right\} $ of functions in $V_{A,B}$ that form
a nearly-orthonormal basis.
\begin{defn}
Let $T=\left\{ \left(i_{1},j_{1}\right),\ldots,\left(i_{k},j_{k}\right)\right\} \subseteq\left[d\right]^{2}$
be  consistent. Let $1_{T}$ be the indicator
function of permutation $\pi$ in $S_{n}$ that satisfy the
restrictions given by $T$, i.e.\ $\pi\left(i_{1}\right)=j_{1},\ldots,\pi\left(i_{i_{k}}\right)=j_{k}$.
We define $v_{T}=\frac{1_{T}}{\|1_{T}\|_{2}}$.
\end{defn}

Since the spaces $V_{A,B}$ are isomorphic (as $S_{n-d}$ bimodules)
for all sets $A,B$ of size $d$, we shall focus on the case where
$A = B =[d]$.
\begin{lem}
Let $d\le\frac{n}{2}$, and let $T\ne S$ be sets of size $d$. Then
$\left\langle v_{T},v_{S}\right\rangle \le O\left(\frac{1}{n}\right)$.
\end{lem}

\begin{proof}
If $T\cup S$ is not consistent, then $1_{T}1_{S}=0$ and so $\left\langle v_{T},v_{S}\right\rangle =0$.
Otherwise,
\[
\left\langle v_{T},v_{S}\right\rangle =\frac{\mathbb{E}\left|1_{T\cup S}\right|}{\|1_{T}\|_{2}\|1_{S}\|_{2}}=\frac{\left(n-\left|T\cup S\right|\right)!}{\sqrt{\left(n-\left|T\right|\right)!\left(n-\left|S\right|\right)!}}\le \frac{(n-d-1)!}{(n-d)!} =
 O\left(\frac{1}{n}\right).\qedhere
\]
\end{proof}
\begin{prop}
\label{large inner product} There exists an absolute constant $c>0$
such that for all consistent $T\subseteq L$ we have
\[
\left\langle \mathrm{T}^{\left(\rho\right)}v_{T},v_{T}\right\rangle \ge\left(c\rho\right)^{\left|T\right|}.
\]
\end{prop}

\begin{proof}
Let $x\sim L^{m},y\sim N_{\rho}\left(x\right)$, and let $\sigma_{x},\sigma_{y}\in S_{n}$
be corresponding permutations chosen according to the coupling. We
have
\[
\left\langle \mathrm{T}^{\left(\rho\right)}v_{T},v_{T}\right\rangle =\frac{n!}{\left(n-\left|T\right|\right)!}\left\langle \mathrm{T}^{\left(\rho\right)}1_{T},1_{T}\right\rangle ,
\]
as $\norm{1_T}_2^2=\frac{\left(n-\left|T\right|\right)!}{n!}$.
We now interpret $\left\langle \mathrm{T}^{\left(\rho\right)}1_{T},1_{T}\right\rangle $
as the probability that both $\sigma_{x}$ and $\sigma_{y}$ satisfy
the restrictions given by $T$. For each ordered subset $S\subseteq\left[2n\right]$
of size $\left|T\right|$ consider the event $A_{S}$ that $x_{S}=y_{S}=T$,
while all the coordinates of the vectors $x_{\left[2n\right]\setminus S},y_{\left[2n\right]\setminus S}$
do not contradict $T$ and do not belong to $T$. Then
\[
\left\langle \mathrm{T}^{\left(\rho\right)}x_{T},x_{T}\right\rangle \ge\sum_{S\text{ an ordered $|T|$-subset of }\left[2n\right]}\Pr\left[A_{S}\right].
\]

Now the probability that $x_{S}=T$ is $\left(\frac{1}{n}\right)^{2\left|T\right|}$.
Conditioned on $x_{S}=T$, the probability that $y_{S}=T$ is at least
$\rho^{\left|T\right|}$. When we condition on $x_{S}=y_{S}=T$, we
obtain that the probability that $x_{\left[n\right]\setminus S}$
and $y_{\left[n\right]\setminus S}$ do not involve any coordinate
contradicting $T$ or in $T$ is at least $\left(1-\frac{2\left|T\right|}{n}\right)^{2n}=2^{-\Theta\left(|T|\right)}$.
Hence $\Pr\left[A_{S}\right]\ge\left(\frac{1}{n}\right)^{2\left|T\right|}\Omega\left(\rho\right)^{\left|T\right|}$.
So wrapping everything up we obtain that
\[
\left\langle \mathrm{T}^{\left(\rho\right)}v_{T},v_{T}\right\rangle \ge
\frac{(2n)!}{(2n-|T|)!}\cdot\frac{n!}{(n-|T|)!}\frac{1}{n^{2|T|}} \Omega(\rho)^{|T|}
=\Omega\left(\rho\right)^{\left|T\right|}.\qedhere
\]
\end{proof}
\begin{lem}
\label{lem:small inner product} Let $\rho\in\left(0,1\right)$. Then
for all sets $T\ne S$ of size at most $n/2$ we have
$\left\langle \mathrm{T}^{\left(\rho\right)}v_{T},v_{S}\right\rangle =O\left(\frac{1}{\sqrt{n}}\right)$.
\end{lem}

\begin{proof}
Suppose without loss of generality that $\|1_{T}\|_{2}^{2}\le\|1_{S}\|_{2}^{2}$,
so $\left|T\right|\ge\left|S\right|$. Choose $x\sim L^{m},y\sim N_{\rho}\left(x\right)$,
and let $\sigma_{x},\sigma_{y}$ by the corresponding random permutations
given by the coupling. We have
\[
\left\langle \mathrm{T}^{\left(\rho\right)}v_{T},v_{S}\right\rangle =\text{\ensuremath{\frac{\Pr\left[1_{T}\left(\sigma_{x}\right)=1,1_{S}\left(\sigma_{y}\right)=1\right]}{\sqrt{\mathbb{E}1_{T}\mathbb{E}1_{S}}}}}
\]
 As the probability in the numerator is at most $\mathbb{E}\left[1_{T}\right]$,
we have
\[
\left\langle T^{\left(\rho\right)}v_{T},v_{S}\right\rangle \le\sqrt{\frac{\mathbb{E}\left[1_{T}\right]}{\mathbb{E}\left[1_{S}\right]}}=\sqrt{\frac{\left(n-\left|T\right|\right)!}{\left(n-\left|S\right|\right)!}},
\]
and the proposition follows in the case that $\left|S\right|<\left|T\right|$.

It remains to prove the proposition provided that $\left|S\right|=\left|T\right|$.
Let $\left(i,j\right)\in S\setminus T$. Note that
\[
\Pr\left[1_{T}\left(\sigma_{x}\right)=1,1_{S}\left(\sigma_{y}\right)=1\right]
\leq
\frac{1}{n}\Pr\left[1_{T}\left(\sigma_{x}\right)=1 \mid \sigma_{y}\left(i\right)=j\right].
\]
Let us condition further on $\sigma_{x}\left(i\right)$. Conditioned
on $\sigma_{x}\left(i\right)=j$, we have that $\sigma_{x}$ is
a random permutation sending $i$ to $j$, and so
$\Pr\left[1_{T}\left(\sigma_{X}\right)=1\right]$ is either
0 (if $\left(i,j\right)$ contradicts $T$) or $\frac{\left(n-1-\left|T\right|\right)!}{\left(n-1\right)!}=O\left(\|1_{T}\|_{2}^{2}\right)$
(if $\left(i,j\right)$ is consistent with $T$).

Conditioned on $\sigma_{x}\left(i\right)\ne j$ (and on $\sigma_{y}\left(i\right)=j$),
we again obtain that $\sigma_{x}$ is a random permutation that does
not send $i$ to $j$, in which case
\[
\Pr\left[1_{T}\left(\sigma_{x}\right)=1\right]=\frac{\left(n-\left|T\right|\right)!}{n!-\left(n-1\right)!}=O\left(\|1_{T}\|_{2}^{2}\right)
\]
 if $\left(i,j\right)$ contradicts $T$, and
\[
\Pr\left[1_{T}\left(\sigma_{x}\right)=1\right]=\frac{\left(n-\left|T\right|\right)!-\left(n-\left|T\right|-1\right)!}{n!-\left(n-1\right)!}=O\left(\|1_{T}\|_{2}^{2}\right)
\]
if $\left(i,j\right)$ is consistent with $T$. This completes the
proof of the lemma.
\end{proof}

\begin{prop}
\label{prop:large inner product} Let $C$ be a sufficiently large
constant. If $n\ge\left(\frac{\rho}{C}\right)^{-d}C^{d^{2}}$ and
$f$ is a $d$-junta, then
\[
\left\langle \mathrm{T}^{\left(\rho\right)}f,f\right\rangle \ge\rho^{O\left(d\right)}\|f\|_{2}^{2}.
\]
\end{prop}

\begin{proof}
Since $\left\{ v_{T}\right\} _{T\subseteq\left[d\right]^{2}}$ span
the space $V_{\left[d\right],\left[d\right]}$ of $\left(\left[d\right],\left[d\right]\right)$-juntas by Lemma~\ref{lem:V_AB spanning set},
we may write $f=\sum a_{T}v_{T}$. Now
\begin{align*}
\left\langle \mathrm{T}^{\left(\rho\right)}f,f\right\rangle  & =\sum_{T}a_{T}^{2}\left\langle \mathrm{T}^{\left(\rho\right)}v_{T},v_{T}\right\rangle +\sum_{T\ne S}a_{T}a_{S}\left\langle \mathrm{T}^{\left(\rho\right)}v_{T},v_{S}\right\rangle .
\end{align*}
 By Lemma \ref{lem:small inner product} we have
\begin{align*}
\left|\sum_{T\ne S}a_{T}a_{S}\left\langle \mathrm{T}^{\left(\rho\right)}v_{T},v_{S}\right\rangle \right|
  \le O\left(\sum_{T\ne S}\frac{\left|a_{T}a_{S}\right|}{\sqrt{n}}\right)
  \le O\left(\frac{1}{\sqrt{n}}\right)\left(\sum_{T}\left|a_{T}\right|\right)^{2}
  \le\frac{2^{O\left(d^{2}\right)}}{\sqrt{n}}\left(\sum_{T}\left|a_{T}\right|^{2}\right),
\end{align*}
 where the last inequality is by Cauchy--Schwarz. On the other hand, by Proposition \ref{large inner product}
we have
\[
\sum_{T}a_{T}^{2}\left\langle \mathrm{T}^{\left(\rho\right)}v_{T},v_{T}\right\rangle \ge\rho^{O\left(d\right)}\left(\sum_{T}a_{T}^{2}\right).
\]
Using a similar calculation, one sees that
\[
\|f\|_{2}^{2}=\left(1\pm\frac{2^{O\left(d^{2}\right)}}{n}\right)\sum_{T}a_{T}^{2},
\]
so we get that
\[
\left\langle \mathrm{T}^{\left(\rho\right)}f,f\right\rangle
\geq\left(\rho^{O(d)} - \frac{2^{O\left(d^{2}\right)}}{\sqrt{n}}\right)\sum_{T}a_{T}^{2}
\geq\left(\rho^{O(d)} - \frac{2^{O\left(d^{2}\right)}}{\sqrt{n}}\right)\norm{f}_2^2
\geq \rho^{O(d)}\norm{f}_2^2.\qedhere
\]

\end{proof}
\begin{cor}
\label{cor:eigenvalues} Let $C$ be a sufficiently large absolute
constant. If $n\ge\left(\frac{\rho}{C}\right)^{-d}C^{d^{2}}$ then
all the eigenvalues of $\mathrm{T}^{\left(\rho\right)}$ as an operator
from $V_{d}$ to itself are at least $\rho^{O\left(d\right)}$.
\end{cor}

\begin{proof}
By Lemma~\ref{lem:contains a junta}, each eigenspace $V_{d,\lambda}$
contains a $d$-junta. Let $f\in V_{d,\lambda}$ be
a nonzero $d$-junta. Then by Proposition \ref{prop:large inner product},
\[
\lambda=\frac{\left\langle \mathrm{T}^{\left(\rho\right)}f,f\right\rangle }{\|f\|_{2}^{2}}\ge\rho^{O\left(d\right)}.\qedhere
\]
\end{proof}

\subsection{Showing that the eigenvalues of $\mathrm{T}^{\left(\rho\right)}$
on $V_{d}$ are concentrated on at most $d$ values}

Let $\lambda_{i}\left(\rho\right)=\left\langle \mathrm{T}^{\left(\rho\right)}v_{T},v_{T}\right\rangle $,
where $T$ is a set of size $i$. Then symmetry implies that $\lambda_{i}\left(\rho\right)$
does not depend on the choice of $T$.
\begin{lem}
\label{lem:approximating the eigenvalues}
Suppose that $n\ge\left(\frac{\rho}{C}\right)^{O\left(d\right)}C^{d^{2}}$.
Then each eigenvalue of $\mathrm{T}^{\left(\rho\right)}$ as an operator
on $V_{d}$ is equal to $\lambda_{i}(\rho)\left(1\pm n^{-\frac{1}{3}}\right)$
for some $i\le d$.
\end{lem}

\begin{proof}
Let $\lambda$ be an eigenvalue of $\mathrm{T}^{\left(\rho\right)}$,
and let $f$ be a corresponding eigenfunction in $V_{\left[d\right],\left[d\right]}$.
Write
\[
f=\sum a_{S}v_{S},
\]
 where the sum is over all $S=\left\{ \left(i_{1},j_{1}\right),\ldots,\left(i_{t},j_{t}\right)\right\} \subseteq\left[d\right]$.
Then $0=\mathrm{T}^{\left(\rho\right)}f-\lambda f$,
 but on the other hand for each set $S$ we have
\begin{align*}
\left\langle \mathrm{T}^{\left(\rho\right)}f-\lambda f,v_{S}\right\rangle  & =a_{S}\left(\left\langle \mathrm{T}^{\left(\rho\right)}v_{S},v_{S}\right\rangle -\lambda\right)
  \pm\sum_{\left|S\right|\ne\left|T\right|}\left|a_{T}\right|\left(\left|\left\langle \mathrm{T}^{\left(\rho\right)}v_{T},v_{S}\right\rangle \right|+\left|\lambda\right|\left|\left\langle v_{T},v_{S}\right\rangle\right|\right)\\
 & =a_{S}\left(\lambda_{\left|S\right|}\left(\rho\right)-\lambda\right)\pm O\left(\frac{\sum_{T\ne S}\left|a_{T}\right|}{\sqrt{n}}\right).
\end{align*}
Thus, for all $S$ we have that
\[
\card{a_S}\card{\lambda_{\left|S\right|}\left(\rho\right)-\lambda}\leq O\left(\frac{\sum_{T\ne S}\left|a_{T}\right|}{\sqrt{n}}\right).
\]
On the other hand, choosing $S$ that maximizes $\card{a_{S}}$, we find that
$\card{a_{S}}\ge\frac{\sum_{T\ne S}\left|a_{T}\right|}{2^{d^{2}}}$,
and plugging that into the previous inequality yields that $\card{\lambda_{\left|S\right|}\left(\rho\right)-\lambda}\le\frac{O\left(2^{d^{2}}\right)}{\sqrt{n}}\le n^{-0.4}\rho^{-d}\leq n^{-1/3}\lambda_{\left|S\right|}\left(\rho\right)$,
provided that $C$ is sufficiently large.
\end{proof}

\subsection{An $L^2$ variant of Lemma~\ref{lem:Approximation}}\label{sec:f_approx_poly}
\begin{lem}\label{lem:2-approximation} Let $n\ge\rho^{-Cd^{3}}$ for a sufficiently
large constant $C$. There exists a polynomial $P(z)=\sum_{i=1}^{k}a_{i}z^{i}$,
such that $\|P\|\le\rho^{-O\left(d^{3}\right)}$ and $\|P\left(\mathrm{T}^{\left(\rho\right)}\right)f-f\|_{2}\le n^{-2d}\|f\|_{2}$.
\end{lem}
\begin{proof}
 Choose $P(z) =
1-\prod_{i=1}^{d}\left(\lambda_{i}^{-1}z-1\right)^{9d}$, where $\lambda_i = \lambda_i(\rho)$.
 Orthogonally decompose $\mathrm{T}^{\left(\rho\right)}$ to write
$f=\sum_{\lambda}f^{=\lambda}$, for nonzero orthogonal functions
$f^{=\lambda}\in V_{d}$ satisfying $\mathrm{T}^{\left(\rho\right)}f^{=\lambda}=\lambda f^{=\lambda}$,
and let $g=P\left(\mathrm{T}^{\left(\rho\right)}\right)f-f$.
Then $g=\sum_{\lambda}\left(P\left(\lambda\right)-1\right)f^{=\lambda}$.
Therefore
\[
\|g\|_{2}^{2}=\sum_{\lambda}\left(P\left(\lambda\right)-1\right)^{2}\|f^{=\lambda}\|_{2}^{2}\le\max_{\lambda}(P\left(\lambda\right)-1)^2\|f\|_{2}^{2}.
\]
Suppose the maximum is attained at $\lambda_{\star}$. By Lemma \ref{lem:approximating the eigenvalues},
there is $i\leq d$ such that $\lambda_{\star}=\lambda_{i}(1\pm n^{-\frac{1}{3}})$, and so
\[
\card{\left(\lambda_i^{-1} \lambda_{\star} - 1\right)^{9d}}
\leq n^{-3d}.
\]
For any $j\neq i$, we have by Proposition \ref{prop:large inner product} that $\lambda_{j}\ge\rho^{O(d)}$, and so
\[
\card{\left(\lambda_i^{-1} \lambda_{\star} - 1\right)^{9d}}
\leq \rho^{-O(d^2)}.
\]
Combining the two inequalities, we get that
\[
(1-P(\lambda_{\star}))^2\leq
\rho^{-O(d^3)} n^{-6d}
\leq n^{-2d},
\]
where the last inequality follows from the lower bound on $n$. To finish up the proof then, we must upper bound $\norm{P}$,
and this is relatively straightforward:
\[
\norm{P}\leq 1 + \norm{\prod_{i=1}^{d}\left(\lambda_{i}^{-1}z-1\right)^{9d}}
\leq 1 + \prod_{i=1}^{d}\norm{\lambda_{i}^{-1}z-1}^{9d}
= 1 + \prod_{i=1}^{d}(1+\lambda^{-1})^{9d}
\leq 1 + \prod_{i=1}^{d}(1+\rho^{-O(d)})^{9d},
\]
which is at most $\rho^{-O(d^3)}$. In the second inequality, we used the fact that $\norm{P_1 P_2}\leq \norm{P_1}\norm{P_2}$.
\end{proof}

\subsection{Deducing the $L^q$ approximation}
To deduce the $L^q$ approximation of the polynomial $P$ from Lemma~\ref{lem:2-approximation}
we use the following basic type of hypercontractive inequality (this bound is often times too weak
quantitatively, but it is good enough for us since we have a very strong $L_2$ approximation).
\begin{lem}
\label{lem:traditional hypercontractivity }Let $C$ be sufficiently
large, and let $n\ge C^{d^{2}}q^{2d}$. Let $f\colon S_{n}\to\mathbb{R}$
be a function of degree $d$. Then
$\|f\|_{q}\le q^{O(d)}n^{d}\|f\|_{2}$.
\end{lem}

\begin{proof}
Let $\rho=\frac{1}{(10\cdot 4^8\cdot q)^{2}}$. Decomposing $f$ into the $\sum\limits_{\lambda} f_{=\lambda}$ where
$T^{(\rho)} f_{=\lambda} = \lambda f_{=\lambda}$, we may find $g$ of degree
$d$, such that $f=\mathrm{T}^{\left(\rho\right)}g$, namely, $g = \sum\limits_{\lambda}\lambda^{-1} f_{=\lambda}$.
By Parseval and Corollary~\ref{cor:eigenvalues}, we get that $\|g\|_{2}\le\rho^{-O(d)}\|f\|_{2}$. Thus, we have that
$\norm{f}_q = \norm{T^{(\rho)} g}_q$, and to upper bound this norm we intend to use Theorem~\ref{thm:Hypercontractivity},
and for that we show that $g$ is global with fairly weak parameters.


Let $T\subseteq L$ be consistent of size at most $2d$. Then
\[
\|g_{\rightarrow T}\|_{2}^2 = \frac{\E_{x} g(x) 1_T(x)}{\E_x[1_T(x)]}
\leq \sqrt{ \frac{\E_{x} g(x)^2}{\E_x[1_T(x)]}}
\leq n^{\frac{\left|T\right|}{2}}\|g\|_{2}^{2}
\leq n^{\frac{\left|T\right|}{2}}\rho^{-O(d)} \|f\|_{2}^{2},
\]
and so $g$ is $(2d,\eps)$ global for  $\epsilon=n^{d/2}\rho^{-O(d)}\|f\|_{2}$.
Lemma \ref{lem:Bootstrapping the globalness} now implies that $g$ is $\eps$-global with constant $4^8$. By the choice of
$\rho$, we may now use Theorem~\ref{thm:Hypercontractivity} to deduce that
\[
\norm{T^{(\rho)} g}_q\leq
\eps^{(q-2)/q} \norm{g}_2^{2/q}
\leq n^{d/2} \rho^{-O(d)}\norm{f}_2
\leq n^{d} q^{O(d)} \norm{f}_2.\qedhere
\]
\end{proof}

Finally, we combine Lemma~\ref{lem:2-approximation} and Lemma~\ref{lem:traditional hypercontractivity } to deduce the $L^q$ approximating polynomial.
\begin{proof}[Proof of Lemma~\ref{lem:Approximation}]
 Let $f$ be a function of degree $d$. By lemma \ref{lem:2-approximation}
there exists a $P$ with $\|P\|\le\rho^{-O\left(d^{3}\right)}$ and
$P\left(0\right)=0$ such that the function $g=P\left(\mathrm{T}^{\left(\rho\right)}\right)f-f$
satisfies $\|g\|_{2}\le n^{-2d}\|f\|_{2.}$ By Lemma \ref{lem:traditional hypercontractivity },
$\|g\|_{q}\le q^{4d}n^{-d}\|f\|_{2}\le\frac{1}{\sqrt{n}}\|f\|_{2}$,
provided that $C$ is sufficiently large, completing the proof.
\end{proof}

\section{Hypercontractivity: the direct approach}\label{sec:direct}
In this section, we give an alternative proof to a variant of Theorem~\ref{thm:Reasonability}.
This approach starts by identifying a trivial spanning set of the space $V_t$ of degree $t$ functions from Definition~\ref{def:level_space}.

\paragraph{Notations.} For technical reasons, it will be convenient for us to work with ordered sets. We denote by $[n]_t$ the collection
of ordered sets of size $t$, which are simply $t$-tuples of distinct elements from $[n]$, but we also allow set operations (such as $\setminus$)
on them. We also denote $n_t = \card{[n]_t} = n(n-1)\cdots(n-t+1)$.
For ordered sets $I = \set{i_1,\ldots,i_t}$, $J = \set{j_1,\ldots,j_t}$, we denote by $1_{I\rightarrow J}(\pi)$ the indicator of
$\pi(i_k) = j_k$ for all $k=1,\ldots,t$; for convenience, we also denote this by $\pi(I) = J$.

With the above notations, the following set clearly spans $V_t$, by definition:
\begin{equation}\label{eq:spanning}
\sett{1_{I\rightarrow J}}{\card{I} = \card{J}\leq t}.
\end{equation}

We remark that this set is not a basis, since these functions are linearly dependent. For example, for $t=1$ we have
$\sum_{i=1}^{n} 1_{\pi(1) = i} - 1 = 0$.
This implies that a function $f\in V_{1}$ has several different representations as a linear combination of functions from the spanning set~\eqref{eq:spanning}.
The key to our approach is to show that there is a way to canonically choose such a linear combination, which is both unique and works well with computations of high moments.

\begin{defn}\label{def:normalized}
  Let $f\in V_{=t}$, and suppose that $f = \sum\limits_{I,J\in[n]_t} \coef{a}{}{I}{J} 1_{I\rightarrow J}$. We say that this representation is \emph{normalized} if
  \begin{enumerate}
    \item For any $1\leq r\leq t$, $J = \set{j_1,\ldots,j_t}$ and $I = \set{i_1,\ldots,i_{r-1}, i_{r+1},\ldots, i_t}$ we have that
    \[
    \sum\limits_{i_r\not\in I} \coef{a}{}{\set{i_1,\ldots,i_t}}{J} = 0.
    \]
    \item Analogously, for any $1\leq r\leq t$, $I = \set{i_1,\ldots,i_t}$ and $J = \set{j_1,\ldots,j_{r-1}, j_{r+1},\ldots, j_t}$ we have that
    \[
    \sum\limits_{j_r\not\in J} \coef{a}{}{I}{\set{j_1,\ldots, j_t}} = 0.
    \]
    \item Symmetry: for all ordered sets $I,J$ of size $t$ and $\pi\in S_t$, we have $\coef{a}{}{I}{J} = \coef{a}{}{\pi(I)}{\pi(J)}$.
  \end{enumerate}
\end{defn}
More loosely, we say that a representation according to the spanning set~\eqref{eq:spanning} is normalized if averaging the coefficients according to a single
coordinate results in $0$. We also refer to the equalities in Definition~\eqref{eq:spanning} as ``normalizing relations''. In this section, we show that
a normalized representation always exists, and then show how it is useful in establishing hypercontractive statements similar to Theorem~\ref{thm:Reasonability}.

Normalized representations first appear in the context of the slice by Dunkl~\cite{Dunkl}, who called normalized representations \emph{harmonic functions}. See also the monograph of Bannai and Ito~\cite[III.3]{BannaiIto} and the papers~\cite{Filmus2016a,FM2019}. Ryan O'Donnell (personal communication) has proposed calling them \emph{zero-flux} representations.

\subsection{Finding a normalized representation}
\begin{lem}\label{lem:normalized_rep_exists}
  Let $0\leq t\leq t$, and let $f\in V_{t}$. Then we may write
  $f = h + g$, where $h\in V_{t-1}$ and $g$ is given by a set of
  coefficients satisfying the normalizing relations
  $g = \sum\limits_{I,J\in[n]_t}\coef{a}{t}{I}{J}1_{I\rightarrow J}(\pi)$.
\end{lem}
\begin{proof}
  The proof is by induction on $t$.

  Fix $t\ge1$ and $f\in V_{t}$. Then we may write
  $f(\pi) = \sum\limits_{I,J\in[n]_t} \coef{a}{}{I}{J} 1_{I\rightarrow J}(\pi)$, where
  the coefficients satisfy the symmetry property from Definition~\ref{def:normalized}.

  Throughout the proof,
  we will change the coefficients in a sequential process, and always maintain the form
  $f = h + \sum\limits_{\card{I} = \card{J}=t} \coef{b}{}{I}{J} 1_{I\rightarrow J}(\pi)$ for $h\in V_{t-1}$.

  Take $r\in [t]$, and for each $I = \set{i_1,\ldots,i_t}$, $J = \set{j_1,\ldots,j_t}$, define the coefficients
  \begin{equation}\label{eq:change_coef}
  \coef{b}{}{I}{J} = \coef{a}{}{I}{J} - \frac{1}{n-t+1}\sum\limits_{i\not\in I\setminus\set{i_r}}
  \coef{a}{}{\set{i_1,\ldots,i_{r-1}, i, i_{r+1},\ldots,i_t}}{J}.
  \end{equation}
  In Claim~\ref{claim:change_coef} below, we prove that after making this change of coefficients, we may write
  $f = h + \sum\limits_{\card{I} = \card{J}=t} \coef{b}{}{I}{J} 1_{I\rightarrow J}(\pi)$, and that the coefficients
  $\coef{b}{}{I}{J}$ satisfy all normalizing relations that the $\coef{a}{}{I}{J}$ do, as well as the normalizing relations
  from the first collection in Definition~\ref{def:normalized} for $r$. We repeat this process for all $r\in[t]$.

  After this process is done, we have $f = h + \sum\limits_{I,J\in[n]_t} \coef{b}{}{I}{J} 1_{I\rightarrow J}(\pi)$,
  where the coefficients $\coef{a}{}{I}{J}$ satisfy the first collection of normalizing relations from Definition~\ref{def:normalized}.
  We can now perform the analogous process on the $J$ part, and by symmetry obtain that after this process, the second collection of normalizing
  relations in Definition~\ref{def:normalized} hold. One only has to check that this does not destroy the first collection of normalizing relations, which we also prove in Claim~\ref{claim:change_coef}.

  Finally, we symmetrize $f$ to ensure that it satisfies the symmetry condition. To do so, we replace $g = \sum\limits_{I,J\in[n]_t} \coef{b}{}{I}{J} 1_{I\to J}(\pi)$ with $g' = \sum_{\pi \in S_t} g^\pi$, where $(1_{I \to J})^\pi = 1_{\pi(I) \mapsto \pi(J)}$ (and extended linearly). It is easy to check that $g = g^\pi$ as functions, and that $g^\pi$ satisfies the two sets of normalizing relations. It follows that so does $g'$, and furthermore by construction, $g'$ is symmetric.
%
  \end{proof}

  \begin{claim}\label{claim:change_coef}
  The change of coefficients~\eqref{eq:change_coef} has the following properties:
    \begin{enumerate}
    \item The coefficients $\coef{b}{}{I}{J}$ satisfy the normalizing relation in the first item for $r$
    in Definition~\ref{def:normalized}.
    \item If the coefficients $\coef{a}{}{I}{J}$ satisfy the normalizing relation in the first item in Definition~\ref{def:normalized} for $r' \neq r$,
    then so do $\coef{b}{}{I}{J}$.
    \item If the coefficients $\coef{a}{}{I}{J}$ satisfy the normalizing relation in the second item in Definition~\ref{def:normalized} for $r'$, then so do $\coef{b}{}{I}{J}$.
    \item We may write $f = h + \sum\limits_{\card{I} = \card{J}=t} \coef{b}{}{I}{J} 1_{I\rightarrow J}(\pi)$, where $h\in V_{t-1}$.
  \end{enumerate}
  \end{claim}
\begin{proof}
We prove each one of the items separately.
  \paragraph{Proof of the first item.}
  Fix $I = \set{i_1,\ldots,i_{r-1},i_{r+1}\ldots,i_t}$, $J = \set{j_1,\ldots,j_t}$, and calculate:
  \begin{align}
    \sum\limits_{i_r\not\in I} \coef{b}{}{\set{i_1,\ldots,i_t}}{J}
    &=
    \sum\limits_{i_r\not\in I}
    \Bigg(\coef{a}{}{\set{i_1,\ldots,i_t}}{J} -
    \frac{1}{n-t+1}\sum\limits_{i\not\in I}\coef{a}{}{\set{i_1,\ldots,i_{r-1}, i, i_{r+1},\ldots,i_t}}{J}
    \Bigg)\notag\\\label{eq:stupid4}
    &=
    \sum\limits_{i_r\not\in I} \coef{a}{}{\set{i_1,\ldots,i_t}}{J}
    -\frac{1}{n-t+1}
    \sum\limits_{\substack{i_r\not\in I\\i\not\in I}}
    \coef{a}{}{\set{i_1,\ldots,i_{r-1}, i, i_{r+1},\ldots,i_t}}{J}.
  \end{align}
  As in the second double sum, for each $i_r$ the coefficient
  $\coef{a}{}{\set{i_1,\ldots,i_{r-1}, i_r, i_{r+1},\ldots,i_t}}{J}$
  is counted $n-\card{I} = n-t+1$ times, we get that the above expression
  is equal to $0$.

  \paragraph{Proof of the second item.} Fix $r' \neq r$, and suppose $\coef{a}{}{\cdot}{\cdot}$ satisfy the first set of normalizing relations for $r'$.
  Without loss of generality, assume $r'<r$.
  Let $I = \set{i_1,\ldots,i_{r'-1},i_{r'+1},\ldots,i_t}$, $J = \set{j_1,\ldots,j_t}$. Below, we let $i, i_{r'}$ be summation indices and we
  denote $I' = \set{i_1,\ldots,i_{r'-1},i_{r'},i_{r'+1},\ldots,i_{r-1},i,i_r,\ldots,,i_t}$.
  Calculating as in~\eqref{eq:stupid4}:
  \begin{align}
    \sum\limits_{i_{r'}\not\in I} \coef{b}{}{\set{i_1,\ldots,i_t}}{J}
    &=
    \sum\limits_{i_{r'}\not\in I} \coef{a}{}{\set{i_1,\ldots,i_t}}{J} -
    \frac{1}{n-t+1}\sum\limits_{i\not\in I\setminus\set{i_r}}\coef{a}{}{I'}{J}\notag\\\label{eq:stupid5}
    &=
    \sum\limits_{i_{r'}\not\in I} \coef{a}{}{\set{i_1,\ldots,i_t}}{J}
    -\frac{1}{n-t+1}
    \sum\limits_{i_{r'}\not\in I}
    \sum\limits_{i\not\in I\setminus\set{i_{r}}}\coef{a}{}{I'}{J}.
  \end{align}
  The first sum is $0$ by the assumption of the second item. For the second sum, we interchange the order of summation to see that it is equal to
  $\sum\limits_{i\not\in I\setminus\set{i_r}}  \sum\limits_{i_{r'}\not\in I} \coef{a}{}{I'}{J}$, and note that for each $i$,
  the inner sum is $0$ again by the assumption of the second item.

  \paragraph{Proof of the third item.} Fix $r'$, and suppose $\coef{a}{}{\cdot}{\cdot}$ satisfy the second set of normalizing relations for $r'$. Fix $I = \{i_1,\ldots,i_t\}$, $J = \{j_1,\ldots,j_{r'-1},j_{r'+1},\ldots,j_t\}$, $I' = \{i_1,\ldots,i_{r-1},i,i_{r+1},\ldots,i_t\}$, $J' = \{j_1,\ldots,j_t\}$, and calculate:
  \begin{align}
  \sum_{j_r \notin J} b(I,J') &=
  \sum_{j_r \notin J} \left(a(I,J') - \frac{1}{n-t+1} \sum_{i \notin I \setminus \{i_r\}} a(I',J') \right) \notag \\ &=	
  \sum_{j_r \notin J} a(I,J') - \frac{1}{n-t+1} \sum_{i \notin I \setminus \{i_r\}} \sum_{j_r \notin J} a(I',J').
  \end{align}
  Once again, both sums vanish due to the assumption.


  \paragraph{Proof of the fourth item.} For $I = \set{i_1,\ldots,i_t}$, $J = \set{j_1,\ldots,j_t}$,
  denote
  \[
  \coef{c}{}{I}{J} = \frac{1}{n-t+1}\sum\limits_{i\not\in I\setminus\set{i_r}}\coef{a}{}{\set{i_1,\ldots,i_{r-1}, i, i_{r+1},\ldots,i_t}}{J},
  \]
  so that $\coef{a}{}{I}{J} = \coef{b}{}{I}{J} + \coef{c}{}{I}{J}$. Plugging this into the representation of $f$, we see that it is enough
  to prove that $h(\pi) = \sum\limits_{I,J} \coef{c}{}{I}{J} 1_{I\rightarrow J}(\pi)$ is in $V_{t-1}$. Writing
  $I' = I\setminus \set{i_r}$, $J' = J\setminus\set{j_r}$ and expanding, we see that
  \begin{align*}
  h(\pi)
  &= \frac{1}{n-t+1}\sum\limits_{I,J}1_{I\rightarrow J}(\pi)\sum\limits_{i\not\in I\setminus\set{i_r}}\coef{a}{}{\set{i_1,\ldots,i_{r-1}, i, i_{r+1},\ldots,i_t}}{J}\\
  &= \frac{1}{n-t+1}\sum\limits_{I',J'}
  \sum\limits_{i\not\in I', j_r\not\in J'}
  \coef{a}{}{\set{i_1,\ldots,i_{r-1}, i, i_{r+1},\ldots,i_t}}{J}
  \sum\limits_{i_r\not\in I'}
  1_{I\rightarrow J}(\pi).
  \end{align*}
  Noting that
  $\sum\limits_{i_r\not\in I'}
  1_{I\rightarrow J}(\pi) = 1_{I'\rightarrow J'}(\pi)$ is in the spanning set~\eqref{eq:spanning} for $t-1$, the proof is concluded.
\end{proof}

Applying Lemma~\ref{lem:normalized_rep_exists} iteratively,
we may write each $f\colon S_n\to\mathbb{R}$ of degree at most $t$
as $f = f_0+\ldots+f_d$, where for each $k=0,1,\ldots,d$, the function $f_k$ is in $V_k$, and is given by a list of coefficients satisfying the normalizing relations.

\subsection{Usefulness of normalized representations}
In this section we establish a claim that demonstrates the usefulness of the normalizing relations.
Informally, this claim often serves as a replacement for the orthogonality property that is so useful in
product spaces. Formally, it allows us to turn long sums into short sums, and is very helpful in various
computations arising in computations in norms of functions on $S_n$ that are given in a normalized representation.
\begin{claim}\label{claim:aux}
  Let $r\in\set{1,\ldots,d}$, $0\leq t<r$.
  Let $J$ be of size $r$, $I$ be of size {\it at least} $r$,
  and $R\subseteq I$ of size $r-t$. Then
  \[
  \sum\limits_{T \in ([n]\setminus I)_t}{\coef{a}{r}{R\circ T}{J}}
  =(-1)^t~\sum\limits_{T\in (I\setminus R)_t}\coef{a}{r}{R\circ T}{J}.
  \]
\end{claim}
\begin{proof}
By symmetry, it suffices to prove the statement for $R$ that are prefixes of $I$.
We prove the claim by induction on $t$.
The case $t=0$ is trivial, so assume the claim holds for $t-1$, where $t\geq 1$, and prove for $t$.
The left hand side is equal to
\[
\sum\limits_{\substack{i_1,\ldots,i_t\not\in I\\ \text{distinct}}}{\coef{a}{r}{R\circ (i_1,\ldots,i_t)}{J}}.
\]
For fixed $i_1,\ldots,i_{t-1}\not\in I$, by the normalizing relations we have that
\[
\sum\limits_{i_t\not\in I\cup\set{i_1,\ldots,i_{t-1}}}{\coef{a}{r}{R\circ (i_1,\ldots,i_{t-1})\circ (i_t)}{J}}
= -\sum\limits_{i_t\in I\setminus R}{\coef{a}{r}{R\circ (i_1,i_2,\ldots,i_{t-1})\circ (i_t)}{J}},
\]
hence
\[
\sum\limits_{\substack{i_1,\ldots,i_t\not\in I\\ \text{distinct}}}{
\coef{a}{r}{R\circ (i_1,\ldots,i_t)}{J}}
=-\sum\limits_{i_t\in I\setminus R}{
\sum\limits_{\substack{~~~i_1,\ldots,i_{t-1}\not\in I\cup \set{i_t}\\ \text{distinct}}}{\coef{a}{r}{R\circ (i_1,i_2,\ldots,i_{t-1})\circ i_t}{J}}}.
\]
For fixed $i_t\in I\setminus R$, using the induction hypothesis, the inner sum is equal to
\[
(-1)^{t-1} \sum\limits_{T\in (I\setminus (R\cup\set{i_t}))_{t-1}}\coef{a}{r}{R\circ T\circ (i_t)}{J}.
\]
Plugging that in,
\begin{align*}
\sum\limits_{\substack{i_1,\ldots,i_t\not\in I\\ \text{distinct}}}{
\coef{a}{r}{R\circ (i_1,\ldots,i_t)}{J}}
&=(-1)^{t-1}~\sum\limits_{i_t\in I\setminus R}{-\sum\limits_{T\in (I\setminus (R\cup\set{i_t}))_{t-1}}\coef{a}{r}{R\circ T\circ (i_t)}{J}}\\
&=(-1)^t \sum\limits_{T'\in (I\setminus R)_t}\coef{a}{r}{R\circ T'}{J}.
\qedhere
\end{align*}
\end{proof}

\subsection{Analytic influences and the hypercontractive statement}
Key to the hypercontractive statement proved in this section is an analytic notion of influence. Given a fixed
representation of $f$ as $\sum\limits_{k=0}^{n}\sum\limits_{I,J\in [n]_k}{\coef{a}{k}{I}{J} 1_{I\rightarrow J}}$
where for each $k$ the coefficients $\coef{a}{k}{I}{J}$ satisfy the normalizing relations, we define the analytic
notion of influences as follows.
\begin{defn}
   For $S,T\subseteq [n]$ of the same size $s$, define
   \[
   I_{S,T}[f] =
   \sum\limits_{r\geq 0}
   \sum\limits_{\substack{I\in ([n]\setminus S)_r\\J\in ([n]\setminus T)_r}}
   {
   (r+s)!^2
   \frac{1}{n^{r+s}}
   \coef{a}{}{S\circ I}{T\circ J}^2}.
   \]
   Here, $S\circ I$ denotes the element in $[n]_r$ resulting from appending $I$ at the end of $S$.
\end{defn}



\begin{defn}
  A function $f$ is called $\eps$-analytically-global if for all $S,T$, $I_{S,T}[f]\leq \eps$.
\end{defn}
\begin{remark}
With some work it can be shown that for $d\ll n$, a degree $d$ function being $\eps$-analytically global is equivalent to
  $f$ being $(2d,\delta)$-global in the sense of Definition~\ref{def:bounded_global}, where $\delta = O_d(\eps)$. Thus, at least
  qualitatively, the hypercontractive statement below is in fact equivalent to Theorem~\ref{thm:Reasonability}.
\end{remark}

We can now state our variant of the hypercontractive inequality that uses analytic influences.
\begin{thm}\label{thm:hyp_sym}
  There exists an absolute constant $C>0$ such that for all $d,n\in\mathbb{N}$ for which $n\geq 2^{C\cdot d\log d}$, the following holds.
  If $f\in V_{d}$ is given by a list of coefficients satisfying the normalizing relations, say $f = \sum\limits_{I,J\in [n]_d} \coef{a}{d}{I}{J} 1_{I\rightarrow J}$,
  then
  \[
    \Expect{\pi}{f(\pi)^4}\leq
    \sum\limits_{\card{S}=\card{T}}\left(\frac{4}{n}\right)^{\card{S}} I_{S,T}[f]^2.
  \]
\end{thm}

\paragraph{$p$-biased hypercontractivity.}
The last ingredient we use in our proof is a hypercontractive inequality on the $p$-biased cube from~\cite{KLLM}.
Let $g\colon\power{m}\to\mathbb{R}$ be a degree $d$ function, where we think of $\power{m}$ as equipped with the $p$-biased
product measure. Then, we may write $g$ in the basis of characters, i.e.\ as a liner combination of $\set{\chi_S}_{S\subseteq[m]}$,
where $\chi_S(x) = \prod\limits_{i\in S}\frac{x_i - p}{\sqrt{p(1-p)}}$. This is the $p$-biased Fourier transform of $f$:
\[
g(x) = \sum\limits_{S}{\widehat{g}(S)\chi_S(x)}.
\]
Next, we define the generalized influences of sets (which are very close in spirit to the analytic notion of influences considered herein).
For $T\subseteq[n]$, we denote
\[
I_T[g] = \sum\limits_{S\supseteq T}{\widehat{g}(S)^2}.
\]
The following results is an easy consequence of~\cite[Theorem 3.4]{KLLM} (the deduction of it from
this result is done in the same way as the proof of~\cite[Lemma 3.6]{KLLM}).
\begin{thm}\label{thm:p_biased}
  Suppose $g\colon\power{m}\to\mathbb{R}$. Then
  $\norm{g}_4^4\leq \sum\limits_{T\subseteq[n]}(3p)^{\card{T}} I_T[g]^2$.
\end{thm}

\subsection{Proof of Theorem~\ref{thm:hyp_sym}}
  Write $f$ according to its normalized representation as $f(\pi) = \sum\limits_{I,J\in[n]_d} \coef{a}{}{I}{J} 1_{I\rightarrow J}$.
  We intend to define a function $g\colon \power{n\times n} \to \mathbb{R}$ that will behaves similary to $f$, as follows.
  We think of $\power{n\times n}$ as equipped with the $p$-biased
  measure for $p=1/n$, and think of an input $x\in \power{n\times n}$ as a matrix. The rationale is that the bit $x_{i,j}$ being $1$ will
  encode the fact that $\pi(i) = j$, but we will never actually think about it this way. Thus, we define $g$ as
  \[
  g(x) = \sum\limits_{I,J\in[n]_d} \coef{a}{}{I}{J} \prod\limits_{\ell=1}^{d}\left(1_{I_\ell\rightarrow J_{\ell}} - \frac{1}{n}\right).
  \]
  For $I,J$, we denote by $S_{I,J}\subseteq[n\times n]$ the set of coordinates $\sett{(I_{\ell},J_{\ell})}{\ell=1,\ldots,d}$,
  and note that with this notation,
  \[
   g(x) = \sum\limits_{I,J\in[n]_d} \sqrt{p(1-p)}^{d}\card{\coef{a}{}{I}{J}} \chi_{S_{I,J}}(x).
  \]
  To complete the proof, we first show (Claim~\ref{claim:g_uppers_f}) that $\norm{f}_4^4\leq (1+o(1))\norm{g}_4^4$,
  and then prove the desired upper bound on the $4$-norm of $g$, using Theorem~\ref{thm:p_biased}.
  \begin{claim}\label{claim:g_uppers_f}
    $\norm{f}_4^4\leq(1+o(1))\norm{g}_4^4$
  \end{claim}
  \begin{proof}
    Deferred to Section~\ref{sec:pf_g_uppers_f}.
  \end{proof}
  We now upper bound $\norm{g}_4^4$. Using Theorem~\ref{thm:p_biased},
  \begin{equation}\label{eq:direct_1}
  \norm{g}_4^4\leq\sum\limits_{T\subseteq[n\times n]}(3p)^{\card{T}} I_T[g]^2,
  \end{equation}
  and the next claim bounds the generalized influences of $g$ by the analytic influences of
  $f$.

  For two sets $I=\{i_1,\ldots,i_t\}$, $J=\{j_1,\ldots,j_t\}$ of the same size, let $S(I,J) = \{(i_1,j_1),\ldots,(i_t,j_t)\} \subseteq [n]\times[n]$.
  \begin{claim}\label{claim:gen_inf_upper}
  Let $T = S(I',J')$ be such that $I_T[g]\neq 0$. Then $I_T[g]\leq I_{I',J'}[f]$.
  \end{claim}
  \begin{proof}
  Take $T$ in this sum for which $I_T[g]\neq 0$,
  and denote $t=\card{T}$.  Then $T = \set{(i_1,j_1),\ldots,(i_t,j_t)} = S(I',J')$ for
  $I' = \set{i_1,\ldots,i_t}$, $J'=\set{j_1,\ldots,j_t}$ that are consistent. For $Q\subseteq [n]\times [n]$ of size $d$
  such that $T\subseteq Q$, let $S_{Q,T} = \sett{(I,J)}{T\subseteq S(I,J) = Q}$,
  and note that by the symmetry normalizing relation, $\coef{a}{}{I}{J}$ is constant on $(I,J)\in S_{Q,T}$.
  We thus get
  \[
  I_T[g]
  =\sum\limits_{Q}{\left(\sum\limits_{(I,J)\in S_{Q,T}}\sqrt{p(1-p)}^d\coef{a}{}{I}{J}\right)^2}
  \leq d! p^d\sum\limits_{Q}{\sum\limits_{(I,J)\in S_{Q,T}}\coef{a}{}{I}{J}^2},
  \]
  where we used the fact that the size of $S_{Q,T}$ is $d!$. Rewriting the sum by first choosing the locations of
  $T$ in $(I,J)$, we get that the last sum is at most
  \[
  d_t
  \sum\limits_{\substack{I\in ([n]\setminus I')_{d-t}\\J\in ([n]\setminus J')_{d-t}}}
  \coef{a}{}{I'\circ I}{J'\circ J}^2
  \]
  Combining all, we get that
  $I_T[g]\leq  \sum\limits_{\substack{I\in ([n]\setminus I')_{d-t}\\J\in ([n]\setminus J')_{d-t}}}
  d!^2\frac{1}{n^d}
  \coef{a}{}{I'\circ I}{J'\circ J}^2
  = I_{I',J'}[g]$.
  \end{proof}
  Plugging in Claim~\ref{claim:gen_inf_upper} into~\eqref{eq:direct_1} and using Claim~\ref{claim:g_uppers_f} finishes the proof of Theorem~\ref{thm:hyp_sym}.
  \subsubsection{Proof of Claim~\ref{claim:g_uppers_f}}\label{sec:pf_g_uppers_f}

  Let $I_r$ and $J_r$ be $d$-tuples of distinct indices from $[n]$. Then
  \[
  \Expect{\pi}{f(\pi)^4}
  =\sum\limits_{\substack{I_1,\ldots,I_4\\ J_1,\ldots,J_4}}{\coef{a}{}{I_1}{J_1}\cdots\coef{a}{}{I_4}{J_4} \Expect{\pi}{1_{\pi(I_1) = J_1}\cdots 1_{\pi(I_4) = J_4}}}.
  \]
  Consider the collection of constraints on $\pi$ in the product of the indicators.
  To be non-zero, the constraints should be consistent, so we only consider such tuples.
  Let $M$ be the number of different elements that appear in $I_1,\ldots,I_4$
  (which is at least $d$ and at most $4d$)
  We partition the outer sum according to $M$, and upper bound the contribution from each $M$ separately.
  Fix $M$; then the contribution from it is:
  \[
  \frac{1}{n_M}\sum\limits_{\substack{I_1,\ldots,I_4\\ J_1,\ldots,J_4\\ \text{type }M}}\coef{a}{}{I_1}{J_1}\cdots\coef{a}{}{I_4}{J_4}.
  \]
  We would like to further partition this sum according to the pattern in which the $M$ different elements of $I_1,\ldots,I_4$
  are divided between them (and by consistency, this determines the way the $M$ different elements of $J_1,\ldots,J_4$ are
  divided between them). There are at most $(2^4-1)^M\leq 2^{16d}$ different such configurations, thus we fix one such configuration
  and upper bound it (at the end multiplying the bound by $2^{16d}$). Thus, we have distinct $i_1,\ldots,i_M$ ranging over
   $[n]$, and the coordinate of each $I_r$ is composed of the $i_1,\ldots,i_M$ (and similarly $j_1,\ldots,j_M$ and the $J_r$'s),
  and our sum is
  \begin{equation}\label{eq:hyp_12}
  \frac{1}{n_M}\sum\limits_{\substack{i_1,\ldots,i_M \text{ distinct}\\ j_1,\ldots,j_M\text{ distinct}}}\coef{a}{}{I_1}{J_1}\cdots\coef{a}{}{I_4}{J_4}.
  \end{equation}

  We partition the $i_t$'s into the number of times they occur: let $A_1,\ldots,A_4$ be the sets of $i_t$ that appear in $1,2,3,$ or $4$ of the $I_r$'s.
  We note that $i_t$ and $j_t$ appear in the same $I_r$'s and always together (otherwise the constraints would be contradictory), and in particular
  $i_t\in A_j$ iff $j_t\in A_j$. Also, $M = \card{A_1} + \card{A_2}+\card{A_3}+\card{A_4}$.

  We consider contributions from configurations where $A_1=\emptyset$ and $A_1\neq \emptyset$ separately, and to control the latter group
  we show that the above sum may be upper bounded by $M^{2M}$ sums of in which $A_1 = \emptyset$.
  To do that, we show how to reduce the size of $A_1$ by allowing more sums, and then apply it iteratively.

  Without loss of generality, assume $i_1\in A_1$; then it is in exactly one of the $I_r$'s --- without loss of generality the last
  coordinate of $I_4$. We rewrite the sum as
  \begin{equation}\label{eq:hyp_9}
  \frac{1}{n_M}\sum\limits_{\substack{i_1,\ldots,i_M}}\coef{a}{}{I_1}{J_1}\coef{a}{}{I_2}{J_2}\coef{a}{}{I_3}{J_3}
  \sum\limits_{\substack{i_1\in [n]\setminus\set{i_2,\ldots,i_M}\\ j_1\in [n]\setminus\set{j_2,\ldots,j_M}}}{\coef{a}{}{I_4}{J_4}}.
  \end{equation}
  Consider the innermost sum. Applying Claim \ref{claim:aux} twice, we have
  \[
   \sum\limits_{\substack{i_1\in [n]\setminus\set{i_2,\ldots,i_M}\\ j_1\in [n]\setminus\set{j_2,\ldots,j_M}}}{\coef{a}{}{I_4}{J_4}}
   =\sum\limits_{\substack{i_1\in \set{i_2,\ldots,i_M}\setminus I_4\\ j_1\in \set{j_2,\ldots,j_M}\setminus J_4}}{\coef{a}{}{I_4}{J_4}}.
  \]
  Plugging that into \eqref{eq:hyp_9}, we are able to write the sum therein using $(M-r)^2$ sums (one for each choice of
  $i_1\in\set{i_2,\ldots,i_M}\setminus I_4$ and $j_1\in \set{j_2,\ldots,j_M}\setminus J_4$) on $i_2,\ldots,i_M,j_2,\ldots,j_M$,
  and thus we have reduced the size of $A_1$ by at least $1$, and have decreased $M$ by at least $1$. The last bit implies that the original normalizing factor is smaller by a factor of at least $1/n$ than the new one. Iteratively applying this procedure, we end up with $A_1=\emptyset$,
  and we assume that henceforth. Thus, letting $\mathcal{H}$ be the set of consistent $(I_1,\ldots,I_4,J_1,\ldots,J_4)$ in which
  each element in $I_1\cup\dots\cup I_4$ appears in at least two of the $I_i$'s, we get that
  \begin{align}
  \Expect{\pi}{f(\pi)^4}
  &\leq
  \left(1+\frac{d^{O(d)}}{n}\right)
  \sum\limits_{\substack{I_1,\ldots,I_4\\ J_1,\ldots,J_4\\ \text{from }\mathcal{H}}}
  {\card{\coef{a}{}{I_1}{J_1}}\cdots\card{\coef{a}{}{I_4}{J_4}} \Expect{\pi}{1_{\pi(I_1) = J_1}\cdots 1_{\pi(I_4) = J_4}}}\notag\\\label{eq:f_4_upper}
  &\leq (1+o(1))
  \sum\limits_{\substack{I_1,\ldots,I_4\\ J_1,\ldots,J_4\\ \text{from }\mathcal{H}}}
  {\frac{1}{n^{\card{I_1\cup\dots\cup I_4}}}\card{\coef{a}{}{I_1}{J_1}}\cdots\card{\coef{a}{}{I_4}{J_4}}},
  \end{align}
  where in the last inequality we used
  \[
  \Expect{\pi}{1_{\pi(I_1) = J_1}\cdots 1_{\pi(I_4) = J_4}} =
  \frac{1}{n\cdot(n-1)\cdots(n-\card{I_1\cup\ldots\cup I_4}+1)}
  \leq (1+o(1))\frac{1}{n^{\card{I_1\cup\ldots\cup I_4}}}.
  \]

  Next, we lower bound $\norm{g}_4^4$. Expanding as before,
  \[
  \Expect{x}{g(\pi)^4}
  =\sum\limits_{\substack{I_1,\ldots,I_4\\ J_1,\ldots,J_4}}{
  \sqrt{p(1-p)}^{4d}\card{\coef{a}{}{I_1}{J_1}}\cdots\card{\coef{a}{}{I_4}{J_4}}
  \Expect{x}{\chi_{S(I_1,J_1)}(x)\cdots\chi_{S(I_4,J_4)}(x)}}.
  \]
  A direct computation shows that the expectation of a normalized $p$-biased bit, i.e.\ $\frac{x_{i,j} - p}{\sqrt{p(1-p)}}$,
  is $0$, the expectation of its square is $1$, the expectation of its third power is $\frac{1+o(1)}{\sqrt{p(1-p)}}$,
  and the expectation of its fourth power is $\frac{1+o(1)}{p(1-p)}$. This tells us that all summands in
  the above formula are non-negative, and therefore we can omit all those that correspond to $(I_1,\ldots,I_4)$ and
  $(J_1,\ldots,J_4)$ not from $\mathcal{H}$, and only decrease the quantity. For $j=2,3,4$, denote by $h_j$ the number of elements
  that appear in $j$ of the $I_1,\ldots,I_4$. Then we get that the inner term is at least
  \[
  (1-o(1))
  \sqrt{p(1-p)}^{4d-h_3-2h_4}\card{\coef{a}{}{I_1}{J_1}}\cdots\card{\coef{a}{}{I_4}{J_4}}.
  \]
  Note that $2h_2+3h_3+4h_4 = 4d$, we get that
  $4d-h_3-2h_4 = 2(h_2+h_3+h_4) = 2\card{I_1\cup\dots\cup I_4}$. Combining everything, we get that
  \begin{align}
  \Expect{x}{g(\pi)^4}
  &\geq(1-o(1))\sum\limits_{\substack{I_1,\ldots,I_4\\ J_1,\ldots,J_4\\\text{from }\mathcal{H}}}{
  (p(1-p))^{\card{I_1\cup\dots\cup I_4}}\card{\coef{a}{}{I_1}{J_1}}\cdots\card{\coef{a}{}{I_4}{J_4}}}\notag\\\label{eq:g_4_upper}
  &\geq (1-o(1))\sum\limits_{\substack{I_1,\ldots,I_4\\ J_1,\ldots,J_4\\\text{from }\mathcal{H}}}{
  \frac{1}{n^{\card{I_1\cup\dots\cup I_4}}}\card{\coef{a}{}{I_1}{J_1}}\cdots\card{\coef{a}{}{I_4}{J_4}}}.
  \end{align}
  Combining~\eqref{eq:f_4_upper} and~\eqref{eq:g_4_upper} shows that $\norm{f}_4^4\leq(1+o(1))\norm{g}_4^4$.
  \qed

\subsection{Deducing hypercontractivity for low-degree functions}
With Theorem~\ref{thm:hyp_sym} in hand, one may deduce the following inequality as an easy corollary.
\begin{cor}\label{cor:hyp_sym}
  There exists an absolute constant $C>0$ such that for all $d,n\in\mathbb{N}$ for which $n\geq 2^{C\cdot d\log d}$, the following holds.
  If $f\in V_{d}(S_n)$ is $\epsilon$-analytically-global, then $\norm{f}_4^4\leq 2^{C \cdot d\log d} \eps^2$.
\end{cor}
\begin{proof}
  Since the proof is straightforward, we only outline its steps.
  Writing $f = f_{0}+\dots+f_{d}$ for $f_{k}\in V_{k}$ given by normalizing relations, one bounds
  $\norm{f}_4^4\leq (d+1)^3 \sum\limits_{k=0}^{d} \norm{f_{k}}_4^4$, uses Theorem~\ref{thm:hyp_sym}
  on each $f_{k}$, and finally $I_{I',J'}[f_{k}]\leq I_{I',J'}[f]\leq \eps$.
\end{proof}
\begin{remark}
  Using the same techniques, one may prove statements analogous to Theorem~\ref{thm:hyp_sym}
  and Corollary~\ref{cor:hyp_sym} for all even $q\in\mathbb{N}$.
\end{remark}

\section{Applications}\label{sec:app}

\subsection{Global functions are concentrated on the high degrees}
The first application of our hypercontractive is the following level-$d$ inequality.
\begin{reptheorem}{thm:lvl_d}
There exists an absolute constant $C>0$ such that the following holds.
Let $d,n\in\mathbb{N}$ and $\eps>0$ such that $n\geq 2^{C d^3} \log(1/\eps)^{C d}$.
If $f\colon S_n\to \{0,1\}$ is $(2d,\eps)$-global, then
$\|f^{\le d}\|^2 \leq 2^{C \cdot d^4}\eps^4 \log^{C\cdot d}(1/\eps)$.
\end{reptheorem}
\begin{proof}
  Deferred to Section~\ref{sec:lvl_d}.
\end{proof}
This result is analogous to the level $d$ inequality on the Boolean
hypercube~\cite[Corollary 9.25]{Od}, however it is quantitatively weaker because
our dependence on $d$ is poorer; for instance, it remains meaningful only for $d\leq \log(1/\eps)^{1/4}$,
wherein the original statement on the Boolean hypercube remains effective up to $d\sim \log(1/\eps)$.
Still, we show in Section~\ref{sec:product_sets} that this statement suffices to recover
results regarding the size of the largest product-free sets in $S_n$.

It would be interesting to prove a quantitatively better version of Theorem~\ref{thm:lvl_d} in terms of
$d$, and in particular whether it is the case that for $d = c\log(1/\eps)$
it holds that $\|f^{=d}\|^2 = \eps^{2-o(1)}$ for sufficiently small (but constant) $c>0$.

We remark that once Theorem~\ref{thm:lvl_d} has been established (or more precisely, the slightly stronger statement in
Proposition~\ref{prop:lvl_d}), one can strengthen it at the expense of assuming that $n$ is larger, namely establish
Theorem~\ref{thm:lvl_d_strong} from the introduction. We defer its proof to Section~\ref{sec:stronger_lvl_d}.

\subsection{Global product-free sets are small}\label{sec:product_sets}
In this section we prove a strengthening of Theorem~\ref{thm:global_prod_free}.
Conceptually, the proof is very simple. Starting with Gowers' approach,
we convert this problem into an independent set in a Cayley graph associated with $F$,
and use a Hoffman-type bound to solve that problem.

Fix a global product-free set $F\subseteq A_n$, and construct the (directed) graph $G_F$ as follows. Its vertex set is $S_n$, and
$(\pi,\sigma)$ is an edge if $\pi^{-1}\sigma\in F$. Note that $G_F$ is a Cayley graph, and that if $F$ is product-free,
then $F$ is an independent set in $G_F$. Our plan is thus to (1) study the eigenvalues of $G_F$ and prove good upper bounds
on them, and then (2) bound the size of $F$ using a Hoffman-type bound.

Let $T_F$ be the adjacency operator of $G_F$, i.e.\ the random walk that from a vertex $\pi$ transitions to a random neighbour $\sigma$
in $G_F$. We may consider the action of $T_F$ on functions $f\colon S_n\to\mathbb{R}$ as
\[
(T_F f)(\pi)
= \Expect{\sigma:(\pi,\sigma)\text{ is an edge}}{f(\sigma)}
= \Expect{a\in F}{f(\pi a)}.
\]

We will next study the eigenspaces and eigenvalues of $T_F$, and for that we need some basic facts
regarding the representation theory of $S_n$. We will then study the fraction of edges between any
two global functions $\mathcal{A},\mathcal{B}$, and Theorem~\ref{thm:global_prod_free} will just be
the special case that $\mathcal{A} = \mathcal{B} = F$.

Throughout this section, we set $\delta = \frac{\card{F}}{\card{S_n}}$.

\subsubsection{Basic facts about representation theory of $S_n$}
We will need some basic facts about the representation theory of $S_n$, and our exposition will follow standard textbooks,
e.g.~\cite{fulton2013representation}.

A partition of $[n]$, denoted by $\lambda\vdash n$, is a sequence of integers
$\lambda = (\lambda_1,\ldots,\lambda_k)$ where $\lambda_1\geq\lambda_2\geq\dots\geq\lambda_k\geq 1$
sum up to $n$. It is well-known that partitions index equivalence classes of representations of
$S_n$, thus we may associate with each partition $\lambda$ a character $\chi_{\lambda}\colon S_n \to \mathbb{C}$,
which in the case of the symmetric group is real-valued. The dimension of $\lambda$ is $\mathsf{dim}(\lambda) = \chi_\lambda(e)$, where $e$ is the identity permutation.

Given a partition $\lambda$, a $\lambda$-tabloid is a partition of $[n]$ into sets $A_1,\ldots,A_k$
such that $\card{A_i} = \lambda_i$. Thus, for $\lambda$-tabloids $A = (A_1,\ldots,A_k)$ and $B = (B_1,\ldots,B_k)$,
we define $T_{A,B} = \sett{\pi\in S_n}{\pi(A_i) = B_i ~\forall i=1,\ldots,k}$,
and refer to any such $T_{A,B}$ as a $\lambda$-coset.

With these notations, we may define the space $V_{\lambda}(S_n)$, which is the linear span of
the indicator functions of all $\lambda$-cosets. We note that
$V_{\lambda}(S_n)$ is clearly a left $S_n$-module, where the action of $S_n$ is given as
$\prescript{\pi}{} f\colon S_n\to \mathbb{R}$ defined by $\prescript{\pi}{} f(\sigma) = f(\pi \sigma)$.

 Next, we need to define an ordering on partitions that will let us further refine the spaces $V_{\lambda}$.
\begin{defn}
  Let $\lambda = (\lambda_1,\ldots,\lambda_k)$, $\mu = (\mu_1,\ldots,\mu_{s})$ be partitions
  of $[n]$. We say that $\lambda$ dominates $\mu$, and denote $\lambda \trianglerighteq \mu$, if for all
  $j=1,\ldots,k$ it holds that $\sum\limits_{i=1}^{j} \lambda_i\geq \sum\limits_{i=1}^{j} \mu_i$.
\end{defn}

With this definition, one may easily show that $V_{\mu}\subseteq V_{\lambda}$ whenever $\mu\trianglerighteq \lambda$,
and furthermore that $V_{\mu} = V_{\lambda}$ if and only if $\mu = \lambda$. It thus makes sense to define the spaces
\[
V_{=\lambda} = V_{\lambda}\cap\bigcap_{\mu\vartriangleright \lambda} V_{\mu}^{\perp}.
\]
%
The spaces $V_{=\lambda}$ are orthogonal and their direct sum is $\set{f\colon S_n\to\mathbb{R}}$, so we may write any function
$f\colon S_n\to\mathbb{R}$ as $f = \sum\limits_{\lambda\vdash n}{ f^{=\lambda}}$ in a unique way.

\begin{defn}
  Let $\lambda = (\lambda_1,\ldots,\lambda_k)$ be a partition of $n$. The transpose partition, $\lambda^t$,
  is $(\mu_1,\ldots,\mu_{k'})$, where $k' = \lambda_1$ and $\mu_j = \card{\sett{i}{\lambda_i\geq j}}$.
\end{defn}

Alternatively, if we think of a partition as represented by top-left justified rows, then the transpose of a partition is obtained by reflecting the diagram across the main diagonal. For example, $(3,1)^t = (2,1,1)$:
\begin{align*}
(3,1) = \ydiagram{3,1} && (2,1,1) = \ydiagram{2,1,1}
\end{align*}

There are two partitions that are very easy
to understand: $\lambda = (n)$, and its transpose, $\lambda = (1^t)$. For $\lambda = (n)$, the space $V_{=\lambda}$ consists of constant functions, and one has $\chi_{\lambda} = 1$. Thus, $f^{=(n)}$ is just the average of $f$, i.e.\ $\mu(f) \defeq \Expect{\pi}{f(\pi)}$.
For $\lambda = (1^n)$, the space $V_{=\lambda}$ consists of multiples of the sign function of permutations,
${\sf sign}\colon S_n\to\set{-1,1}$, and $\chi_{\lambda} = {\sf sign}$. One therefore has $f^{=\lambda} = \inner{f}{{\sf sign}} {\sf sign}(f)$.

For general partitions $\lambda$, it is well-known that the dimensions of $\lambda$ and $\lambda^t$ are equal, and one has that
$\chi_{\lambda^{t}} = {\sf sign}\cdot \chi_{\lambda}$. We will need the following statement that generalizes this correspondence to
$f^{=\lambda}$ and $f^{=\lambda^t}$.

\begin{lem}\label{lem:multiplying_by_parity}
  Let $f\colon S_n\to\mathbb{R}$, and let $\lambda\vdash n$. Then
  $(f\cdot {\sf sign})^{=\lambda} = f^{=\lambda^t}{\sf sign}$.
\end{lem}
\begin{proof}
  The statement follows directly from the inversion formula for $f^{=\lambda}$, which states that
  $f^{=\lambda}(\pi) = {\sf dim}(\lambda)\Expect{\sigma\in S_n}{f(\sigma)\chi_{\lambda}(\pi\sigma^{-1})}$.
  By change of variables, we see that
  \[
  (f\cdot {\sf sign})^{=\lambda}(\pi)
  ={\sf dim}(\lambda)\Expect{\sigma\in S_n}{f(\sigma^{-1}\pi){\sf sign}(\sigma^{-1}\pi)\chi_{\lambda}(\sigma)}
  ={\sf sign}(\pi){\sf dim}(\lambda)\Expect{\sigma\in S_n}{f(\sigma^{-1}\pi){\sf sign}(\sigma)\chi_{\lambda}(\sigma)},
  \]
  where we used the fact that ${\sf sign}$ is multiplicative and ${\sf sign}(\sigma^{-1}) = {\sf sign}(\sigma)$. Now, as
  ${\sf sign}(\sigma)\chi_{\lambda}(\sigma) = \chi_{\lambda^t}(\sigma)$, we get by changing variables again that
  \[
  (f\cdot {\sf sign})^{=\lambda}(\pi)=
  {\sf sign}(\pi){\sf dim}(\lambda)\Expect{\sigma\in S_n}{f(\sigma)\chi_{\lambda^t}(\pi\sigma^{-1})}
  ={\sf sign}(\pi){\sf dim}(\lambda^t)\Expect{\sigma\in S_n}{f(\sigma)\chi_{\lambda^t}(\pi\sigma^{-1})},
  \]
  which is equal to ${\sf sign}(\pi) f^{=\lambda^t}(\pi)$ by the inversion formula.
\end{proof}

Lastly, we remark that if $\lambda$ is a partition such that $\lambda = n-k$, then $V_{=\lambda}\subseteq V_k$. It follows
by Parseval that
\begin{equation}\label{eq:parseval_triviality}
\sum\limits_{\substack{\lambda\vdash n\\\lambda_1 = n-k}}\norm{f^{=\lambda}}_2^2\leq \norm{f^{\leq k}}_2^2.
\end{equation}
\subsubsection{The eigenvalues of $T_F^{*}T_F$}
\begin{claim}\label{claim:invariant}
  For all $\lambda\vdash n$ we have that $T_F V_{=\lambda}\subseteq V_{=\lambda}$; the same holds for
  $T_F^{*}$.
\end{claim}
\begin{proof}
  First, we show that $T_F V_{\lambda}\subseteq V_{\lambda}$, and for that it is enough to show that $T_F 1_{T_{A,B}}\in V_{\lambda}$
  for all $\lambda$-tabloids $A = (A_1,\ldots,A_k)$ and
  $B = (B_1,\ldots,B_k)$. Fix $a\in F$, and note that
  $1_{T_{A,B}}(\sigma a) = 1_{T_{a(A),B}}(\sigma)$ where $a(A) = (a(A_1),\ldots,a(A_k))$,
  so $1_{T_{A,B}}(\sigma a)$, as a function of $\sigma$, is also an indicator of a $\lambda$-coset. Since
  $T_F 1_{T_{A,B}}$ is a linear combination of such functions, it follows that $T_F 1_{T_{A,B}}\in V_{\lambda}$.
  A similar argument shows that the same holds for the adjoint operator of $T_F^{*} = T_{F^{-1}}$,  where $F^{-1} = \sett{a^{-1}}{a\in F}$.

  Thus, for $f\in V_{=\lambda}$ we automatically have that
  $f\in V_{\lambda}$, and we next show orthogonality to $V_{\mu}$ for all $\mu\triangleright \lambda$.
  Indeed, let $\mu$ be such partition and let $g\in V_{\mu}$; then by the above $T_F^{*} g\in V_{\mu}$ and so
  $\inner{T_F f}{g} = \inner{f}{T_F^{*} g} = 0$,
  and the proof is complete. The argument for $T_F^{*}$ is analogous.
\end{proof}
Thus, we may find a basis of each $V_{=\lambda}$ consisting of eigenvectors of $T_F^{*} T_F$.
The following claim shows that the multiplicity
of each corresponding eigenvalue is at least ${\sf dim}(\lambda)$.

\begin{claim}\label{claim:rotate_ev}
  Let $f\in V_{=\lambda}(S_n)$ be non-zero. Then ${\sf dim}({\sf Span}(\set{^{\pi} f}_{\pi\in S_n}))\geq {\sf dim}(\lambda)$.
\end{claim}
\begin{proof}

Let $\rho_{\lambda}\colon S_n \to V_{=\lambda}$ be a representation, and denote by $W$ the span of $\set{\prescript{\pi}{} f}_{\pi\in S_n}$.
Note that $W$ is a subspace of $V_{=\lambda}$, and it holds that $(\rho|_{W},W)$ is a sub-representation of $\rho$. Since each
irreducible representation $V\subseteq V_{=\lambda}$  of $S_n$ has dimension ${\sf dim}(\lambda)$, it follows that
${\sf dim}(W)\geq {\sf dim}(\lambda)$, and we're done.
\end{proof}

We can thus use the trace method to bound the magnitude of each eigenvalue.
\begin{lem}\label{lem:ev_up_TF}
Let $f\in V_{=\lambda}$ be an eigenvector of $T_F^{*} T_F$ with eigenvalue $\alpha_{\lambda}$. Then
\[
\alpha_{\lambda}\leq \frac{1}{{\sf dim}(\lambda)\delta}.
\]
\end{lem}
\begin{proof}
By Claim~\ref{claim:rotate_ev}, we may find a collection of ${\sf dim}(\lambda)$ permutations, call it
$\Pi$, such that $\set{\prescript{\pi}{} f}_{\pi\in \Pi}$ is linearly independent. Since $f$ is an eigenvector of $T_F^{*} T_F$,
it follows that each one of $\prescript{\pi}{} f$ is an eigenvector with eigenvalue $\alpha_{\lambda}$. It follows that
${\sf Tr}(T_{F}^{*}T_F) \geq \card{\Pi} \alpha_{\lambda} = {\sf dim}(\lambda)\alpha_{\lambda}$.

On the other hand, interpreting ${\sf Tr}(T_{F}^{*} T_F)$ probabilistically as the probability to return to the starting
vertex in $2$-steps,
\[
{\sf Tr}(T_{F}^{*} T_F)
=\sum\limits_{\pi}
\Prob{a_1\in F^{-1},a_2\in F}{\pi = \pi a_1 a_2}
=n!\Prob{a_1\in F^{-1},a_2\in F}{a_2 = a_{1}^{-1}}
= n! \frac{1}{\card{F}}
=\frac{1}{\delta}.
\]
Combining the two bounds on ${\sf Tr}(T_{F}^{*} T_F)$ completes the proof.
\end{proof}

To use this lemma effectively, we have the following bound on ${\sf dim}(\lambda)$ that follows from the hook length formula.
\begin{lem}[Claim 1, Theorem 19 in~\cite{EFP}]\label{lem:dim_lb}
  Let $\lambda\vdash n$ be given as $\lambda = (\lambda_1,\ldots,\lambda_k)$,
  and denote $d = \min(n-\lambda_1,k)$.
  \begin{enumerate}
    \item If $\lambda = (n)$, then ${\sf dim}(\lambda) = 1$.
    \item If $d>0$, then ${\sf dim}(\lambda)\geq \left(\frac{n}{d\cdot e}\right)^d$.
    \item If $d > n/10$, then ${\sf dim}(\lambda)\geq 1.05^n$.
  \end{enumerate}
\end{lem}

\subsubsection{Applying Hoffman's bound}
With the information we have gathered regarding the representation theory of $S_n$ and
the eigenvalues of $T_F$, we can use the spectral method to prove lower bounds on
$\inner{T_F g}{h}$ for Boolean functions $g,h$ that are global, as in the following lemma.
\begin{lem}\label{lem:Hoffman}
  There exists $C>0$ such that the following holds.
  Let $n\in\mathbb{N}$ and $\eps>0$ be such that $n\geq \log(1/\eps)^C$,
  and suppose that $g,h\colon A_n\to\power{}$ are $(6,\eps)$-global. Then
  \[
  \inner{T_F g}{h} \geq \frac{\E[g]\E[h]}{4} - C\frac{\eps^4 \log^{C}(1/\eps)}{\sqrt{n\delta}} - \frac{C}{\sqrt{n^4\delta}}\sqrt{\E[g]\E[h]}.
  \]
\end{lem}
\begin{proof}
  Extend $g,h$ to $S_n$ by defining them to be $0$ outside $A_n$.

  Recall that $T_F$ preserves each $V_{=\lambda}$. Decomposing
  $g = \sum_{\lambda\vdash n} g^{=\lambda}$ where $g^{=\lambda}\in V_{=\lambda}$
  and $h$ similarly, we have by Plancherel that
  $\inner{T_F g}{h}
  =\sum\limits_{\lambda,\theta} \inner{T_F  g^{=\lambda}}{h^{=\lambda}}$.
  For the trivial partition $\lambda = (n)$ we have
  that $g^{=\lambda} \equiv\mu(g)=\E[g]/2$, $h^{=\lambda} \equiv\mu(h)=\E[h]/2$. For $\lambda = (1^n)$,
  since $F\subseteq A_n$ it follows that $T_F {\sf sign} = {\sf sign}$, and so
  $T_F g^{=\lambda} = \beta_{\lambda}{\sf sign}$,
  $h^{=\lambda} = \gamma_{\lambda}{\sf sign}$ for $\beta_{\lambda},\gamma_{\lambda}\geq 0$, so the term corresponding to
  $\lambda$ in the above is non-negative. Thus, denoting $\lambda = (\lambda_1,\ldots,\lambda_k)$ we have that
  \begin{equation}\label{eq:Hoffman}
   \inner{T_F g}{h}\geq \mu(g)\mu(h) -
   \sum\limits_{\substack{\lambda\vdash n\\\lambda\neq (n),(1^n)\\ \lambda_1\geq n-3\text{ or }k\geq n-3}}
   \norm{T_F g^{=\lambda}}_{2}\norm{h^{=\lambda}}_2
   -
   \sum\limits_{\substack{\lambda\neq (n),(1^n)\\ \lambda_1\leq n-4\text{ and }k\leq n-4}} \norm{T_F g^{=\lambda}}_{2}\norm{h^{=\lambda}}_2.
  \end{equation}
  We upper-bound the second and third terms on the right-hand side, from which the lemma follows.
  We begin with the second term, and handle separately $\lambda$'s such that $\lambda_1\geq n-3$, and $\lambda$'s such that
  $k\geq n-3$.

  \paragraph{$\lambda$'s such that $\lambda\neq (n), (1^n)$ and $\lambda_1\geq n-3$.}
  We first upper bound $\norm{T_F g^{=\lambda}}_{2}$.
  As $T_{F}^{*} T_F$ preserves each space $V_{=\lambda}$ and is symmetric, we may write this space
  as a sum of eigenspaces of $T_{F}^{*} T_F$, say $\bigoplus_{\theta} V_{=\lambda}^{\theta}$.
  Writing $g^{=\lambda} = \sum_{\theta} g^{=\lambda,\theta}$ where $g^{=\lambda,\theta}\in V_{=\lambda}^{\theta}$,
  we have that
  \[
  \norm{T_F g^{=\lambda}}_{2}^2
  =\inner{g^{=\lambda}}{T_F^{*}T_F g^{=\lambda}}
  =\sum\limits_{\theta}{\inner{g^{=\lambda,\theta}}{T_F^{*}T_F g^{=\lambda,\theta}}}
  =\sum\limits_{\theta}{\theta\norm{g^{=\lambda,\theta}}_2^2}.
  \]
  By Lemma~\ref{lem:ev_up_TF} we have $\theta\leq \frac{1}{{\sf dim}(\lambda)\delta}$,
  which by Fact~\ref{lem:dim_lb} is at most $O\left(\frac{1}{n\delta}\right)$. We thus get that
  \[
  \norm{T_F g^{=\lambda}}_{2}^2
  \leq O\left(\frac{1}{n\delta}\right)\sum\limits_{\theta}{\norm{g^{=\lambda,\theta}}_2^2}
  \leq O\left(\frac{1}{n\delta}\right)\norm{g^{=\lambda}}_2^2.
  \]
  Plugging this into the second sum in~\eqref{eq:Hoffman}, we get that the contribution from $\lambda$ such that $\lambda_1\geq n-3$ is
  at most
  \[
  O\left(\frac{1}{\sqrt{n\delta}}\right)
  \sum\limits_{\substack{\lambda\vdash n\\\lambda\neq (n),(1^n)\\ \lambda_1\geq n-3}}
   \norm{g^{=\lambda}}_{2}\norm{h^{=\lambda}}_2
  \leq O\left(\frac{1}{\sqrt{n\delta}}\right) \norm{g^{\leq 3}}_2\norm{h^{\leq 3}}_2,
  \]
  where we used Cauchy-Schwarz and~\eqref{eq:parseval_triviality}.
  By Theorem~\ref{thm:lvl_d}, $\norm{g^{\leq 3}}_2^2,\norm{h^{\leq 3}}_2^2\leq C\cdot \eps^4 \log^{C}(1/\eps)$ for some absolute constant $C$.
  We thus get that
  \[
   \sum\limits_{\substack{\lambda\vdash n\\\lambda\neq (n),(1^n)\\ \lambda_1\geq n-3}}
   \norm{T_F g^{=\lambda}}_{2}\norm{h^{=\lambda}}_2\leq \frac{1}{\sqrt{n\delta}}C'\cdot \eps^4 \log^{C}(1/\eps).
  \]

  \paragraph{$\lambda$'s such that $k\geq n-3$.}
  The treatment here is pretty much identical to the previous case, except that we look at the functions $\tilde{g} = g\cdot {\sf sign}$
  and $\tilde{h} = h\cdot {\sf sign}$. That is, first note that the globalness of $g,h$ implies that $\tilde{g},\tilde{h}$ are also global with the same parameters,
  and since $g,h$ are Boolean, $\tilde{g},\tilde{h}$ are integer valued. Moreover, by Lemma~\ref{lem:multiplying_by_parity} we have that
  \[
  \sum\limits_{\substack{\lambda\vdash n\\\lambda\neq (n),(1^n)\\ k\geq n-3}}
  \norm{T_F g^{=\lambda}}_{2}\norm{h^{=\lambda}}_{2}
  =
  \sum\limits_{\substack{\lambda\vdash n\\\lambda\neq (n),(1^n)\\ k\geq n-3}}
  \norm{T_F \tilde{g}^{=\lambda^t}}_{2}\norm{\tilde{h}^{=\lambda^t}}_{2}
  =
  \sum\limits_{\substack{\lambda\vdash n\\\lambda\neq (n),(1^n)\\ \lambda_1\geq n-3}}
  \norm{T_F \tilde{g}^{\lambda}}_{2}\norm{\tilde{h}^{\lambda}}_{2},
  \]
  and from here the argument is identical.

  \paragraph{Bounding the third term in~\eqref{eq:Hoffman}.}
  Repeating the eigenspace argument from above, for all $\lambda\vdash n$ such that
  $\lambda_1\leq n-4$ and $k\leq n-4$ we have
  \[
  \norm{T_F g^{=\lambda}}_{2}\leq O\left(\frac{1}{\sqrt{n^4\delta}}\right) \norm{g^{=\lambda}}_{2}.
  \]
  Thus, the third sum in~\eqref{eq:Hoffman} is at most
  \[
  O\left(\frac{1}{\sqrt{n^4\delta}}\right)\sum\limits_{\lambda\vdash n}\norm{g^{=\lambda}}_{2}\norm{h^{=\lambda}}_{2}
  \leq  O\left(\frac{1}{\sqrt{n^4\delta}}\right)\norm{g}_2\norm{h}_2,
  \]
  where we used Cauchy--Schwarz and Parseval.
\end{proof}

We can now prove the strengthening of Theorem~\ref{thm:global_prod_free}, stated below.
\begin{cor}\label{cor:global_prod_free_strong}
  There exists $K\in\mathbb{N}$ such that the following holds for all $\eps>0$ and $n\geq \log^K(1/\eps)$.
  If $\mathcal{A},\mathcal{B}\subseteq A_n$
  are $(6,\eps)$-global, and $\mu(\mathcal{A})\mu(\mathcal{B})\geq K\max(n^{-4}\delta^{-1}, (n\delta)^{-1/2}\eps^{4}\log^{K}(1/\eps))$,
  then
  \[
  \inner{T_F g}{h}\geq \frac{1}{5}\mu(\mathcal{A})\mu(\mathcal{B}).
  \]
\end{cor}
\begin{proof}
Taking
$g = 1_{\mathcal{A}}$, $h = 1_{\mathcal{B}}$, by Lemma~\ref{lem:Hoffman} we have
\[
\inner{T_F g}{h}
\geq \frac{1}{4} \mu(\mathcal{A})\mu(\mathcal{B}) - C'\frac{\eps^4\log^{C'}(1/\eps)}{\sqrt{n\delta}} - \frac{C'}{n^2\sqrt{\delta}}\sqrt{\mu(\mathcal{A})\mu(\mathcal{B})},
\]
where $C'$ is an absolute constants. Now the conditions on the parameters implies that the first term dominates the other two.
\end{proof}
We note that Theorem~\ref{thm:global_prod_free} immediately follows, since there one has $g = h = 1_F$ and $\inner{T_F g}{h} = 0$,
so one gets that the condition on the parameters fail, and therefore the lower bound on $\mu(\mathcal{A})\mu(\mathcal{B})$ (which in
this case is just $\delta^2$) fails; plugging in $\eps = C\cdot \sqrt{\delta}$ and rearranging finishes the proof.

\subsubsection{Improving on Theorem~\ref{thm:global_prod_free}?}\label{sec:improve_prod_free}
We remark that it is within reason to expect that global, product-free families in $A_n$ must in fact be much smaller.
To be more precise, one may expect that for all $t\in\mathbb{N}$, there is $j\in\mathbb{N}$ such that for
$n\geq n_0(t)$, if $F$ is $(j,O(\sqrt{\delta}))$-global (where $\delta = \card{F}/\card{S_n}$), then $\delta\leq O_t(n^{-t})$.
The bottleneck in our approach comes from the use of the trace method (which doesn't use the globalness of $F$ at all),
and the bounds it gives on the eigenvalues of $T_F^{*} T_F$ corresponding to low-degree functions: they become meaningless as soon as $\delta\geq 1/n$.

Inspecting the above proof, our approach only requires a super-logarithmic upper bound on the eigenvalues to go through.
More precisely, we need that the first few non-trivial eigenvalues of $T_F^{*} T_F$ are at most
$(\log n)^{-K(t)}$, for sufficiently large $K(t)$. We feel that something like that should follow
in greater generality from the fact that the set of generators in the Cayley graph, namely $F$, is global. To support that,
note that if we were dealing with Abelian groups, then the eigenvalue $\alpha$ of $T_F$ corresponding to a
character $\chi$ could be computed as
$\lambda = \frac{1}{\card{F}}\sum\limits_{a\in F}{\chi(a)}$, which by rewriting
is nothing but a (normalized) Fourier coefficient of $F$, i.e. $\frac{1}{\delta} \widehat{1_F}(\chi)$,
which we expect to be small by the globalness of $F$.

\subsection{Isoperimetric inequalities in the transpositions Cayley graph} \label{sec:isoperimetric}
In this section, we consider $\mathrm{T}$ which is the adjacency operator of
the transpositions graph. That is, it is the transition matrix of the (left) Cayley
graph $(S_n,A)$, where $A$ is the set of transpositions (and the multiplication happens
from the left). We show that for a
global set $S$, starting a walk from a vertex in $S$ and performing
$\approx c n$ steps according to $\mathrm{T}$ escapes $S$ with probability
close to $1$.

\paragraph{Poisson process random walk.}
To be more precise, we consider the following random walk: from a permutation $\pi\in S$,
choose a number $k\sim {\sf Poisson}(t)$, take $\tau$ which is a product of $k$ random transpositions,
and go to $\sigma = \tau\circ \pi$. We show that starting with a random $\pi \in S$, the probability that
we escape $S$, i.e.\ that $S\sigma\not\in S$, is close to $1$.

To prove this result, we first note that the distribution of an outgoing neighbour from $\pi$ is exactly
$e^{-t(I-\mathrm{T})} 1_{\pi}$, where $1_{\pi}$ is the indicator vector of $\pi$. Therefore, the distribution
of $\sigma$ where $\pi\in S$ is random is $e^{-t(I-\mathrm{T})} \frac{1_S}{\card{S}}$, where $1_S$ is the indicator
vector of $S$. Thus, the probability that $\sigma$ is in $S$ (i.e.\ of the complement event) is
\[
\frac{1}{\mu(S)} \inner{1_S}{ e^{-t(I-\mathrm{T})} 1_S},
\]
where $\mu(S)$ is the measure of $S$. We upper-bound this quantity using spectral considerations.
We will only need our hypercontractive inequality and basic knowledge of the eigenvalues of
$\mathrm{T}$, which can be found, for example, in~\cite[Corollary 21]{FOW}. This is the content
of the firs three items in the lemma below (we also prove a fourth item, which will be useful for us later on).

\begin{lem}\label{lem:FOW}
  Let $\lambda\in \mathbb{R}$ be an eigenvalue of $\mathrm{T}$, and $f\in V_{d}(S_n)$ be a corresponding eigenvector.
  \begin{enumerate}
  \item $\mathrm{T} V_{=d}(S_n) \subseteq V_{=d}(S_n)$.
  \item $1-\frac{2d}{n-1}\leq \lambda\leq 1-\frac{d}{n-1}$.
  \item If $d\leq n/2$, then we have the stronger bound $1-\frac{2d}{n-1}\leq \lambda\leq 1-\left(1-\frac{d-1}{n}\right)\frac{2d}{n-1}$.
  \item If $\mathrm{L}$ is a Laplacian of order $1$, then $\mathrm{L}$ and $\mathrm{T}$ commute. Thus, $\mathrm{T}$ commutes with all Laplacians.
\end{enumerate}
\end{lem}
\begin{proof}
  For the first item, we first note that $\mathrm{T}$ commutes with the right action of $S_n$ on functions:
  \[
  (\mathrm{T} (f^{\pi}))(\sigma)
  =\E_{\pi'\text{ a transposition}}\left[f^{\pi}(\pi'\circ \sigma)\right]
  =\E_{\pi'\text{ a transposition}}\left[f(\pi'\circ \sigma\circ \pi)\right]
  =\mathrm{T}f(\sigma\circ \pi)
  =(\mathrm{T}f)^{\pi}(\sigma).
  \]
  Also, $\mathrm{T}$ is self adjoint, so $\mathrm{T}^{*}$ also commutes with the action of $S_n$. The first item now follows as
  in the proof of Claim~\ref{claim:invariant}. 

  The second and third items are exactly~\cite[Corollary 21]{FOW}. For the last item, for any function $f$ and
  an order $1$ Laplacian $\mathrm{L} = \mathrm{L}_{(i,j)}$,
  \[
  \mathrm{T} \mathrm{L} f
  =\mathrm{T}\left(f - f^{(i,j)}\right)
  =\mathrm{T} f - \mathrm{T}\left(f^{(i,j)}\right)
  =\mathrm{T} f - \left(\mathrm{T}f\right)^{(i,j)}
  =\mathrm{L} \left(\mathrm{T} f\right),
  \]
  where in the third transition we used the fact that $\mathrm{T}$ commutes with the right action of $S_n$.
\end{proof}
We remark that the first item above implies that we may find a basis of the space of real-valued functions
consisting of eigenvectors of $\mathrm{T}$, where each function is from $V_{=d}(S_n)$ for some $d$.
Lastly, we need the following (straightforward) fact.
\begin{fact}\label{fact:exp_ev}
  If $f\in V_{d}(S_n)$ is an eigenvector of $\mathrm{T}$ with eigenvalue $\lambda$,
  then $f$ is an eigenvector of $e^{-t(I-\mathrm{T})}$ with eigenvalue $e^{-t(1-\lambda)}$.
\end{fact}

\begin{thm}\label{thm:SSE}
  There exists $C>0$ such that the following holds for all $d\in\mathbb{N}$, $t,\eps>0$ and
  $n\in\mathbb{N}$ such that $n\geq 2^{C\cdot d^3} \log^{C\cdot d}(1/\eps)$.
  If $S\subseteq S_n$ is a set of vertices such that $1_S$ is $(2d,\eps)$-global, then
  \[
  \Prob{\substack{\pi\in S \\ \sigma\sim e^{-t(I-T)} \pi}}{\sigma\not\in S}\geq
  1-\left(2^{C\cdot d^4}\eps\log^{C\cdot d}(1/\eps) + e^{-\frac{(d+1)t}{n-1}}\right).
  \]
\end{thm}
\begin{proof}
  Consider the complement event that $\sigma\in S$,
  and note that the desired probability can be written analytically as
  $\frac{1}{\mu(S)} \inner{1_S}{ e^{-t(I-\mathrm{T})} 1_S}$, where $\mu(S)$ is the measure of $S$.
  Now, writing $f = 1_S$ and expanding $f = f_{=0} + f_{=1}+\cdots$, we consider
  each one of $e^{-t(I-\mathrm{T})} f^{=j}$ separately. We claim that
  \begin{equation}\label{eq:isop}
  \norm{e^{-t(I-\mathrm{T})} f^{=j}}_2\leq e^{-\frac{jt}{n-1}} \norm{f^{=j}}_2.
  \end{equation}
  Indeed, note that we may write $f^{=j} = \sum\limits_{}{a_r f_{j,r}}$, where $f_{j,r}\in V_{=j}(S_n)$ are orthogonal
  and eigenvectors of $\mathrm{T}$ with eigenvalue $\lambda_{j,r}$, and so by Fact~\ref{fact:exp_ev},
  $e^{-t(I-\mathrm{T})} f^{=j}
  =\sum\limits_{r} e^{-t(1-\lambda_{j,r})} f_{j,r}$.
  By Parseval we deduce that
  \[
  \norm{e^{-t(I-\mathrm{T})} f^{=j}}_2^2
  \leq \sum\limits_{r} e^{-t(1-\lambda_{j,r})} \norm{f_{j,r}}_2^2
  \leq \max_{r}e^{-t(1-\lambda_{j,r})}\sum\limits_{r}  \norm{f_{j,r}}_2^2
  =\max_{r}e^{-t(1-\lambda_{j,r})}\norm{f^{=j}}_2^2.
  \]
  Inequality~\eqref{eq:isop} now follows from the second item in Lemma~\ref{lem:FOW}.

  We now expand out the expression we have for the probability of the complement event using Plancherel:
  \begin{align}
  \frac{1}{\mu(S)} \inner{1_S}{ e^{-t(I-\mathrm{T})} 1_S}
  =\frac{1}{\mu(S)}\sum\limits_{j} \inner{f^{=j}}{e^{-t(I-\mathrm{T})} f^{=j}}
  &\leq \frac{1}{\mu(S)}\sum\limits_{j} \norm{f^{=j}}_2\norm{e^{-t(I-\mathrm{T})} f^{=j}}_2 \notag\\\label{eq:isop_2}
  &\leq\frac{1}{\mu(S)}\sum\limits_{j} e^{-\frac{jt}{n-1}}\norm{f^{=j}}_2^2,
  \end{align}
  where in the last two transitions we used Cauchy--Schwarz and inequality~\eqref{eq:isop}.
  Lastly, we bound $\norm{f^{=j}}_2^2$. For $j > d$ we have that $\sum\limits_{j > d}\norm{f^{=j}}_2^2\leq \mu(S)$ by Parseval,
  and for $j\leq d$ we use hypercontractivity.

  First, bound $\norm{f^{=j}}_2\leq \norm{f^{\leq j}}_2$, and note
  that the function $f^{\leq j}$ is $(2j, 2^{O(j^4)}\eps^2 \log^{O(j)}(1/\eps))$-global by Claim~\ref{claim:globalness_of_low_deg_part}.
  Thus, using H\"{o}lder's inequality and Theorem~\ref{thm:Reasonability} we get that
  \begin{align*}
  \norm{f^{\leq j}}_2^2
  =\inner{f}{f^{\leq j}}
  \leq \norm{f}_{4/3} \norm{f^{\leq j}}_4
  \leq \mu(S)^{3/4} 2^{O(j^3)}\sqrt{2^{O(j^4)}\eps^2 \log^{O(j)}(1/\eps)}\norm{f^{\leq j}}_2^{1/2}.
  \end{align*}
  Rearranging gives $\norm{f^{\leq j}}_2^2 \leq 2^{O(j^4)} \mu(S) \eps\log^{O(j)}(1/\eps)$.

  Plugging our estimates into~\eqref{eq:isop_2} we get
  \begin{align*}
  \frac{1}{\mu(S)} \inner{1_S}{ e^{-t(I-\mathrm{T})} 1_S}
  \leq
  \sum\limits_{j=0}^{d} 2^{O(j^4)}e^{-\frac{jt}{n-1}} \eps\log^{O(j)}(1/\eps)
  +e^{-\frac{(d+1)t}{n-1}}
  \leq 2^{O(d^4)}\eps\log^{O(d)}(1/\eps) + e^{-\frac{(d+1)t}{n-1}}.
  \end{align*}
  \end{proof}

  Using exactly the same technique, one can prove a lower bound on the probability of escaping a global
  set in a single step, as stated below. This result is similar in spirit to a variant of the KKL Theorem
  over the Boolean hypercube~\cite{KKL}, and therefore we modify the formulation slightly. Given a function
  $f\colon S_n\to \mathbb{R}$, we define the influence of coordinate $i\in[n]$ to be
  \[
  I_i[f] = \Expect{j\neq i}{\norm{L_{(i,j)} f}_2^2},
  \]
  and define the total influence of $f$ to be $I[f] = I_1[f]+\dots+I_n[f]$.
  \begin{thm}\label{thm:KKL_analog}
  There exists $C>0$ such that the following holds for all $d\in\mathbb{N}$ and
  $n\in\mathbb{N}$ such that $n\geq 2^{C\cdot d^3}$.
  Suppose $S\subseteq S_n$ is such that for all derivative operators $\mathrm{D}\neq I$ of order at most $d$,
  it holds that  $\norm{\mathrm{D}1_S}_2\leq 2^{-C\cdot d^4}$. Then
  \[
  I[1_S]\geq \frac{1}{4} d\cdot {\sf var}(1_S).
  \]
  \end{thm}
  \begin{proof}
    Deferred to Appendix~\ref{sec:KKL}.
  \end{proof}

\subsection{Deducing results for the multi-cube}\label{sec:deduce_other_domains}
Our hypercontractive inequalities also imply similar hypercontractive inequalities on different non-product domains. One example
from~\cite{BKM} is the domain of $2$-to-$1$ maps, i.e. $\sett{\pi\colon [2n]\to[n]}{\card{\pi^{-1}(i)} = 2~\forall i\in [n]}$.
A more general domain, which we consider below, is the multi-slice.

\begin{defn}
  Let $m,n\in\mathbb{N}$ such that $n \ge m$, and let $k_1,\ldots,k_m\in\mathbb{N}$ sum up to $n$.
  The multi-slice $\mathcal{U}_{k_1,\ldots,k_m}$ of dimension $n$ consists of all vectors $x\in [m]^n$
  that, for all $j\in[m]$, have exactly $k_j$ of their coordinates equal to $j$.

  We consider the multi-slice as a probability space with the uniform measure.
\end{defn}
In exactly the same way one defines the degree decomposition over $S_n$, one may consider the degree decomposition over
the mutli-slice. A function $f\colon \mathcal{U}_{k_1,\ldots,k_m}\to\mathbb{R}$ is said to be a $d$-junta if there are
$A\subseteq[n]$ of size at most $d$ and $g\colon [m]^d\to\mathbb{R}$ such that $f(x) = g(x_A)$. We then define the
space $V_d(\mathcal{U}_{k_1,\ldots,k_m})$ spanned by $d$-juntas. Also, one may analogously define globalness of functions over the
multi-slice. A $d$-restriction consists of a set $A\subseteq[n]$ of size $d$ and $\alpha\in [m]^A$, and the corresponding
restriction is the function $f_{A\rightarrow \alpha}(z) = f(x_A = \alpha, x_{\bar{A}} = z)$ (whose domain is a different multi-slice).

\begin{defn}
  We say $f\colon \mathcal{U}_{k_1,\ldots,k_m}\to\mathbb{R}$ is $(d,\eps)$-global if for any $d$-restriction $(A,\alpha)$
  it holds that $\norm{f_{A\rightarrow \alpha}}_2\leq \eps$.
\end{defn}

\subsubsection{Hypercontractivity}
Our hypercontractive inequality for the multi-slice reads as follows.
\begin{thm}\label{thm:hypercontract_multi}
  There exists an absolute constant $C>0$ such that the following holds.
  Let $d,q,n\in\mathbb{N}$ be such that $n\geq q^{C\cdot d^2}$, and let $f\in V_d(\mathcal{U}_{k_1,\ldots,k_m})$.
  If $f$ is $(2d,\eps)$-global, then
  \[
   \|f\|_{q}\le q^{O\left(d^{3}\right)}\epsilon^{\frac{q-2}{q}}\|f\|_{2}^{\frac{2}{q}}.
  \]
\end{thm}
\begin{proof}
  We construct a simple deterministic coupling $\mathcal{C}$ between $S_n$ and $\mathcal{U}_{k_1,\ldots,k_m}$.

  Fix a partition of $[n]$ into sets $K_1,\ldots,K_m$ such that $\card{K_j} = k_j$ for all $j$.
  Given a permutation $\pi$, we define $\mathcal{C}(\pi) = x$ as follows: for all $i\in[n]$, $j\in[m]$,
  we set $x_i = j$ if $\pi(i)\in K_j$. Define the mapping $M\colon L_2(\mathcal{U}_{k_1,\ldots,k_m})\to L_2(S_n)$
  that maps a function $h\colon \mathcal{U}_{k_1,\ldots,k_m}\to\mathbb{R}$ to $M h\colon S_n\to\mathbb{R}$ defined by
  $(M h)(\pi) = h(\mathcal{C}(\pi))$.

  Let $g = Mf$.
  We claim that $g$ has degree at most $d$ and is global. To see that $g\in V_d(S_n)$, it is enough to show that the
  mapping $f\rightarrow g$ is linear (which is clear), and maps a $d$-junta into a $d$-junta, which is also straightforward.
  To see that $g$ is global, let $T = \set{(i_1,r_1),\ldots,(i_\ell,r_\ell)}$ be consistent, and define the $r$-restriction
  $(A,\alpha)$ as: $A = \set{i_1,\ldots,i_{\ell}}$, and $\alpha_{i_s} = j$ if $r_s\in K_j$. Note that the distribution
  of $x\in\mathcal{U}_{k_1,\ldots,k_m}$ conditioned on $x_A$ is exactly the same as of $\mathcal{C}(\pi)$ conditioned on
  $\pi$ respecting $T$, so if $r\leq 2d$ we get that
  \[
  \norm{g_{A\rightarrow \alpha}}_2 = \norm{f_{\rightarrow T}}_2\leq \eps,
  \]
  and $g$ is $(2d,\eps)$-global. The result thus follows from Theorem~\ref{thm:Reasonability} and the fact that $M$ preserves $L_p$ norms
  for all $p\geq 1$.
\end{proof}

The coupling in the proof of Theorem~\ref{thm:hypercontract_multi} also implies
in the same way a level-$d$ inequality over $\mathcal{U}_{k_1,\ldots,k_m}$ from the corresponding result in $S_n$, Theorem~\ref{thm:lvl_d},
as well as isoperimetric inequalities, as we describe next.

\subsubsection{Level-$d$ inequality}
As on $S_n$, for $f\colon\mathcal{U}_{k_1,\ldots,k_m}\to\mathbb{R}$ we let $f^{\leq d}$ be the projection of $f$
onto $V_d(\mathcal{U}_{k_1,\ldots,k_m})$. Our level-$d$ inequality for the multi-slice thus reads:
\begin{cor}\label{cor:lvl_d_multi}
  There exists an absolute constant $C>0$ such that the following holds.
Let $d,n\in\mathbb{N}$ and $\eps>0$ such that $n\geq 2^{C d^3} \log(1/\eps)^{C d}$.
If $f\colon \mathcal{U}_{k_1,\ldots,k_m}\to \{0,1\}$ is $(2d,\eps)$-global, then
$\norm{f^{\leq d}}_2^2\leq 2^{C \cdot d^4}\eps^4 \log^{C\cdot d}(1/\eps)$.
\end{cor}
\begin{proof}
  The proof relies on an additional easy property of the mapping $M$ from the proof of Theorem~\ref{thm:hypercontract_multi}.
  As in $S_n$, we define the space of pure degree $d$ functions over $\mathcal{U}_{k_1,\ldots,k_m}$ as
  $V_{=d}(\mathcal{U}_{k_1,\ldots,k_m}) = V_{d}(\mathcal{U}_{k_1,\ldots,k_m})\cap V_{d-1}(\mathcal{U}_{k_1,\ldots,k_m})^{\bot}$,
  and let $f^{=d}$ be the projection of $f$ onto $V_{=d}(\mathcal{U}_{k_1,\ldots,k_m})$. We thus have
  $f^{\leq d} = f^{=0}+f^{=1}+\dots+f^{=d}$, and so $f^{=d} = f^{\leq d} - f^{\leq d-1}$.

  Write $h_i = M f^{=i}$, and note that $h_i$ is of degree at most $i$. Also, we note that as restrictions of size $r<i$ over $S_n$
  are mapped to restrictions of size $r$ over $\mathcal{U}_{k_1,\ldots,k_m}$, it follows that $h_i$ is perpendicular to degree $i-1$
  functions, and so $h_i\in V_{=i}(S_n)$. By linearity of $M$, $M f = h_0+h_1+\dots + h_n$, and by uniqueness of the pure degree
  decomposition, it follows that $h_i = (M f)^{=i}$. We therefore have that
  \[
  \norm{f^{\leq d}}_2^2
  = \sum\limits_{i\leq d} \norm{f^{=i}}_2^2
  = \sum\limits_{i\leq d} \norm{h_i}_2^2
  = \sum\limits_{i\leq d} \norm{(M f)^{=i}}_2^2
  =  \norm{(Mf)^{\leq d}}_2^2
  \leq 2^{C\cdot d^4}\eps^4\log^{C\cdot d}(1/\eps),
  \]
  where the last inequality is by Theorem~\ref{thm:lvl_d}.
\end{proof}

%
%

\subsubsection{Isoperimetric inequalities}\label{sec:isop_multi}
One can also deduce the obvious analogs of Theorems~\ref{thm:SSE},~\ref{thm:KKL_analog} for the multi-slice.
Since we use it for our final application, we include here the statement of the analog of Theorem~\ref{thm:KKL_analog}.

For $f\colon\mathcal{U}_{k_1,\ldots,k_m}\to\mathbb{R}$, consider the Laplacians
$\mathrm{L}_{i,j}$ that map a function $f$ to a function $L_{i,j} f$ defined as
$\mathrm{L}_{i,j} f(x) = f(x) - f(x^{(i,j)})$, and define
$I_i[f] = \Expect{j\neq i}{\norm{L_{i,j} f}_2^2}$
and $I[f] = \sum\limits_{i=1}^{n} I_i[f]$. Similarly to Definition~\ref{def:derivative}, we define
a derivative of $f$ as a restriction of the corresponding Laplacian, i.e.\ for $i,j\in [n]$, $a,b\in [m]$
we define $\mathrm{D}_{(i,j)\rightarrow (a,b)} f = (\mathrm{L}_{i,j} f(x))_{(i,j)\rightarrow (a,b)}$.
\begin{thm}\label{thm:KKL_multi_slice}
  There exists $C>0$ such that the following holds for all $d\in\mathbb{N}$ and
  $n\in\mathbb{N}$ such that $n\geq 2^{C\cdot d^3}$.
  Suppose $S\subseteq \mathcal{U}_{k_1,\ldots,k_m}$ such that for all derivative operators
  $\mathrm{D}\neq I$ of order at most $d$ it holds that  $\norm{\mathrm{D}1_S}_2\leq 2^{-C\cdot d^4}$. Then
  $I[1_S]\geq \frac{1}{4} d\cdot {\sf var}(1_S)$.
\end{thm}
We omit the straightforward derivation from Theorem~\ref{thm:KKL_analog}.

\subsection{Stability result for the Kruskal--Katona theorem on the slice}
Our final application is the following sharp threshold result for the slice, which can
be also seen as a stability version of the Kruskal--Katona theorem (see~\cite{OW,Keevash} for other, incomparable stability versions). For a
family of subsets $\mathcal{F} \subseteq{[n]\choose k}$, we denote $\mu(\mathcal{F}) = \card{\mathcal{F}}/{n\choose k}$.
and define the upper shadow of $\mathcal{F}$ as
\[
\mathcal{F}\uparrow = \sett{X\in{n\choose k+1}}{\exists A\subseteq X, A\in \mathcal{F}}.
\]
The Kruskal--Katona theorem is a basic result in combinatorics that gives a lower bound on the measure of the upper shadow
of a family $\mathcal{F}$ in terms of the measure of the family itself. Below we state a convenient, simplified version
of it due to Lov\'{a}sz', which uses the generalized binomial coefficients.
\begin{thm}\label{thm:KK}
  Let $\mathcal{F}\subseteq {[n]\choose k}$ and suppose that $\card{\mathcal{F}} = {n\choose x}$. Then
  $\card{\mathcal{F}\uparrow}\geq {n\choose x+1}$.
\end{thm}

In general, Theorem~\ref{thm:KK} is tight, as can be shown by considering ``subcubes'', i.e.\ families
of the form $\mathcal{H} = \sett{ X\in{[n]\choose k}}{X\supseteq A}$ for some $A\subseteq [n]$. This raises
the question of whether a stronger version of Theorem~\ref{thm:KK} holds for families that are ``far from
having a structure such as $\mathcal{H}$''. Alternatively, this question can be viewed as a stability version
of Theorem~\ref{thm:KK}: must a family for which Theorem~\ref{thm:KK} is almost tight be of a similar structure
to $\mathcal{H}$?

Below, we mainly consider the case that $k=o(n)$, and show in improved version
of Theorem~\ref{thm:KK} for families that are ``far from $\mathcal{H}$''. To formalize
this, we consider the notion of restrictions: for $A\subseteq I\subseteq [n]$, we define
\[
\mathcal{F}_{I\rightarrow A}
=\sett{X\subseteq [n]\setminus I}{ X\cup A\in \mathcal{F}},
\]
and also define its measure $\mu(\mathcal{F}_{I\rightarrow A})$ appropriately. We say a family
$\mathcal{F}$ is $(d,\eps)$-global if for any $\card{I}\leq d$ and $A\subseteq I$ it holds that
$\mu(\mathcal{F}_{I\rightarrow A})\leq \eps$.

\begin{thm}\label{thm:KK_stability}
  There exists $C>0$, such that the following holds for all $d,n\in\mathbb{N}$ such that $n\geq 2^{C\cdot d^4}$.
  Let $\mathcal{F}\subseteq {[n]\choose k}$, and suppose that $\mathcal{F}$ is $(d,2^{-C\cdot d^4})$-global. Then
  $\mu(\mathcal{F}\uparrow)\geq \left(1+\frac{d}{64k}\right) \mu(\mathcal{F})$.
\end{thm}

\begin{proof}
  Let $f = 1_\mathcal{F}$, $g = 1_{\mathcal{F}\uparrow}$, and consider the operator $M\colon {[n]\choose k} \to {[n]\choose k+1}$ that from
  a set $A\subseteq[n]$ of size $k$ moves to a random set of size $k+1$ containing it. We also consider $M$
  as an operator $M\colon L_2\left({[n]\choose k}\right)\rightarrow L_2\left({[n]\choose k+1}\right)$ defined
  as $M f(B)=\Expect{A\subseteq B}{f(A)}$ (this operator is sometimes known as the \emph{raising} or \emph{up} operator). Note that for all $B\in{[n]\choose k+1}$, it holds that
  $g(B) M f(B) = M f(B)$, and that the average of $M f$ is the same as the average of $f$, i.e.\ $\mu(\mathcal{F})$.
  Thus,
  \[
  \mu(\mathcal{F})^2 = \inner{g}{M f}^2\leq \norm{g}_2^2\norm{M f}_2^2 = \norm{g}_2^2 \inner{f}{M^{*} M f}.
  \]
  Using the fact that the $2$-norm of $g$ squared is the measure of $\mathcal{F}\uparrow$ and rearranging,
  we get that
  \begin{equation}\label{eq:KK_stab}
  \mu(\mathcal{F}\uparrow)\geq\frac{\mu(\mathcal{F})^2}{\inner{f}{M^{*} M f}}
  = \frac{\mu(\mathcal{F})^2}{\Prob{\substack{x\in_R{[n]\choose k}\\ y\sim MM^{*} x}}{x\in \mathcal{F},y\in \mathcal{F}}}.
  \end{equation}
  We next lower bound $\Prob{\substack{x\in_R{[n]\choose k}\\ y\sim MM^{*} x}}{x\in \mathcal{F},y\not\in \mathcal{F}}$,
  which will give us an upper bound on the denominator. Towards this end, we relate this probability to
  the total influence of $1_{\mathcal{F}}$ as defined in Section~\ref{sec:isop_multi}.
  Note that the distribution of $y$ conditioned on $x$ is:
  with probability $1/(k+1)$ we have $y=x$, and otherwise $y = x^{(i,j)}$, where $i,j$ are random coordinates such that
  $x_i \neq x_j$. Consider $z\sim \mathrm{T} x$, where $\mathrm{T}$ is the operator of applying a random transposition;
  the probability that it interchanges two coordinates $i,j$ such that $x_i\neq x_j$ is $k(n-k)/\binom{n}{2}$, and so we get
  \begin{multline*}
  \Prob{\substack{x\in_R{[n]\choose k}\\ y\sim MM^{*} x}}{x\in \mathcal{F},y\not\in \mathcal{F}}
  =\frac{k}{k+1}\frac{n(n-1)}{2k(n-k)}\Prob{\substack{x\in_R{[n]\choose k}\\ y\sim \mathrm{T} x}}{x\in \mathcal{F},y\not\in \mathcal{F}} \\
  =\frac{k}{k+1}\frac{n(n-1)}{2k(n-k)} \frac{1}{2} \Prob{\substack{x\in_R{[n]\choose k}\\ y\sim \mathrm{T} x}}{1_{\mathcal{F}}(x) \neq 1_{\mathcal{F}}(y)}
  =\frac{k}{k+1}\frac{n(n-1)}{2k(n-k)}\frac{1}{2n} I[1_{\mathcal{F}}]
  \geq \frac{1}{8k}I[1_{\mathcal{F}}],
  \end{multline*}
  which is at least $\frac{d}{64k}\mu(f)$ by Theorem~\ref{thm:KKL_multi_slice}
  (and the fact that ${\sf var}(f) = \mu(f)(1-\mu(f))\geq \mu(f)/2$).
  It follows that the denominator in~\eqref{eq:KK_stab} is at most
  $\mu(f)\left(1-\frac{d}{64k}\right)$, and plugging this into~\eqref{eq:KK_stab} we get that
  \[
  \mu(\mathcal{F}\uparrow)
  \geq \left(1+\frac{d}{64k}\right) \mu(\mathcal{F}).\qedhere
  \]
\end{proof}
We finish this section by noting that Theorem~\ref{thm:KK_stability} indeed improves on Theorem~\ref{thm:KK} in some range of parameters.
Namely, in the case that $x = \Theta(k)$, $x\leq k-2$ and $n\geq 2^{C\cdot k^3}$.
Normalizing the inequality in Theorem~\ref{thm:KK}, we get that
\[
\mu(\mathcal{F}\uparrow)\geq \frac{{n\choose k}}{{n\choose k+1}} \frac{{n\choose x+1}}{{n\choose x}}\mu(\mathcal{F})
=\frac{k+1}{n-k}\frac{n-x}{x+1} = \left(1 + \Theta\left(\frac{k-x}{k}\right)\right)\mu(\mathcal{F}),
\]
so it is enough to note that $\mathcal{F}$ is $(d,2^{-C\cdot d^4})$-global for $d= \left\lfloor\frac{k-x}{2}\right\rfloor$. Indeed, if
$\card{I}=d$ and $A\subseteq I$, then
\[
\mu(\mathcal{F_{I\rightarrow A}})
\leq \frac{{n \choose x}}{{n-d \choose k-\card{A}}}
\leq \frac{{n \choose x}}{{n-d \choose k-d}}
=\frac{n(n-1)\cdots(n-x+1)}{(n-d)(n-d-1)\cdots(n-k+1)}\frac{(k-d)!}{x!}
\leq k^{d+1}
\frac{n^x}{n^{k-d}}
\leq k^{d+1} n^{-d},
\]
which at most $2^{-C\cdot d^4}$ provided that $n$ is large enough.

\section{Proof of the level-$d$ inequality}\label{sec:lvl_d}
The goal of this section is to prove Theorem~\ref{thm:lvl_d}.

\subsection{Proof overview}
\paragraph{Proof overview in an idealized setting.}
We first describe the proof idea in an idealized setting in which derivative operators, and truncations, interact well. By that, we mean that
if $\mathrm{D}$ is an order $\ell$ derivative, and $f$ is a function, then $\mathrm{D} (f^{\leq d}) = (\mathrm{D} f)^{\leq d-\ell}$. We remark
that this property holds in product spaces, but may fail in non-product domains such as $S_n$.

Adapting the proof of the level-$d$ inequality from the hypercube (using Theorem~\ref{thm:Reasonability} instead of standard hypercontractivity),
one may easily establish a weaker version of Theorem~\ref{thm:lvl_d}, wherein $\eps^2$ is replaced by $\eps^{3/2}$, as follows. Take
$q =\log(1/\eps)$, then
\[
\norm{f^{\leq d}}_2^2
= \langle f^{\leq d}, f\rangle
\leq \norm{f^{\leq d}}_q \norm{f}_{1+1/(q-1)}.
\]
Since $f$ is integer-valued, we have that $\norm{f}_{1+1/(q-1)}$ is at most
$\norm{f}_2^{2(q-1)/q}\leq \eps^{2(q-1)/q}$.
Using the assumption of our idealized setting and Parseval, we get that for every derivative $\mathrm{D}$ of order $\ell$ we have
that $\norm{\mathrm{D} (f^{\leq d})}_2 = \norm{(\mathrm{D} f)^{\leq d-\ell}}_2\leq \norm{\mathrm{D} f}_2$. Thus, using the globalness
of $f$ and both items of Claim~\ref{claim:globalness equivalence}, we get that $f^{\leq d}$ is $(d,2^{O(d)}\eps)$-global, and so
by Theorem~\ref{thm:Reasonability} we get that $\norm{f^{\leq d}}_q\leq (2q)^{O(d^3)}\eps$. All in all, we get that
$\norm{f^{\leq d}}_2^2\leq (2q)^{O(d^3)}\eps^{3}$, which falls short of Theorem~\ref{thm:lvl_d} by a factor of $\eps$.

The quantitative deficiency in this argument stems from the fact that $f^{\leq d}$ in fact is much more global than what
the simplistic argument above establishes, and to show that we prove things by induction on $d$. This induction is also the reason
we have strengthened Theorem~\ref{thm:lvl_d} from the introduction to the statement above.

\paragraph{Returning to the real setting.} To lift the assumption of the ideal setting, we return to discuss restrictions (as opposed to derivatives).
Again, we would have been in good shape if restrictions were to commute with degree truncations, but this again fails, just like derivatives.
Instead, we use the following observation (Claim~\ref{claim:Stupidtrivial}).
Suppose $k\geq d+\ell+2$, and let $g$ be a function of pure degree $k$, and $S$ be a restriction of size at most $\ell$. Then the restricted function $g_S$ is perpendicular
to degree $k-\ell-1 > d$ functions, and so $(g_S)^{\leq d} = ((g^{\leq k})_S)^{\leq d}$.

Note that for $k=d$, this statement exactly corresponds to truncations and restrictions commuting, but the conditions of the statement always require that $k>d$ at the
very least. In fact, in our setting we will have $\ell = 2d$, so we would need to use the statement with $k = 3d+2$. Thus, to use this statement effectively we cannot
apply it on our original function $f$, and instead have to find an appropriate choice of $g$ such that $g^{\leq k}, g^{\leq d}\approx f^{\leq d}$,
and moreover that they remain close under restrictions (so in particular we preserve our globalness). Indeed, we are able to design such $g$
by applying appropriate sparse linear combinations of powers of the natural transposition operator of $S_n$ on $f$.

\subsection{Constructing the auxiliary function $g$}

In this section we construct the function $g$.
\begin{lem}\label{lem:operator that makes everything nice}
There is an absolute constant $C>0$, such that the following holds.
Suppose $n\ge2^{C\cdot d^{3}}$, and let $\mathrm{T}$ be the adjacency operator of the transpositions graph (see Section~\ref{sec:isoperimetric}).
There exists a polynomial $P$ with $\norm{P}\le 2^{C\cdot d^4}$
such that
\[
\norm{P(\mathrm{T}) (f^{\leq 4d})-f^{\le d}}_{2}\le\left(\frac{1}{n}\right)^{19d}\norm{f^{\le 4d}}_{2}.
\]
\end{lem}

\begin{proof}
Let
\[
Q(z) = \sum_{i=1}^{d}\prod_{j\in\left[4d\right]\setminus\left\{ i\right\} }\left(\frac{z^{n}-e^{-2j}}{e^{-2i}-e^{-2j}}\right)^{20d},
\]
and define $P(z) = 1 - (1-Q(z))^{20d}$. We first prove the upper bound on $\norm{P}$; note that
\[
\norm{Q} \le\sum_{i=1}^{d}\prod_{j\in\left[4d\right]\setminus\left\{ i\right\} }\norm{\frac{z^{n}-e^{-2j}}{e^{-2i}-e^{-2j}}}^{20d}
  =\sum_{i=1}^{d}\prod_{j\in\left[4d\right]\setminus\left\{ i\right\} }\left(\frac{1+e^{-2j}}{e^{-2i}-e^{-2j}}\right)^{20d}=2^{O\left(d^{3}\right)},
\]
so $\norm{P}\leq (1+2^{O(d^3)})^{20d} = 2^{O(d^4)}$.

Next, we show that for $g = P(\mathrm{T}) f$, it holds that
$\norm{g^{\le 4d}-f^{\le d}}_{2}\le\left(\frac{1}{n}\right)^{10d}\norm{f^{\le 4d}}_{2}$, and we do so by eigenvalue considerations.
Let $d<\ell\leq 4d$, and let $\lambda$ be an eigenvalue of $\mathrm{T}$ corresponding to a function of pure degree $\ell$. Since $\ell\leq n/2$, Lemma~\ref{lem:FOW} implies that
$\lambda = 1-\frac{2\ell}{n} + O\left(\frac{\ell^2}{n^2}\right)$, and so
$\lambda^{n}=e^{-2\ell}\pm O\left(\frac{d^{2}}{n}\right)$. Thus, as each one of the products in $Q(\lambda)$ contains a term for $\ell$, we get that
\[
\card{Q(\lambda)}\le d\cdot\left(\frac{2^{O(d^2)}}{n}\right)^{20d}
\leq \frac{2^{O(d^3)}}{n^{20d}},
\]
so $\card{P(\lambda)} = 1-(1-\frac{2^{O(d^3)}}{n^{20d}})^d \leq \frac{1}{n^{19d}}$. Next, let $\ell\leq d$, and let $\lambda$ be an eigenvalue of $\mathrm{T}$ corresponding to a function of
pure degree $\ell$. As before, $\lambda^n = e^{-2\ell}\pm O\left(\frac{d^{2}}{n}\right)$, but now in $Q(\lambda)$ there is one product that omits the term for $\ell$.
A direct computation gives that
\[
Q(\lambda)
= \prod\limits_{j\in [4d]\setminus\{\ell\}}
\left(\frac{\lambda^n-e^{-2j}}{e^{-2\ell}-e^{-2j}}\right)^{20d} + \frac{2^{O(d^3)}}{n^{20d}}
=\prod\limits_{j\in [4d]\setminus\{\ell\}}\left(1-O\left(\frac{2^{O(d)}}{n}\right)\right)+ \frac{2^{O(d^3)}}{n^{20d}},
\]
so $Q(\lambda) = 1- O\left(\frac{2^{O(d)}}{n}\right)$. Thus,
\[
\card{P(\lambda)-1} = O\left(\frac{2^{O(d^2)}}{n^{20d}}\right) \leq \frac{1}{n^{19d}}.
\]
It follows that $g^{\leq 4d} - f^{\leq d} = \sum\limits_{\ell = 0}^{4d} c_{\ell} f^{=\ell}$ for $\card{c_{\ell}} \leq \frac{1}{n^{19d}}$,
and the result follows from Parseval.
\end{proof}

\subsection{Properties of Cayley operators and restrictions}
In this section, we study random walks along Cayley graphs on $S_n$. The specific transition operator we will later be concerned with
is the transposition operator from Lemma~\ref{lem:FOW} and its powers, but we will present things in greater generality.
\subsubsection{Random walks}
\begin{defn}
  A Markov chain $\mathrm{M}$ on $S_n$ is called a Cayley random walk if for any $\sigma,\tau,\pi\in S_n$,
  the transition probability from $\sigma$ to $\tau$ is the same as the transition probability from
  $\sigma \pi$ to $\tau \pi$.
\end{defn}
In other words, a Markov chain $\mathrm{M}$ is called Cayley if the transition probability from $\sigma$ to
$\tau$ is only a function of $\sigma\tau^{-1}$. We will be interested in the interaction between random walks
and restrictions, and towards this end we first establish the following claim, asserting that a Cayley random
walk either never transitions between two restrictions $T$ and $T'$, or can always transition between the two.
\begin{claim}
  Suppose $\mathrm{M}$ is a Cayley random walk on $S_n$, let $i_1,\ldots,i_t\in[n]$ be distinct, and
  let $T=\left\{ \left(i_{1},j_{1}\right),\ldots,\left(i_{t},j_{t}\right)\right\}$,
  $T'=\left\{ \left(i_{1},j_{1}'\right),\ldots,\left(i_{t},j_{t}'\right)\right\}$
  be consistent sets. Then one of the following two must hold:
  \begin{enumerate}
    \item $\Pr_{\substack{u\in S_{n}^{T'}\\ v\sim \mathrm{M}v}}\big[v\in S_{n}^{T}\big] = 0$.
    \item For all $\pi\in S_{n}^{T}$, it holds that
    $\Pr_{\substack{u\in S_{n}^{T'}\\ v\sim \mathrm{M}v}}\big[v = \pi\big] > 0$.
  \end{enumerate}
\end{claim}
\begin{proof}
  If the first item holds then we're done, so let us assume otherwise. Then there are $u\in S_{n}^{T'}$,
  $v\in S_{n}^{T}$ such that $\mathrm{M}$ has positive probability of transitioning from $u$ to $v$. Denoting
  $\tau = u v^{-1}$, we note that $\tau(j_{\ell}) = j_{\ell}'$ for all $\ell=1,\ldots, t$. Fix $\pi\in S_n^{T}$.
  Since $\mathrm{M}$ is a Cayley operator, the transition probability from $\tau\pi$ to $\pi$ is positive, and since
  $\tau\pi$ is in $S_{n}^{T'}$, the proof is concluded.
\end{proof}
If $\mathrm{M}$ satisfies the second item of the above claim with $T$ and $T'$, we say that $\mathrm{M}$ is compatible with $(T,T')$.

\subsubsection{Degree decomposition on restrictions}
Let $T=\left\{ \left(i_{1},j_{1}\right),\ldots,\left(i_{t},j_{t}\right)\right\}$ be consistent. A function $f\in L^2(S_n^{T})$ is called
a $d$-junta, if there is $S\subseteq [n]\setminus\{i_1,\ldots,i_t\}$ of size $d$ such that $f(\pi)$ only depends on $\pi(i)$ for $i \in S$ (we say that $f(\pi)$ only depends on $\pi(S)$). With this definition in
hand, we may define the space of degree $d$ functions on $S_n^{T}$, denoted by $V_d(S_n^T)$, as the span of all $d$-juntas, and subsequently
define projections onto this subspaces. That is, for each $f\in L^2(S_n^{T})$ we denote by $f^{\leq d}$ the projection of
$f$ onto $V_d(S_n^T)$. Finally, we define the pure degree $d$ part of $f$ as $f^{=d} = f^{\leq d} - f^{\leq d-1}$.

We have the following basic property of pure degree $d$ functions.
\begin{claim}\label{claim:Stupidtrivial}
Suppose that $f\colon S_n\to\mathbb{R}$ is of pure degree $d$.
Let $T$ be a set of size $\ell<d$. Then $f_{T}$ is orthogonal to all
functions in $V_{d-1-\ell}$.
\end{claim}
\begin{proof}
Clearly, it is enough to show that $f_T$ is orthogonal to all $(d-1-\ell)$-juntas.
Fix $g\colon S_{n}^{T}\to\mathbb{R}$ to be  a  $(d-1-\ell)$-junta, and let $h$ be its extension
to $S_n$ by setting it to be $0$ outside $S_n^T$. Then $h$ is a $\left(d-1\right)$-junta, and so
\[
0=\left\langle f,h\right\rangle =\frac{\left(n-\ell\right)!}{n!}\left\langle f_{T},g\right\rangle
.\qedhere
\]
\end{proof}

\subsubsection{Extension to functions}

Any random walk $\mathrm{M}$ on $S_n$ extends to an operator on functions on $S_n$, which maps $f\colon S_n \to \mathbb{R}$ to the function $\mathrm{M}f\colon S_n \to \mathbb{R}$ given by
\[
 \mathrm{M}f(\pi) = \E_{\substack{u \in S_n \\ v \sim \mathrm{M} u}}[f(u)\mid v=\pi].
\]

\subsection{Strengthening Proposition~\ref{prop:n to m}}
Our main goal in this section is to prove the following statement that both strengtheners and generalizes Proposition~\ref{prop:n to m}.
\begin{prop}
\label{prop:globalness of Tf} Let $f\colon S_{n}\to\mathbb{R}$.
Let $\mathrm{M}$ be a Cayley random walk on $S_{n}$,
let $g=\mathrm{M}f$, and let $T=\left\{ \left(i_{1},j_{1}\right),\ldots,\left(i_{t},j_{t}\right)\right\} $
be a consistent set. Then for all $d$,
\[
\|\left(g_{T}\right)^{\le d}\|_{2}\le
\max_{\substack{T'=\left\{ \left(i_{1},j_{1}'\right),\ldots,\left(i_{t},j_{t}'\right)\right\}\\ \mathrm{M}\text{ compatible  with }(T,T')}}\|\left(f_{T'}\right)^{\le d}\|_{2}.
\]
\end{prop}
Let $\mathrm{M}$ be a Cayley random walk and let $T=\left\{ \left(i_{1},j_{1}\right),\ldots,\left(i_{t},j_{t}\right)\right\}$ and $T'=\left\{ \left(i_{1},j_{1}'\right),\ldots,\left(i_{t},j_{t}'\right)\right\}$ be
consistent so that $\mathrm{M}$ is compatible with $(T,T')$.
Put $I = \{(i_1,i_1),\ldots,(i_t,i_t)\}$. 
Define the operator
$\mathrm{M}_{S_n^T\rightarrow S_n^{T'}}\colon L^2(S_n^T)\to L^2(S_n^{T'})$ in the following way: given a function
$f\in  L^2(S_n^T)$, we define
\[
\mathrm{M}_{S_n^{T'}\rightarrow S_n^{T}} f(\pi) = \E_{\substack{u\in_R S_n^{T}\\ v\sim \mathrm{M} u}}\big[f(u)\;\big|\;v = \pi\big].
\]

Drawing inspiration from the proof of Proposition~\ref{prop:n to m}, we study
the operator $\mathrm{M}_{S_{n}^{T}\to S_{n}^{T'}}$. Since we are also dealing with degree truncations, we have to study
its interaction with this operator. Indeed, a key step in the proof is to show that the two operators commute, in the following sense:
for all $d\in\mathbb{N}$ and $f\in L^{2}(S_n^T)$, it holds that
\[
\left(\mathrm{M}_{S_{n}^{T}\to S_{n}^{T'}}f\right)^{=d}=\mathrm{M}_{S_{n}^{T}\to S_{n}^{T'}}\left(f^{=d}\right).
\]
Towards this end, we view $L^2(S_n^{T})$ (and similarly $L^2(S_n^{T'})$) as a right $S_n^I$-module using the following operation:
a function-permutation pair $(f,\pi)\in L^2(S_n^{T})\times S_n^I$ is mapped to a function $f^{\pi} \in L^2(S_n^{T})$ defined as
\[
f^{\pi}(\sigma) = f(\sigma \pi^{-1}).
\]

\begin{claim}
\label{claim:symmerty of transitioning operator}
With the setup above, $\mathrm{M}_{S_{n}^{T}\to S_{n}^{T'}}\colon L^{2}\left(S_{n}^{T}\right)\to L^{2}\left(S_{n}^{T'}\right)$
is a homomorphism of $S_{n}^{I}$-modules.
\end{claim}
\begin{proof}
The proof is essentially the same as the proof of Lemma~\ref{lem:Trho is symmetric}, and is therefore omitted.
\end{proof}

Therefore, it is sufficient to prove that any homomorphism commutes with taking pure degree $d$ part, which is the content of the following claim.
\begin{claim}
\label{claim:pure degree d}
Let $T,T'$ be consistent as above, and let $\mathrm{A}\colon L^{2}(S_{n}^{T})\to L^{2}(S_{n}^{T'})$
be a homomorphism of right $S_{n}^{I}$-modules. Then for all $f\in L^2(S_n^T)$ we have that
\[
\left(\mathrm{A}f\right)^{=d}
=\mathrm{A}\left(f^{=d}\right).
\]
\end{claim}

\begin{proof}
We first claim that $\mathrm{A}$ preserves degrees, i.e.\
$\mathrm{A} V_{d}(S_n^{T})\subseteq V_{d}(S_n^{T'})$. To show this, it is enough to note that if $f\in L^2(S_n^T)$ is a $d$-junta,
then $\mathrm{A} f$ is a $d$-junta. Let $f$ be a $d$-junta, and suppose that $S\subseteq[n]$ is a set
of size at most $d$ such that $f(\sigma)$ only depends on $\sigma(S)$. Then for any $\pi$ that has $S$
as fixed points, we have that $f(\sigma) = f(\sigma \pi^{-1}) = f^{\pi}(\sigma)$, so $f = f^{\pi}$. Applying
$\mathrm{A}$ and using the previous claim we get that $\mathrm{A} f = \mathrm{A} f^{\pi} = (\mathrm{A} f)^{\pi}$.
This implies that $\mathrm{A}f$ is invariant under any permutation that keeps $S$ as fixed points, so it is an
$S$-junta.

Let $V_{=d}(S_n^{T})$ be the space of functions of pure degree $d$,
i.e.\ $V_{d}(S_n^{T})\cap V_{d-1}(S_n^{T})^{\perp}$. We claim that $\mathrm{A}$ also preserves pure degrees,
i.e.\ $\mathrm{A} V_{=d}(S_n^{T})\subseteq V_{=d}(S_n^{T'})$. By the previous paragraph it is enough to show that
if $f\in V_{=d}(S_n^{T})$, then $\mathrm{A}f$ is orthogonal to $V_{d-1}(S_n^{T'})$. Letting $\mathrm{A}^{*}$ be the
adjoint operator of $\mathrm{A}$, it is easily seen that $\mathrm{A}^{*}\colon L^{2}(S_{n}^{T'})\to L^{2}(S_{n}^{T})$
is also a homomorphism between right $S_n^{I}$-modules, and by the previous paragraph it follows that
$\mathrm{A}^{*}$ preserves degrees. Thus, for any $g\in V_{d-1}(S_n^{T'})$ we have that $\mathrm{A}^{*} g\in V_{d-1}(S_n^{T})$,
and so
\[
\left\langle \mathrm{A}f,g\right\rangle =\left\langle f,\mathrm{A}^{*}g\right\rangle =0.
\]

We can now prove the statement of the claim. Fix $f\in L^2(S_n^T)$ and $d$. Then by the above paragraph,
$\mathrm{A}\left(f^{=d}\right)\in V_{=d}(S_n^{T'})$, and by linearity of $\mathrm{A}$ we have
$\sum\limits_{d}\mathrm{A}\left(f^{=d}\right) = \mathrm{A} f$. The claim follows from the uniqueness of the degree
decomposition.
\end{proof}

We define a transition operator on restrictions as follows. From a restriction $T=\left\{ \left(i_{1},j_{1}\right),\ldots,\left(i_{t},j_{t}\right)\right\}$,
we sample $T'\sim N(T)$ as follows. Take $\pi\in S_{n}^T$ uniformly, sample $\sigma\sim \mathrm{M}\pi$, and then let $T'$ be $\{(i_1,\sigma(i_1)),\ldots,(i_t,\sigma(i_t))\}$.
The following claim is immediate:
\begin{claim}\label{claim:imm}
$(\mathrm{M} f)_T = \E_{T'\sim N(T)}\big[\mathrm{M}_{S_{n}^{T'}\to S_{n}^{T}} f_{T'}\big]$.
\end{claim}

We are now ready to prove Proposition~\ref{prop:globalness of Tf}.
\begin{proof}[Proof of Proposition~\ref{prop:globalness of Tf}]
By Claim~\ref{claim:imm}, we have $g_{T}=\mathbb{E}_{T'\sim T}\mathrm{M}_{S_{n}^{T'}\to S_{n}^{T}}f_{T'}$.
Using Claim~\ref{claim:pure degree d}
and the linearity of the operator $f\mapsto f^{= d}$, we get
\[
\left(g_{T}\right)^{= d}=\mathbb{E}_{T'\sim N(T)}\mathrm{M}_{S_{n}^{T'}\to S_{n}^{T}}\left(\left(f_{T'}\right)^{= d}\right).
\]
Summing this up using linearity again, we conclude that
\[
\left(g_{T}\right)^{\leq d}=\mathbb{E}_{T'\sim N(T)}\mathrm{M}_{S_{n}^{T'}\to S_{n}^{T}}\left(\left(f_{T'}\right)^{\leq d}\right).
\]
Taking norms and using the triangle inequality gives us that
\[
\|\left(g_{T}\right)^{\le d}\|_{2}
\le
\mathbb{E}_{T'\sim N(T)}\left\Vert \mathrm{M}_{S_{n}^{T'}\to S_{n}^{T}}\left(\left(f_{T'}\right)^{\le d}\right)\right\Vert _{2}
\leq
\max_{T'\colon \mathrm{M}\text{ consistent with }(T,T')}\left\Vert \mathrm{M}_{S_{n}^{T'}\to S_{n}^{T}}\left(\left(f_{T'}\right)^{\le d}\right)\right\Vert _{2}.
\]
The proof is now concluded by appealing to Fact~\ref{fact:contraction}.
\end{proof}

\subsection{A weak level-$d$ inequality}
The last ingredient we will need in the proof of Theorem \ref{thm:lvl_d} is a weak version of the level-$d$ inequality,
which does not take the globalness of $f$ into consideration.
\begin{lem}
\label{lem:weak level d} Let $C$ be sufficiently large, let $n\ge\log\left(1/\epsilon\right)^{d}C^{d^{2}}$,
and let $f\colon S_{n}\to\left\{ 0,1\right\} $ satisfy $\|f\|_{2}\le\epsilon$.
Then
\[
\|f^{\le d}\|_{2}\le n^{d}\log\left(1/\epsilon\right)^{O\left(d\right)}\epsilon^{2}.
\]
\end{lem}

\begin{proof}
Set $q=\log(1/\eps)$, and without loss of generality assume $q$ is an even integer (otherwise we may change $q$ by a constant factor to ensure that).
Using H\"{o}lder's inequality, Lemma~\ref{lem:traditional hypercontractivity },
and the fact that $\|f\|_{q/(q-1)}=O\left(\epsilon^{2}\right)$,
we obtain
\[
\norm{f^{\le d}}_{2}^{2}  =\left\langle f^{\le d},f\right\rangle \le
\norm{f^{\le d}}_{q}\norm{f}_{q/(q-1)}
\le
\log\left(1/\epsilon\right)^{O\left(d\right)}n^{d}\norm{f^{\le d}}_{2}\epsilon^{2},
\]
 and the lemma follows by rearranging.
\end{proof}

\subsection{Interchanging truncations and derivatives with small errors}
\begin{lem}
\label{lem:upper bound on restriction  of truncation }
There is $C>0$, such that the following holds for $n\geq 2^{C\cdot d^3}$.
For all derivatives $\mathrm{D}$ of order $t\leq d$ we have:
\[
\norm{D\left(f^{\le d}\right)}_{2}
\le 2^{O\left(d\right)^{4}}\max_{t-\text{derivative }\mathrm{D}'}\norm{\left(\mathrm{D}'f\right)^{\le d-t}}_{2}
 +\left(\frac{1}{n}\right)^{10 d}\norm{f^{\le 4d}}_2.
\]
\end{lem}

\begin{proof}
Let $\mathrm{T}_{d}=P\left(\text{\ensuremath{\mathrm{T}}}\right)$
be as in Lemma \ref{lem:operator that makes everything nice},
and write $f^{\leq d} = \mathrm{T}_d (f^{\leq 4d}) + g$, where
$\norm{g}_2\leq n^{-19d} \norm{f^{\leq 4d}}_2$.
Let $S$ be a consistent restriction of $t$ coordinates, and
let $\mathrm{D}$ be a derivative along $S$.
Then there is $R\subseteq L$ of size $t$ such that $\mathrm{D}f=\left(\mathrm{L}f\right)_{S\rightarrow R}$.
By Claim~\ref{claim:degree reduction},
the degree of $\mathrm{D}(f^{\leq d})$ is at most $d-t$, thus
\begin{equation}\label{eq:stupid3}
\mathrm{D} (f^{\leq d}) = \left(\mathrm{D} (f^{\leq d})\right)^{\leq d-t}.
\end{equation}
We want to compare the right-hand side with $\left(\mathrm{D}\left(\mathrm{T}_{d} f\right)\right)^{\leq d-t}$, but first
we show that in it one may truncate all degrees higher than $4d$ in $f$.
Note that by Claim~\ref{claim:pure degree d}, for each $k > 4d$ the function $\mathrm{T}_{d} f^{=k}$ has pure
degree $k$, so $\mathrm{D}(\mathrm{T}_{d} f^{=k})$ is perpendicular to degree $k-t-1$ functions. Since $k-2t-1 \geq d-t$,
we have that its level $d-t$ projection is $0$, so
$\left(\mathrm{D}\left(\mathrm{T}_{d} f\right)\right)^{\leq d-t} = \left(\mathrm{D}\left(\mathrm{T}_{d} f^{\leq 4d}\right)\right)^{\leq d-t}$.
It follows that
\begin{align}
\norm{\mathrm{D} (f^{\leq d}) - \left(\mathrm{D}\left(\mathrm{T}_{d} f\right)\right)^{\leq d-t}}_2
=\norm{\left(\mathrm{D}\left(f^{\leq d} - \mathrm{T}_{d} (f^{\leq 4d})\right)\right)^{\leq d-t}}_2
\leq \norm{\mathrm{D} g}_2
&\leq n^{2t}\norm{g}_2 \notag\\\label{eq:aux_1}
&\leq n^{2t-19d}\norm{f^{\leq 4d}}_2.
\end{align}

Our task now is to bound $\norm{\left(\mathrm{D}\left(\mathrm{T}_{d} f\right)\right)^{\leq d-t}}_2$.
Since $\mathrm{T}$ commutes with Laplacians, it follows that $\mathrm{T}_d$ also commutes with Laplacians,
and so
\begin{equation}\label{eq:aux_2}
\left(\mathrm{D}\left(\mathrm{T}_{d} f\right)\right)^{\leq d-t}
= ((\mathrm{L} \mathrm{T}_{d} f )_{S\rightarrow R})^{\leq d-t}
= ((\mathrm{T}_{d} \mathrm{L} f )_{S\rightarrow R})^{\leq d-t}.
\end{equation}
By Proposition \ref{prop:globalness of Tf}, for all $i$ and $h\colon S_n\to\mathbb{R}$ we have
\[
\|\left(\left(\mathrm{T}^{i}h\right)_{S}\right)^{\le d}\|_{2}\le
\max_{S'=\left\{ \left(i_{1},j_{1}'\right),\ldots,\left(i_{t},j_{t}'\right)\right\} }\norm{\left(h_{S'}\right)^{\le d}}_2,
\]
and so
\[
\|\left(\left(\mathrm{T}_d h\right)_{S}\right)^{\le d}\|_{2}\le
\norm{P}\max_{S'=\left\{ \left(i_{1},j_{1}'\right),\ldots,\left(i_{t},j_{t}'\right)\right\} }\norm{\left(h_{S'}\right)^{\le d}}_2
\leq
2^{O(d^4)}\max_{S'=\left\{ \left(i_{1},j_{1}'\right),\ldots,\left(i_{t},j_{t}'\right)\right\} }\norm{\left(h_{S'}\right)^{\le d}}_2.
\]
Applying this for $h = Lf$ gives that
\begin{equation}\label{eq:aux_3}
\norm{((\mathrm{T}_{d} \mathrm{L} f )_{S\rightarrow R})^{\leq d-t}}_2
 \le 2^{O\left(d^{4}\right)}\max_{R'} \norm{\left(\left(\mathrm{L}f\right)_{S\rightarrow R'}\right)^{\le d-t}}_{2}
 =2^{O\left(d^{4}\right)}\max_{D'} \norm{\left(\mathrm{D}'f\right)^{\le d-t}}_{2},
\end{equation}
where the last transition is by the definition of derivatives. Combining~\eqref{eq:aux_1},~\eqref{eq:aux_2},~\eqref{eq:aux_3}
and using the triangle inequality finishes the proof.
\end{proof}

\subsection{Proof of the level-$d$ inequality}
We end this section by deriving the following proposition, which by Claim~\ref{claim:globalness equivalence} implies Theorem~\ref{thm:lvl_d}.
\begin{prop}\label{prop:lvl_d}
There exists an absolute constant
$C>0$ such that the following holds for all $d\in\mathbb{N}$, $\eps>0$ and $n\geq 2^{C\cdot d^3} \log(1/\eps)^{C\cdot d}$. Let $f\colon S_n\to \mathbb{Z}$ be a function, such
that for all $t\le d$ and all $t$-derivatives $\mathrm{D}$ we have
$\|\mathrm{D}f\|_{2}\le\epsilon$. Then
\[
\norm{f^{\le d}}_{2}\le 2^{Cd^{4}}\epsilon^{2}\log\left(1/\epsilon\right)^{Cd}.
\]
\end{prop}

\begin{proof}
The proof is by induction on $d$. If $d=0$, then
\[
\norm{f^{\le d}}_{2}
= \card{\Expect{}{f(\pi)}}
\leq \Expect{}{\card{f(\pi)}^2}
=\norm{f}_2^2
\leq \eps^2,
\]
where in the second transition we used the fact that $f$ is integer-valued.

We now prove the inductive step. Fix $d\geq 1$.
Let $1\leq t\leq d$, and
let $\mathrm{D}$ be a $t$-derivative.
By Lemma~\ref{lem:upper bound on restriction  of truncation }, there is an absolute constant $C_1>0$ such that
\begin{equation}\label{eq2}
\norm{\mathrm{D}\left(f^{\le d}\right)}_{2}  \le e^{C_{1}\left(d^{4}\right)}\max_{\mathrm{D}'\text{ a }t-\text{derivative}}\norm{\left(\mathrm{D}'f\right)^{\le d-t}}_{2}+
n^{-10 d}
\norm{f^{\le 4d}}_{2}.
\end{equation}
Fix $\mathrm{D}'$. The function $\mathrm{D}'f$ takes integer values and
is defined on a domain that is isomorphic to $S_{n-t}$, so by the induction hypothesis we have
\[
\norm{\left(\mathrm{D}'f\right)^{\le d-t}}_{2}\le e^{C\left(d-t\right)^{4}}\epsilon^{2}\log\left(\frac{1}{\epsilon}\right)^{C\left(d-t\right)}.
\]
As for $\|f^{\le 4d}\|_{2}^{2}$, applying Lemma~\ref{lem:weak level d} we see it is at most $n^{8d}\eps^4 \log^{Cd}(1/\eps)$.
Plugging these two estimates into~\eqref{eq2} we get that
\[
\norm{\mathrm{D}\left(f^{\le d}\right)}_{2}
 \le e^{Cd^{4}}\epsilon^{2}\log^{C\cdot d}\left(1/\epsilon\right),
\]
 provided that $C$ is sufficiently large.

If
\[
\norm{f^{\le d}}_{2}^{2}\le e^{2Cd^{4}}\epsilon^{4}\log\left(\frac{1}{\epsilon}\right)^{2Cd}
\]
we're done, so assume otherwise. We get that
 $\norm{\mathrm{D}'\left(f^{\le d}\right)}_{2} \leq \|f^{\le d}\|_{2}$ for all derivatives of order at most $d$, and from Claim~\ref{claim:degree reduction},
 $\norm{\mathrm{D}'\left(f^{\le d}\right)}_{2} = 0$ for higher-order derivatives, and so by Claim~\ref{claim:globalness equivalence},
 the function $f^{\le d}$ is $(2d,4^{d}\|f^{\le d}\|_{2})$-global, and by Lemma~\ref{lem:Bootstrapping the globalness}, we get that
 $f^{\leq d}$ is $4^{d}\|f^{\le d}\|_{2}$-global with constant $4^8$.
 In this case, we apply the standard argument as presented in the overview, as outlined below.

Set $q=\log\left(1/\epsilon\right)$, and without loss of generality assume $q$ is an even integer (otherwise we may change $q$ by a constant factor to ensure that).
Set $\rho=\frac{1}{(10 q 4^8)^2}$.
From Lemmas~\ref{lem:Trho is symmetric},~\ref{lem:abjuntas} we have that $\mathrm{T}^{(\rho)}$ preserves degrees,
and so by Corollary~\ref{cor:eigenvalues} we get
\[
\norm{f^{\le d}}_{2}^{2}
\leq \rho^{-C_2\cdot d}\inner{\mathrm{T}^{\left(\rho\right)}f^{\le d}}{f^{\le d}}
=\rho^{-C_2\cdot d}\inner{\mathrm{T}^{\left(\rho\right)}f^{\le d}}{f}
\le
\rho^{-C_2\cdot d}
\norm{\mathrm{T}^{\left(\rho\right)}f^{\le d}}_{q}
\norm{f}_{q/(q-1)},
\]
where we also used H\"{o}lder's inequality.
 By Theorem~\ref{thm:Hypercontractivity}, we have
$\norm{\mathrm{T}^{(\rho)} f^{\le d}}_{q}\le 4^{d}\norm{f^{\le d}}_{2}$, and
by a direction computation $\norm{f}_{q/(q-1)}\leq \eps^{2(q-1)/q}$. Plugging these two estimates into the inequality above and rearranging yields that
\[
\norm{f^{\le d}}_{2}^{2}  \le
\rho^{-2C_2\cdot d} 4^{2d} \norm{f}_{q/(q-1)}^2
\leq \rho^{-3C_2\cdot d} \eps^4
= 2^{6C_2 \log(10C)} \eps^4 \log^{6C_2\cdot d}(1/\eps)
\leq 2^{C\cdot d^4} \eps^4 \log^{C \cdot d}(1/\eps),
\]
for large enough $C$.
\end{proof}

\subsection{Deducing the strong level-$d$ inequality: proof of Theorem~\ref{thm:lvl_d_strong}}\label{sec:stronger_lvl_d}
Let $\delta = 2^{C_1\cdot d^4} \eps^2\log^{C_1\cdot d}(1/\eps)$ for sufficiently large absolute constant $C_1$.
By Claim~\ref{claim:globalness_of_low_deg_part} we get that $f^{\leq d}$ is $\delta$-global with constant $4^8$.
Set $q = \log(1/\norm{f}_2)$, and let $\rho = 1/(10\cdot 4^8\cdot q)^2$ be from Theorem~\ref{thm:Hypercontractivity}.
From Lemmas~\ref{lem:Trho is symmetric},~\ref{lem:abjuntas} we have that $\mathrm{T}^{(\rho)}$ preserves degrees,
and so by Corollary~\ref{cor:eigenvalues} we get
\[
\norm{f^{\leq d}}_2^2 =
\inner{f^{\leq d}}{f^{\leq d}}
\leq \rho^{-O(d)}\inner{f^{\leq d}}{\mathrm{T}^{(\rho)} f^{\leq d}}
= \rho^{-O(d)}\inner{f}{\mathrm{T}^{(\rho)} f^{\leq d}}
\leq \rho^{-O(d)}\norm{f}_{q/(q-1)}\norm{\mathrm{T}^{(\rho)} f^{\leq d}}_q.
\]
Using $\norm{f}_{q/(q-1)} \leq \norm{f}_2^{2(q-1)/q} = \norm{f}_2^2 \norm{f}_2^{-2/q}\leq O(\norm{f}_2^2)$ and
Theorem~\ref{thm:Hypercontractivity} to bound
$\norm{\mathrm{T}^{(\rho)} f^{\leq d}}_q\leq \delta$, it follows that
\[
\norm{f^{\leq d}}_2^2\leq \rho^{-O(d)} \norm{f}_2^2\delta \leq 2^{C\cdot d^4} \norm{f}_2^2\eps^2\log^{C\cdot d}(1/\epsilon),
\]
where we used $\norm{f}_2^2\leq \eps$.\qed
\bibliographystyle{abbrv}
\bibliography{ref}

\appendix
\section{Missing proofs}\label{sec:KKL}
\subsection{Globalness of $f$ implies globalness of $f^{\leq d}$}
\begin{claim}\label{claim:globalness_of_low_deg_part}
  There exists an absolute constant $C>0$ such that the following holds for all
  $n,d\in\mathbb{N}$ and $\eps>0$ satisfying $n\geq 2^{C\cdot d^3} \log(1/\eps)^{C\cdot d}$.
  Suppose $f\colon S_n\to\mathbb{Z}$ is $(2d,\eps)$-global.
  Then for all $j\leq d$, the function $f^{\leq j}$ is
  \begin{enumerate}
    \item $(2j,2^{O(j^4)}\eps^2 \log^{O(j)}(1/\eps))$-global.
    \item $2^{O(j^4)}\eps^2 \log^{O(j)}(1/\eps)$-global with constant $4^8$.
  \end{enumerate}
\end{claim}
\begin{proof}
  If $j=0$, then the claim is clear as $f^{\leq j}$ is just the constant
  $\Expect{}{f(\pi)}$, and its absolute value is at most $\norm{f}_2^2\leq \eps^2$.

  Suppose $j\geq 1$ and let $\mathrm{D}$ be a derivative of order $1\leq r\leq j$, then
  by Claim~\ref{claim:globalness equivalence} we have
  $\norm{\mathrm{D} f}_2\leq 2^{2j}\eps$. Therefore, applying Proposition~\ref{prop:lvl_d}
  on $\mathrm{D} f$, we get that
  \[
  \norm{(\mathrm{D} f)^{\leq j-1}}_2 \leq 2^{O((j-1)^4)} \eps^2 \log^{O(j)}(1/\eps).
  \]
  Using Lemma~\ref{lem:upper bound on restriction  of truncation } we get that
  \[
  \norm{\mathrm{D}(f^{\leq j})}_2
  \leq 2^{O\left(j\right)^{4}}\max_{1-\text{derivative }\mathrm{D}'}\norm{\left(\mathrm{D}'f\right)^{\le j-1}}_{2}
  +\left(\frac{1}{n}\right)^{10j}\norm{f^{\le 4j}}_2
  \leq
  2^{O(j^4)} \eps^2 \log^{O(j)}(1/\eps),
  \]
  where in the last inequality we our earlier estimate and Lemma~\ref{lem:weak level d}.
  For derivatives of order higher than $j$, we have that $\mathrm{D}(f^{\leq j}) = 0$ from Claim~\ref{claim:degree reduction}.
  Thus, Claim~\ref{claim:globalness equivalence} implies that $f^{\leq j}$ is
  $(2j,2^{O(j^4)}\eps^2 \log^{O(j)}(1/\eps))$-global. The second item immediately follows from Lemma~\ref{lem:Bootstrapping the globalness}.
\end{proof}

\subsection{Proof of Theorem~\ref{thm:KKL_analog}}
Our proof will make use of the following simple fact.
\begin{fact}\label{fact:poincare}
 Let $g\colon S_n\to\mathbb{R}$.
 \begin{enumerate}
   \item We have the Poincar\'{e} inequality:
   ${\sf var}(g)\leq \frac{1}{n}\sum\limits_{\mathrm{L}_1} \norm{\mathrm{L}_1 g}_2^2$,
   where the sum is over all $1$-Laplacians.
   \item We have $I[g] = \frac{2}{n-1}\sum\limits_{\mathrm{L}_1} \norm{\mathrm{L}_1 g}_2^2$,
   where again the sum is over all $1$-Laplacians.
 \end{enumerate}
\end{fact}
\begin{proof}
  The second item is straightforward by the definitions, and we focus on the first one.
  Let $\tilde{\mathrm{L}} g = \Expect{\mathrm{L}_1}{\mathrm{L}_1 g} = (I-\mathrm{T}) g$.
  If $\alpha_{d,r}$ is an eigenvalue of $\mathrm{T}$ corresponding to a function from
  $V_{=d}(S_n)$, then by the second item in Lemma~\ref{lem:FOW} we have
  $\alpha_{d,r}\leq 1-\frac{d}{n-1}$.

  Note that we may find an orthonormal basis of $V_{=d}(S_n)$ consisting of eigenvectors of $\mathrm{T}$,
  and therefore we may first write $g = \sum\limits_{d} g^{=d}$ where
  $g^{=d}\in V_{=d}(S_n)$, and then further decompose each $g_d$ to
  $g_d = \sum\limits_{r=0}^{r_d} g_{d,r}$ where $g_{d,r}\in V_{=d}(S_n)$
  are all orthogonal and eigenvectors of $\mathrm{T}$. We thus get
  \begin{equation}\label{eq:apx5}
  \inner{g}{\tilde{\mathrm{L}} g}
  =\sum\limits_{d}\sum\limits_{r=0}^{r_d} (1-\alpha_{d,r}) \norm{g^{d,r}}_2^2
  \geq \sum\limits_{d}\sum\limits_{r=0}^{r_d}\frac{d}{n-1} \norm{g^{d,r}}_2^2
  = \sum\limits_{d}\frac{d}{n-1}\norm{g^{=d}}_2^2
  \geq \frac{1}{n-1}{\sf var}(g).
  \end{equation}
  On the other hand,
  \[
  \inner{g}{\tilde{L} g}
  =\Expect{\pi}{\Expect{\tau\text{ a transposition}}{g(\pi)(g(\pi) - g(\pi\circ \tau))}}
  =\frac{1}{2} \Expect{\tau\text{ a transposition}}{\Expect{\pi}{(g(\pi) - g(\pi\circ \tau))^2}},
  \]
  which is the same as $\frac{1}{2{n\choose 2}}\sum\limits_{\mathrm{L}_1}{\norm{\mathrm{L}_1g}_2^2}$.
  Combining this with the previous lower bound gives the first item.
\end{proof}

\begin{proof}[Proof of Theorem~\ref{thm:KKL_analog}]
    Let $f = 1_S$. Then
    $I[f] = \frac{n-1}{2}\Prob{\substack{\pi\in S_n\\ \sigma\sim \mathrm{T}\pi}}{f(\pi)\neq f(\sigma)}$, and arithmetizing
    that we have that it is equal to $\frac{n-1}{2}\inner{f}{(I-\mathrm{T}) f}$. Thus, writing
    $f = f^{=0}+f^{=1}+\dots$, where $f^{=j}\in V_{=j}(S_n)$, we have, as in inequality~\eqref{eq:apx5},
    that
    \begin{equation}\label{eq:apx4}
    \frac{n-1}{2}\inner{f}{(I-\mathrm{T}) f}
    \geq \frac{n-1}{2}\sum\limits_{j=0}^{n} \frac{j}{n-1} \norm{f^{=j}}_2^2
    \geq \frac{d}{2}\norm{f^{>d}}_2^2.
    \end{equation}
    To finish the proof, we show that $\norm{f^{>d}}_2^2\geq \Omega({\sf var}(f))$.
    To do that, we upper-bound the weight of $f$ on degrees $1$ to $d$.

    Let $g = f^{\leq d}$. We intend to bound
    ${\sf var}(g)$ using the Poincar\'{e} inequality, namely the first item in Fact~\ref{fact:poincare}.
    Fix an order~$1$ Laplacian $\mathrm{L}_1$. We have
    \begin{equation}\label{eq:apx3}
    \norm{\mathrm{L}_1 g}_2^2
    =\inner{\mathrm{L}_1 g}{\mathrm{L}_1 f}
    \leq\norm{\mathrm{L}_1 g}_4\norm{\mathrm{L}_1 f}_{4/3}.
    \end{equation}
    As $f$ is Boolean, $\mathrm{L}_1 f$ is $\set{-1,0,1}$-valued and so
    $\norm{\mathrm{L}_1 f}_{4/3} = \norm{\mathrm{L}_1 f}_{2}^{3/2}$, and next we bound
    $\norm{\mathrm{L}_1 g}_4$. Note that
    \begin{equation}\label{eq:apx1}
    \norm{\mathrm{L}_1 g}_4^4
    =
    \Expect{\substack{\mathrm{D}_1\\\text{order $1$ derivative}\\ \text{consistent with $\mathrm{L}_1$}}}
    {\norm{\mathrm{D}_1 g}_4^4},
    \end{equation}
    and we analyze $\norm{\mathrm{D}_1 g}_4^4$ for all derivatives $\mathrm{D}_1$. For that we use hypercontractivity, and we first have
    to show that $\mathrm{D}_1 g$ is global.

    Fix a $1$-derivative $\mathrm{D}_1$, and set $h = \mathrm{D}_1 g$.
    By Lemma~\ref{lem:upper bound on restriction  of truncation } (with $\tilde{f} = f-\Expect{}{f}$ instead of $f$), we get that for all $r\leq d-1$
    and order $r$ derivatives $\mathrm{D}$ we have
    \begin{align*}
    \norm{\mathrm{D} h}_{2}
    =
    \norm{\mathrm{D}\mathrm{D}_1\left(\tilde{f}^{\le d}\right)}_{2}
    &\leq
     2^{O\left(d^{4}\right)}
     \max_{\substack{\mathrm{D}'\text{ an }r-\text{derivative}\\\mathrm{D}_1'\text{ a }1-\text{derivative}}}
     \norm{\left(\mathrm{D}'\mathrm{D}_1'\tilde{f} \right)^{\le d-r-1}}_{2}+n^{-10d}\norm{\tilde{f}^{\le 4d}}_{2}\\
    &\leq 2^{-C\cdot d^4/2} + n^{-10d}\sqrt{{\sf var}(f)} \defeq \delta,
    \end{align*}
    where we used $\mathrm{D}'\mathrm{D}_1'\tilde{f} = \mathrm{D}'\mathrm{D}_1'f$, which by assumption has $2$-norm at most $2^{-C\cdot d^4}$,
    and $\norm{\tilde{f}^{\le 4d}}_{2}\leq \norm{\tilde{f}}_2 = \sqrt{{\sf var}(f)}$.
    For $r\geq d$, we have by Claim~\ref{claim:degree reduction} that $\norm{\mathrm{D} h}_{2}=0$. Thus, all derivatives of $h$
    have small $2$-norm, and by Claim~\ref{claim:globalness equivalence} we get that $h$ is $(2d,2^{d}\delta)$-global. Thus,
    from Theorem~\ref{thm:Reasonability} we have that
    \begin{equation}\label{eq:apx2}
      \norm{\mathrm{D}_1 g}_4\leq 2^{O(d^3)}\delta^{1/2}\norm{\mathrm{D}_1 g}_2^{1/2}
    \end{equation}

    Plugging inequality~\eqref{eq:apx2} into~\eqref{eq:apx1} yields that
    \[
    \norm{\mathrm{L}_1 g}_4^4
    \leq 2^{O(d^3)}\delta^{2}
    \Expect{\substack{\mathrm{D}_1\\\text{order $1$ derivative}\\ \text{consistent with $\mathrm{L}_1$}}}
    {\norm{\mathrm{D}_1 g}_2^2}
    =2^{O(d^3)}\delta^{2}\norm{\mathrm{L}_1 g}_2^2
    \leq 2^{O(d^3)}\delta^{2}\norm{\mathrm{L}_1 f}_2^2.
    \]
    Plugging this, and the bound we have on the $4/3$-norm $\mathrm{L}_1 f$, into~\eqref{eq:apx3}, we get that
    \[
    \norm{\mathrm{L}_1 g}_2^2
    \leq
    2^{O(d^3)}\delta^{1/2}
    \norm{\mathrm{L}_1 f}_2^2.
    \]
    Summing this inequality over all $1$-Laplacians and using Fact~\ref{fact:poincare}, we get that
    \[
    {\sf var}(g)
    \leq \frac{1}{n}\sum\limits_{\mathrm{L}_1}\norm{\mathrm{L}_1 g}_2^2
    \leq 2^{O(d^3)}\delta^{1/2}\frac{2}{n-1}\sum\limits_{\mathrm{L}_1}\norm{\mathrm{L}_1 f}_2^2
    = 2^{C\cdot d^3}\delta^{1/2}I[f]
    \]
    for some absolute constant $C$, and we consider two cases.
    \paragraph{The case that $I[f]\leq 2^{-C\cdot d^3}\delta^{-1/2}{\sf var}(f)/2$.}
    In this case we get that ${\sf var}(g)\leq {\sf var}(f)/2$, and so
    $\norm{f^{>d}} = {\sf var}(f) - {\sf var}(g)\geq {\sf var}(f)/2$. Plugging
    this into~\eqref{eq:apx4} finishes the proof.

    \paragraph{The case that $I[f]\geq 2^{-C\cdot d^3}\delta^{-1/2}{\sf var}(f)/2$.}
    By definition of $\delta$ we get that either $I[f]\geq 2^{C\cdot d^4/4} {\sf var}(f)$, in which case we are done,
    or $I[f]\geq 2^{-O(d^3)} n^{5d}{\sf var}(f)^{3/4}$, in which case we are done by the lower bound on $n$.
\end{proof}
\end{document}